\documentclass[11pt]{article}
\usepackage[utf8]{inputenc}
\usepackage{cmap}
\usepackage[T1]{fontenc}
\usepackage{lmodern}
\usepackage[margin=1in]{geometry}
\usepackage{microtype}

\usepackage{amsmath,amssymb,amsthm,mathtools,dsfont}
\usepackage{graphicx}
\usepackage{tikz}
\usetikzlibrary{arrows.meta,calc,fit,positioning}
\usepackage{float}
\usepackage{booktabs,array}
\usepackage[titles]{tocloft}
\usepackage{xcolor}
\definecolor{revisionblue}{RGB}{25,75,145}

\usepackage{aliascnt}
\usepackage{xurl}
\usepackage{hyperref}
\usepackage{authblk}
\usepackage[numbers,sort&compress]{natbib}

\hypersetup{
    colorlinks,
    linkcolor={blue!55!black},
    citecolor={blue!55!black},
    urlcolor={blue!55!black},
    pdftitle={Learning Local Fermionic Lindbladians under Parity Superselection},
    pdfauthor={Tim M\"{o}bus, Daniel Stilck Fran\c{c}a, Cambyse Rouz\'e}
}

\theoremstyle{plain}
\newtheorem{theorem}{Theorem}[section]
\newaliascnt{lemma}{theorem}
\newtheorem{lemma}[lemma]{Lemma}
\aliascntresetthe{lemma}
\newaliascnt{proposition}{theorem}
\newtheorem{proposition}[proposition]{Proposition}
\aliascntresetthe{proposition}
\newaliascnt{corollary}{theorem}
\newtheorem{corollary}[corollary]{Corollary}
\aliascntresetthe{corollary}
\theoremstyle{definition}
\newaliascnt{definition}{theorem}
\newtheorem{definition}[definition]{Definition}
\aliascntresetthe{definition}
\newaliascnt{assumption}{theorem}
\newtheorem{assumption}[assumption]{Assumption}
\aliascntresetthe{assumption}
\newaliascnt{example}{theorem}
\newtheorem{example}[example]{Example}
\aliascntresetthe{example}
\theoremstyle{remark}
\newaliascnt{remark}{theorem}
\newtheorem{remark}[remark]{Remark}
\aliascntresetthe{remark}

\DeclareMathOperator{\tr}{tr}
\DeclareMathOperator{\supp}{supp}
\newcommand{\fsupp}{\supp_{\rm f}}
\DeclareMathOperator{\dist}{dist}
\DeclareMathOperator{\diam}{diam}

\DeclareMathOperator{\RePart}{Re}
\DeclareMathOperator{\ImPart}{Im}

\newcommand{\cE}{\mathcal E}
\newcommand{\cI}{\mathcal I}
\newcommand{\cK}{\mathcal K}
\newcommand{\cL}{\mathcal L}
\newcommand{\cM}{\mathcal M}
\newcommand{\cQ}{\mathcal Q}
\newcommand{\cR}{\mathcal R}
\newcommand{\cO}{\mathcal O}
\newcommand{\eps}{\varepsilon}
\newcommand{\Id}{\mathds 1}
\newcommand{\norm}[1]{\left\lVert #1\right\rVert}
\newcommand{\abs}[1]{\left\lvert #1\right\rvert}

\title{\texorpdfstring{Learning Local Fermionic Lindbladians \\ under Parity Superselection}{Learning Local Fermionic Lindbladians under Parity Superselection}}
\author[1,2]{Tim M\"{o}bus\thanks{moebustim@gmail.com}}
\author[3]{Daniel Stilck Fran\c{c}a\thanks{dsfranca@math.ku.dk}}
\author[4]{Cambyse Rouz\'e\thanks{cambyse.rouze@inria.fr}}
\affil[1]{\small Department of Applied Mathematics and Theoretical Physics, University of Cambridge, Cambridge CB3 0WA, United Kingdom}
\affil[2]{\small Department of Mathematics, University of T\"ubingen, 72076 T\"ubingen, Germany}
\affil[3]{\small Department of Mathematical Sciences, University of Copenhagen, 2100 Copenhagen, Denmark}
\affil[4]{\small Inria, T\'el\'ecom Paris--LTCI, Institut Polytechnique de Paris, 91120 Palaiseau, France}
\date{}

\begin{document}

\maketitle

\vspace*{-6ex}

\begin{abstract}
    Parity superselection forbids direct measurement of odd Majorana observables. We learn time-independent, parity-covariant, $k$-mode-local Lindbladians on $m$ modes using even preparations and measurements with uninterrupted short-time evolution. Internal markers turn odd probes into even observables, and pair measurements among three separated markers calibrate their dynamical contributions. Signed Fierz inversion recovers canonical coefficients, while semidefinite fitting yields a valid generator. For finite-range models on known bounded-degree graphs with suitable marker access and a supplied weighted-strength bound $\bar\alpha$, entrywise error $\varepsilon$ is achieved using $\widetilde{\mathcal{O}}(\bar\alpha^2\varepsilon^{-2}\log(m/\delta))$ samples, without external modes or known nonzero coefficient locations. Recovery to diamond-norm error $\varepsilon$ costs an additional factor $m^2$, matching lower bounds for short-time experiments on fresh systems up to logarithmic and fixed geometric factors. Without a supplied graph, one idle ancillary mode per system mode and potentially nonlocal pair operations give total absolute coefficient error at most $\varepsilon$ per mode using $\widetilde{\mathcal{O}}_k(\bar\alpha^2\mathsf d^2m^{\lfloor k/2\rfloor}\varepsilon^{-2}\log(m/\delta))$ samples, under a supplied approximate coefficient-degree bound $\mathsf d$ and controlled weak-coefficient tails. All guarantees hold with probability at least $1-\delta$. For geometric models, fixed-support even observables can be predicted with logarithmic system-size sample complexity at fixed time and accuracy. We also give finite-volume and exponential-tail extensions.
\end{abstract}

\setcounter{tocdepth}{1}
\vspace*{-1ex}
\setlength{\cftbeforesecskip}{2pt}
\setlength{\cftbeforesubsecskip}{1pt}
\tableofcontents
\vspace*{\fill}~

\thispagestyle{empty}

\newpage

\section{Introduction}

\subsection{Motivation and related work}
    Learning dissipative fermionic dynamics requires protocols that respect both fermionic locality and parity superselection. Without an external parity reference, physical states and measurement effects commute with total fermion parity and are called even. Odd probes anticommute with parity and cannot be measured directly. Moreover, an operator supported on a few fermionic modes can become a long qubit string under a Jordan--Wigner mapping \cite{SzalayEtAl2021}. The Fermi--Hubbard model is central to strongly correlated matter \cite{Hubbard1963,ArovasEtAl2022} and ultracold-atom simulation \cite{Esslinger2010,TarruellSanchezPalencia2018,HartEtAl2015}. These platforms also allow studies of particle loss and engineered dissipation \cite{SponseleeEtAl2018}.

    Our running example is the dissipative Fermi--Hubbard model with Hamiltonian
    \begin{equation*}
        H_{\rm FH}
        =-\sum_{\langle x,y\rangle,\sigma}
        \left(J_{xy,\sigma}a_{x,\sigma}^\dagger a_{y,\sigma}+\mathrm{h.c.}\right)
        +\sum_x U_xn_{x,\uparrow}n_{x,\downarrow}
        +\sum_{x,\sigma}\mu_{x,\sigma}n_{x,\sigma},
    \end{equation*}
    and dissipation defined below. In the above line, $a_{x,\sigma}$ annihilates a fermion of spin $\sigma\in\{\uparrow,\downarrow\}$ at site $x$, $n_{x,\sigma}=a_{x,\sigma}^\dagger a_{x,\sigma}$, and $U_x,\mu_{x,\sigma}\in\mathbb R$.  Each nearest-neighbor edge $\langle x,y\rangle$ is counted once. Common dissipative processes are described by the following jump operators $L_\ell=\sqrt{\gamma_\ell}F_\ell$, $\gamma_\ell\ge0$, with
    \begin{equation}
        F_\ell\in\left\{
        a_{x,\sigma},\ a_{x,\sigma}^\dagger,\
        n_{x,\sigma}-\tfrac12\Id,\
        a_{x,\downarrow}a_{x,\uparrow},\
        a_{y,\sigma}^\dagger a_{x,\sigma}
        \right\},
        \label{eq:intro-fermi-hubbard-jumps}
    \end{equation}
    describing loss, injection, dephasing, pair loss, and nearest-neighbor incoherent hopping, respectively. All are covered by Theorem~\ref{thm:geo-noaux} under its geometry, strength, and marker-access assumptions, as detailed in Example~\ref{ex:dissipative-fermi-hubbard}. The goal is to learn the canonical generator coefficients and their nonzero locations from short-time data.

    More generally, we consider Markovian quantum dynamics governed by the master equation
    \begin{equation*}
        \frac{d}{dt}\rho_t=\cL(\rho_t),
    \end{equation*}
    where $\rho_t$ denotes the density operator at time $t$ and $\cL$ is a time-independent Gorini--Kossakowski--Sudarshan--Lindblad (GKSL) generator \cite{GKS1976,Lindblad1976}. Quantum process tomography reconstructs the evolution map from known input states and measurements of the corresponding output states \cite{NielsenChuang2000,ChuangNielsen1997,Poyatos1997}. Without structural assumptions, its resource requirements grow exponentially with the number of modes, making it impractical for large systems.

    Work on open qubit dynamics spans steady-state reconstruction \cite{BaireyEtAl2020}, quench-based inference and ansatz selection in weakly dissipative systems \cite{OlsacherEtAl2025}, characterization of digital simulations \cite{PastoriEtAl2022}, experimental characterization \cite{SamachEtAl2022,BirkeEtAl2026,SeverinEtAl2026}, and numerical inference \cite{HeightmanEtAl2026}. Recent theoretical advances address sparse and local generators under different assumptions on structure, experimental access, and identifiability \cite{IvashkovEtAl2026,AradEtAl2026,LewisTangWright2026,MobusEtAl2026,ZhouGong2026,Sinha2026,ChenYu2026}, as well as coefficient learning within a known local generator dictionary \cite{SongEtAl2026} and learning generators with structured time dependence \cite{FrancaMobusRouzeWerner2025}. Gauge-invariant Lindbladian components can also be learned under state-preparation and measurement noise using trusted single-qubit controls \cite{FlammiaSinhaTong2026}. Beyond Markovian dynamics, protocols estimate initial derivatives of stationary Gaussian noise kernels with logarithmic system-size sample complexity under locality and regularity assumptions \cite{MontanaLopezEtAl2025}. Their general qubit construction inserts a layer of single-qubit Clifford gates during the evolution.

    Precision scaling and experimental control provide a complementary distinction. Short-time protocols can attain standard-quantum-limit sample scaling under suitable structural assumptions \cite{StilckFranca2024,AradEtAl2026,LewisTangWright2026,MobusEtAl2026,SongEtAl2026,FrancaMobusRouzeWerner2025,MobusBluhmCaroWernerRouze2023}. Heisenberg-limited total evolution time has been established for Hamiltonian learning with additional control or engineered dissipation in qubit, bosonic, and fermionic systems \cite{HuangTongFangSu2023,LiTongNiGefenYing2024,MobusBluhmGefenTongWernerRouze2025,MiraniHayden2024}. Quantum error correction also enables Lindblad learning and, under additional assumptions, Heisenberg scaling for certain Hamiltonian components \cite{RomanovEtAl2026}.

    For fermions, state and observable estimation \cite{ZhaoRubinMiyake2021,WanEtAl2023,OGorman2022,Low2024,DenzlerEtAl2024,MeleHerasymenko2025,WestCerezoLarocca2026,McNultyCalegariOszmaniec2024,MajsakMcNultyOszmaniec2025} complements dynamical learning of Hamiltonians, Gaussian and weakly non-Gaussian transformations, and noise \cite{NiLiYing2024,BluhmLemmMobusSiebert2026Coulomb,CudbyStrelchuk2024,ChristensenZhao2026,Iyer2025,AustinMoralesGorshkov2025,HelsenEtAl2022,ChapmanFlammia2026,ZhangHuoLi2021}. Recent advances include optimal Gaussian-state tomography \cite{ChenEtAl2026GaussianTomography}, separations between collective and single-copy measurements \cite{Rubin2026GaussianTomography}, and improved pure-state certification \cite{DangniamEtAl2026Certification}. Short-time potential learning is already available for free fermions \cite{BluhmLemmMobusSiebert2026Coulomb}, and base protocols for fermionic Gaussian-unitary tomography attain inverse-square precision scaling in query complexity \cite{ChristensenZhao2026}. Gaussian probes and connected-correlation responses have also been proposed for short-time identification of nonquadratic Hamiltonian components \cite{Sasaki2026GaussianProbes}. Non-Markovian kernel learning has also been demonstrated numerically for free fermions using even bilinear preparations and measurements \cite{MontanaLopezEtAl2025}. To our knowledge, prior work has not established a general scalable scheme for interacting local fermionic Hamiltonians or Lindbladians that combines local even preparations and measurements, uninterrupted short-time evolution without interleaved control, and standard-quantum-limit sample scaling. We establish this combination under the stated geometry, strength, and marker-access assumptions.

    The present work builds on our local Lindbladian reconstruction framework \cite{MobusEtAl2026}, short-time response estimation and learning of time-dependent local dynamics \cite{StilckFranca2024,FrancaMobusRouzeWerner2025}, and fermionic learning and measurement methods \cite{NiLiYing2024,MiraniHayden2024,ZhaoRubinMiyake2021,WanEtAl2023,ChristensenZhao2026}. In particular, the parity-compatible use of all three pairings within a triplet is adapted from our continuum free-fermion potential-learning protocol \cite{BluhmLemmMobusSiebert2026Coulomb}. Here the pair responses instead calibrate evolving internal markers for an interacting Lindbladian. The remaining challenge is to recover information carried by odd probes while preserving fermionic locality, which requires adapting the local inversion to fermionic anticommutation.

\subsection{Parity-superselected physical models}

    Separate parity-preserving measurements on disjoint subsystems do not, in general, suffice to reconstruct their joint state \cite{DArianoEtAl2014,BartlettRudolphSpekkens2007,VidalEtAl2021,SzalayEtAl2021}. Full fermionic process tomography likewise generally needs an ancillary mode \cite{ZhangHuoLi2021}. We study parity-covariant generators: the Hamiltonian is even and the jumps may be chosen even or odd \cite{VidalEtAl2021}. Such generators arise from parity-invariant system-bath models with an even bath state, under the assumptions leading to a Markovian master equation \cite{breuer2002theory}. The class allows odd loss or injection jumps and includes density dephasing, pair loss, incoherent hopping, dissipative Fermi--Hubbard models \cite{SponseleeEtAl2018}, engineered pairing \cite{DiehlEtAl2010}, and reservoirs for topological paired phases \cite{BardynEtAl2013,IeminiEtAl2016}. Parity covariance allows odd jumps and does not require conservation of system parity.

\section{Main results}\label{sec:main-results}
\paragraph{Common assumptions and experimental access.}
    The unknown time-independent generator acts on $m$ fermionic modes and admits a sum of parity-covariant physical GKSL terms, each supported on at most $k$ modes, with a supplied weighted interaction-strength bound $\bar\alpha$. We learn its Hamiltonian and Kossakowski coefficients in a fixed Majorana basis, using a traceless Hamiltonian and identity-free dissipative representation. Each independent sample uses a fresh full system, even preparations and measurements, and one uninterrupted evolution for a preassigned time at most $(4k\bar\alpha)^{-1}$. Experiments are nonadaptive, and disjoint local blocks may be prepared and measured in parallel within one sample. The learning guarantees hold with probability at least $1-\delta$, where $\delta\in(0,1)$. Below, strength and access parameters are fixed, and $\widetilde{\cO}$ hides polylogarithmic factors.
\paragraph{Geometric learning without external modes.}
    Assume a known bounded-degree mode graph, finite interaction range, suitable nearby internal markers, and parity-preserving pair operations within each target and marker region. The graph fixes the allowed supports, but neither nonzero coefficient locations nor a local decomposition is supplied. At fixed locality and geometry, Theorem~\ref{thm:geo-noaux} estimates every canonical coefficient to entrywise error $\eps$ using
    \begin{equation*}
        N=\widetilde{\cO}\!\left(\eps^{-2}\log\frac m\delta\right)
    \end{equation*}
    samples without auxiliary modes. Returning a valid generator with diamond-norm error $\eps$ costs an additional factor $m^2$. Corollary~\ref{cor:geometric-diamond-lower} matches this dependence up to logarithmic factors at constant failure probability, in its stated accuracy regime and at the same maximum evolution time. Thresholding identifies significant coefficients, with exact support recovery under a coefficient gap.
\paragraph{Structure learning without a supplied graph.}
    Alternatively, allow one idle ancillary mode per system mode and potentially nonlocal pair operations. Assume that, after discarding a total coefficient magnitude of at most a sufficiently small $k$-dependent multiple of $\eps$ per mode, at most $\mathsf d$ coefficients involve any one mode, with $\mathsf d$ supplied. Theorem~\ref{thm:ancilla-structure} uses
    \begin{equation*}
        N=\widetilde{\cO}_k\!\left(
            \mathsf d^2m^{\lfloor k/2\rfloor}\eps^{-2}\log\frac m\delta
        \right)
    \end{equation*}
    samples to return coefficients with local error at most $\eps$ and at most $\cO_k(\mathsf d)$ nonzero entries involving each mode. Here the so-called local error is the largest, over modes, sum of absolute errors in coefficients involving that mode. Neither a graph nor nonzero coefficient locations are supplied. The estimate need not generate physical dynamics. A subsequent convex fit returns a valid generator with comparable local error, up to numerical tolerance, but may increase its coefficient degree.

\paragraph{Consequences and extensions.}
    For geometric models, Corollaries~\ref{cor:simultaneous-local-observables} and~\ref{cor:finite-volume-local-prediction} give simultaneous prediction of fixed-support even observables with logarithmic system-size dependence, or reconstruction near a prescribed region with system-size-independent local resources, at fixed time, accuracy, and model parameters. In the latter case, experiments still evolve the full system, while subsequent classical simulation uses only the reconstructed neighborhood. On sufficiently large periodic lattices with the required marker access, exponential tails with supplied bounds add only polylogarithmic overhead for coefficient and diamond-norm learning (Corollary~\ref{cor:exponential-lattice}). Example~\ref{ex:dissipative-fermi-hubbard} treats dissipative Fermi--Hubbard models.

    A marker Majorana on a disjoint mode turns an odd monomial $\Gamma_A$ into the even observable $i\Gamma_A\eta$. For internal markers, the pair responses of three separated modes carrying $\eta_1,\eta_2,\eta_3$ determine the needed corrections; signed Fierz inversion \cite{fierz1937relativistische} and a local recursion then recover the canonical coefficients from the corrected short-time derivatives (see Fig.~\ref{fig:marker-schematic}).
    \vspace*{-1ex}
    \begin{figure}[H]
        \centering
        \includegraphics[width=\linewidth]{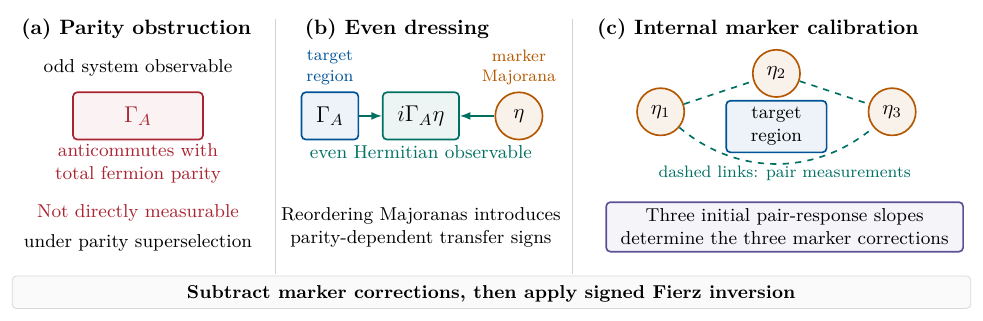}
        \vspace*{-5ex}
        \caption{Parity dressing and internal-marker calibration. (a) An odd Majorana monomial $\Gamma_A$ is not directly measurable. (b) A disjoint marker makes $i\Gamma_A\eta$ even. (c) Three nearby internal markers, beyond the target and each other's interaction range, calibrate even pair responses.}
        \label{fig:marker-schematic}
    \end{figure}

\section{Preliminaries}\label{sec:model}

\subsection{Majorana basis and parity}

    We work with $m$ fermionic modes and fix the order $\gamma_1<\cdots<\gamma_{2m}$ of their Majorana operators, which satisfy the canonical anticommutation relations (CAR),
    \begin{equation}
        \gamma_1,\ldots,\gamma_{2m},
        \qquad
        \gamma_a^\dagger=\gamma_a,
        \qquad
        \{\gamma_a,\gamma_b\}=2\delta_{ab}\Id .
    \end{equation}
    Equivalently, we may use fermionic creation and annihilation operators, defined by
    \begin{equation}
        a_j:=\frac{\gamma_{2j-1}+i\gamma_{2j}}2,
        \qquad
        n_j:=a_j^\dagger a_j
        =
        \frac{\Id+i\gamma_{2j-1}\gamma_{2j}}2 .
    \end{equation}
    The fermionic Fock space is the Hilbert space spanned by occupation states $|\nu_1,\ldots,\nu_m\rangle$, with $\nu_j\in\{0,1\}$ indicating whether mode $j$ is occupied. Its dimension is therefore $d=2^m$, including both parity subspaces. Here $a_j^\dagger$ creates a fermion in an empty mode $j$, while $a_j$ removes one from an occupied mode, and both give zero otherwise. We write $\tr$ for the ordinary, unnormalized trace on this space, so $\tr[\Id]=d$. For a subset $A=\{\alpha_1<\cdots<\alpha_{|A|}\}\subseteq[2m]$, set
    \begin{equation}
        \Gamma_A
        :=
        i^{|A|(|A|-1)/2}\gamma_{\alpha_1}\cdots\gamma_{\alpha_{|A|}},
        \qquad
        \Gamma_\varnothing=\Id .
    \end{equation}
    The phase is chosen so that every basis element is a Hermitian unitary: $\Gamma_A=\Gamma_A^\dagger$ and $\Gamma_A^2=\Id$. Moreover,
    \begin{equation}
        d^{-1}\tr[\Gamma_A\Gamma_B]=\delta_{A,B}.
    \end{equation}
    Write
    \begin{equation}
        \pi(A):=|A|\bmod 2.
    \end{equation}
    With $\mathsf N:=\sum_{j=1}^m n_j$, the particle-number parity operator satisfies
    \begin{equation}
        \Pi=(-1)^{\mathsf N}=(-i)^m\gamma_1\cdots\gamma_{2m},
        \qquad
        \Pi\Gamma_A\Pi=(-1)^{\pi(A)}\Gamma_A .
    \end{equation}
    The physical algebra without a parity reference is the even algebra, consisting of operators $X$ with $\Pi X\Pi=X$. A linear map $\mathcal T$ is parity covariant if $\mathcal T(\Pi X\Pi)=\Pi\mathcal T(X)\Pi$ for every operator $X$, or equivalently, it preserves the even and odd operator subspaces. For a mode set $e\subseteq[m]$, a Majorana label set $A\subseteq[2m]$, and $j\in[m]$, define the associated labels and local Majorana parity by
    \begin{equation}
        \cM(e):=\{2j-1,2j:j\in e\},
        \qquad
        q_j(A):=
        |A\cap\{2j-1,2j\}|\bmod2.
        \label{eq:local-majorana-parity}
    \end{equation}
    The fermionic support is
    \begin{equation}
        \supp_{\rm f}(A):=
        \{j:A\cap\{2j-1,2j\}\neq\varnothing\}.
    \end{equation}
    Thus $\pi(A)=\sum_{j=1}^m q_j(A)\bmod2$. We use the same notation with two inputs for the combined fermionic support of $A$ and $B$:
    \begin{equation}
        \fsupp(A,B):=\supp_{\rm f}(A)\cup\supp_{\rm f}(B).
    \end{equation}
    For a mode set $e\subseteq[m]$, let
    \begin{equation*}
        \mathfrak A_e
        :=
        \operatorname{span}\{\Gamma_A:A\subseteq\cM(e)\}
    \end{equation*}
    be its local CAR algebra. An operator is supported on $e$ if it belongs to $\mathfrak A_e$. Fermionic locality refers to this mode support, which can differ from the qubit support of the same operator in a Jordan--Wigner representation.

\subsection{Local parity-covariant GKSL cone}

    Fix an integer locality parameter $k\ge1$.  We parameterize dissipation by local positive semidefinite (PSD) Kossakowski blocks.  For a region $e\subseteq[m]$ with $|e|\le k$, define
    \begin{equation}
        \cI_{e,p}:=
        \{A\subseteq\cM(e):A\neq\varnothing,\ \pi(A)=p\},
        \qquad p\in\{0,1\}.
    \end{equation}
    \begin{definition}[Local parity-covariant GKSL cone]\label{def:localcone}
        A generator is a $k$-local fermionic Lindbladian, written $\cL\in\mathfrak L_k$, if it has a decomposition
        \begin{equation}
            \cL=\sum_{e:|e|\le k}\cL_e
        \end{equation}
        with
        \begin{align}
            \cL_e(\rho)
            =
            -i[H_e,\rho]
            +
            \sum_{p=0}^1
            \sum_{A,B\in\cI_{e,p}}
            K^{(e,p)}_{A,B}
            \left(
                \Gamma_A\rho\Gamma_B
                -
                \frac12\{\Gamma_B\Gamma_A,\rho\}
            \right),\label{eq:local-gksl}
        \end{align}
        where $H_e$ is even, Hermitian, traceless, and supported on $e$, and each $K^{(e,p)}$ is PSD.
    \end{definition}

    Equivalently, each local term has parity-homogeneous jump operators in the identity-free GKSL gauge: their identity components are removed, with the induced Hermitian Hamiltonian corrections absorbed into $H_e$ \cite{GKS1976,Lindblad1976,breuer2002theory}. Define the global Hamiltonian and Kossakowski matrix by
    \begin{equation}
        H:=\sum_e H_e=
        \sum_{\substack{A\neq\varnothing\\ \pi(A)=0\\ |\supp_{\rm f}(A)|\le k}}
        h_A\Gamma_A,
        \qquad h_A\in\mathbb R,
    \end{equation}
    and
    \begin{equation}
    G_{A,B}
        :=
        \sum_{\substack{e\supseteq \fsupp(A,B)\\ |e|\le k}}
        K^{(e,\pi(A))}_{A,B},
        \qquad A,B\neq\varnothing,\quad \pi(A)=\pi(B).
        \label{eq:aggregate-G}
    \end{equation}
    Local blocks are extended by zero outside their index sets, and $G_{A,B}:=0$ when $\pi(A)\neq\pi(B)$. We write $\cL_{H,G}$ for the resulting generator. Locality implies $G_{A,B}=0$ when $|\fsupp(A,B)|>k$. Membership in $\mathfrak L_k$ additionally requires the local PSD decomposition in \eqref{eq:aggregate-G}.

    Every linear map has a unique expansion in the Majorana-superoperator basis,
    \begin{equation}
        \cL
        =
        \sum_{A,B\subseteq[2m]}
        \chi_{A,B}\,\Gamma_A(\cdot)\Gamma_B .
        \label{eq:chi-expansion}
    \end{equation}
    Parity covariance is exactly the block condition
    \begin{equation}
        \chi_{A,B}=0
        \qquad\text{unless}\qquad
        \pi(A)=\pi(B).
        \label{eq:chi-parity}
    \end{equation}
    For $k$-local generators, we use the nonempty coefficient index set
    \begin{equation}
        \cI_k^\chi
        :=
        \{(A,B):\pi(A)=\pi(B),\ 1\le|\fsupp(A,B)|\le k\}.
        \label{eq:punctured-chi-index-set}
    \end{equation}

    Comparing with \eqref{eq:local-gksl} gives
    \begin{align}
        G_{A,B}&=\chi_{A,B},
        && A,B\neq\varnothing,\quad \pi(A)=\pi(B),
        \label{eq:G-from-chi}\\
        h_A&=\frac{i}{2}
        \left(\chi_{A,\varnothing}-\chi_{\varnothing,A}\right),
        && A\neq\varnothing,\quad \pi(A)=0 .
        \label{eq:h-from-chi}
    \end{align}
    The dissipative anticommutator contributes equally to the two identity-index coefficients, so it cancels in the skew difference.

    These identities uniquely determine $h$ and $G$ from $\cL$. However, a decomposition $\cL=\sum_e\cL_e$ into overlapping local physical generators in $\mathfrak L_k$ need not be unique. We call any such decomposition a \emph{witness} for membership in $\mathfrak L_k$ and fix one for the analysis. All references to local terms and their supports are relative to this fixed decomposition unless stated otherwise.

    For operators, use the unnormalized Hilbert--Schmidt norm
    \begin{equation*}
        \norm{X}_2:=\bigl(\tr[X^\dagger X]\bigr)^{1/2},
    \end{equation*}
    and for a linear map $\mathcal T$ set
    \begin{equation*}
        \norm{\mathcal T}_{2\to2}
        :=
        \sup_{X\neq0}\frac{\norm{\mathcal T(X)}_2}{\norm{X}_2}.
    \end{equation*}
    We use $\dagger$ for operator and matrix adjoints and for the Hilbert--Schmidt adjoint of a linear map. The weighted interaction strength of the fixed decomposition is
    \begin{equation}
        \alpha(\{\cL_e\})
        :=
        \max_{u\in[m]}
        \sum_{e\ni u}
        \norm{\cL_e^\dagger}_{2\to2}.
        \label{eq:weighted-strength}
    \end{equation}
    The fixed decomposition and its coefficient locations are not supplied to the learner. Every learning theorem that uses a short-time grid instead assumes that the learner receives an upper bound $\bar\alpha>0$ satisfying
    \begin{equation*}
        \alpha(\{\cL_e\})\le\bar\alpha .
    \end{equation*}
    This numerical input fixes the time window and makes the physical-time dependence explicit.
    For a local parity-covariant map, the $2\to2$ norm is evaluated on its support algebra. Its fermionic extension has the same induced Hilbert--Schmidt norm by Lemma~\ref{lem:graded-diamond-amplification}.

    \begin{lemma}[Canonical local coefficient bound]\label{lem:local-coefficient-norm}
        Let $1\le j\le k$ and write a parity-covariant $j$-mode generator in the local form \eqref{eq:local-gksl}, with Hamiltonian coefficients $h_A$, Kossakowski blocks $K^{(p)}$, and Majorana-superoperator coefficients $\chi_{A,B}$. Then
        \begin{equation*}
            \sum_A|h_A|
            +
            \sum_{p=0}^1\sum_{A,B}|K^{(p)}_{A,B}|
            +
            \sum_{A,B}|\chi_{A,B}|
            \le
            c_k\|\cL^\dagger\|_{2\to2},
            \qquad c_k=2\cdot4^k.
        \end{equation*}
        The norm is taken on the full local operator algebra, including its even and odd subspaces, and $c_k$ is independent of the ambient system size.
    \end{lemma}

    \begin{proof}
        Set $d_j=2^j$. For a single local term, \eqref{eq:G-from-chi}, \eqref{eq:h-from-chi}, and the triangle inequality give
        \begin{equation*}
            \sum_A|h_A|+\sum_{p=0}^1\sum_{A,B}|K^{(p)}_{A,B}|
            \le \sum_{A,B}|\chi_{A,B}|.
        \end{equation*}
        Let $M$ be the $d_j^2\times d_j^2$ matrix of $\cL$ under column vectorization. Since vectorization is an isometry for the Hilbert--Schmidt norm, $\|M\|_{\rm op}=\|\cL\|_{2\to2}$. The basis matrices $\Gamma_B^{\mathsf T}\otimes\Gamma_A$ are orthogonal and have Frobenius norm $d_j$, so
        \begin{equation*}
            \|M\|_{\rm F}^2=d_j^2\sum_{A,B}|\chi_{A,B}|^2.
        \end{equation*}
        If $s_\ell$ are the singular values of $M$, then
        \begin{equation*}
            \|M\|_{\rm F}^2
            =\sum_\ell s_\ell^2
            \le d_j^2\max_\ell s_\ell^2
            =d_j^2\|\cL\|_{2\to2}^2.
        \end{equation*}
        Hence $\|M\|_{\rm F}\le d_j\|\cL\|_{2\to2}$, and Cauchy--Schwarz over the $d_j^4$ coefficient pairs gives
        \begin{equation*}
            \sum_{A,B}|\chi_{A,B}|
            \le d_j\|M\|_{\rm F}
            \le d_j^2\|\cL\|_{2\to2}
            =d_j^2\|\cL^\dagger\|_{2\to2}.
        \end{equation*}
        The last equality uses invariance of the induced Hilbert--Schmidt norm under adjoints. The full coefficient sum is therefore at most $2d_j^2\|\cL^\dagger\|_{2\to2}$, and $d_j^2=4^j\le4^k$ concludes the proof.
    \end{proof}


\section{Parity dressing and local signed inversion}\label{sec:parity-dressing-and-local-signed-inversion}

    We first derive a signed inversion formula for even probes obtained by parity dressing. Restricting this formula to local labels yields a descending reconstruction with controlled error.

\subsection{Parity dressing}\label{subsec:parity-dressing}

Add one ancillary fermionic mode with Majoranas $\eta$ and $\eta'$, ordered after all system Majoranas. Only $\eta$ is used for parity dressing. The operator $\eta'$ completes the physical ancillary mode. The system map is extended to the enlarged CAR algebra using the same Majorana-superoperator expansion,
    \begin{equation}
        \cL_{\rm ext}(X)
        :=
        \sum_{A,B}
        \chi_{A,B}\,\Gamma_A X\Gamma_B ,
        \label{eq:graded-extension}
    \end{equation}
    We call this the \emph{fermionic extension}. Parity-homogeneous operators on disjoint system and ancillary modes obey the $\mathbb Z_2$-graded commutation rule
    \begin{equation*}
        XY=(-1)^{\pi(X)\pi(Y)}YX.
    \end{equation*}
    Thus two odd operators anticommute, whereas they commute if either operator is even.  Using this fermionic sign rule avoids the ambiguity of an ordinary tensor-product extension for odd operators \cite{SzalayEtAl2021}. An ancillary mode is called idle when it is appended only for algebraic parity dressing and carries no Hamiltonian or dissipative support of its own.

    \begin{lemma}[Norm stability under fermionic extension]\label{lem:graded-diamond-amplification}
        Let $\mathcal T$ be a parity-covariant linear map on the system CAR algebra and let $\mathcal T_{\rm ext}$ be its fermionic extension by finitely many idle fermionic modes.  Then
        \begin{equation*}
            \|\mathcal T_{\rm ext}\|_{2\to2}
            =
            \|\mathcal T\|_{2\to2},
            \qquad
            \|\mathcal T_{\rm ext}(X)\|_1
            \le
            \|\mathcal T\|_\diamond\|X\|_1
        \end{equation*}
        for every $X$ on the enlarged algebra.  In particular, for parity-covariant Lindbladians $\cL$ and $\cK$,
        \begin{equation*}
            \|e^{t\cL_{\rm ext}}-e^{t\cK_{\rm ext}}\|_{1\to1}
            \le
            t\|\cL-\cK\|_\diamond ,
            \qquad t\ge0 .
        \end{equation*}
    \end{lemma}

    \begin{proof}
        Order the system modes before the idle modes. Under the resulting identification of their Fock spaces, system Majorana monomials act as $\Gamma_A\otimes\Id$, while the fermionic grading is carried by parity strings on odd ancillary operators. Thus the fermionic extension is represented on the full matrix algebra by the ordinary amplification $\mathcal T_{\rm ext}=\mathcal T\otimes\operatorname{id}$. Consequently,
        \begin{equation*}
            \|\mathcal T_{\rm ext}\|_{2\to2}=\|\mathcal T\|_{2\to2},
            \qquad
            \|\mathcal T_{\rm ext}\|_{1\to1}\le\|\mathcal T\|_\diamond.
        \end{equation*}
        These bounds are independent of the chosen fermionic ordering because CAR-mode permutations are unitarily implemented. Finally, Duhamel's formula gives
        \begin{equation*}
            e^{t\cL_{\rm ext}}-e^{t\cK_{\rm ext}}
            =
            \int_0^t
            e^{(t-s)\cL_{\rm ext}}
            (\cL-\cK)_{\rm ext}
            e^{s\cK_{\rm ext}}\,ds .
        \end{equation*}
        The two amplified semigroups are trace-norm contractions, so integration yields the stated estimate.
    \end{proof}

    Define the dressing map
    \begin{equation}
        J_\eta(\Gamma_A)
        :=
        i^{\pi(A)}\Gamma_A\eta^{\pi(A)}
        =
        \begin{cases}
            \Gamma_A, & \pi(A)=0,\\
            i\Gamma_A\eta, & \pi(A)=1.
        \end{cases}
        \label{eq:ancilla-dressing}
    \end{equation}
    Then every $J_\eta(\Gamma_A)$ is an even Hermitian unitary, and
    \begin{equation}
        (2d)^{-1}
        \tr\!\left[
            J_\eta(\Gamma_A)J_\eta(\Gamma_B)
        \right]
        =
        \delta_{A,B}.
    \end{equation}
    The physically queried dressed transfer entries are
    \begin{equation}
        \Lambda_{C,D}
        :=
        (2d)^{-1}
        \tr\!\left[
            J_\eta(\Gamma_C)\,
            \cL_{\rm ext}(J_\eta(\Gamma_D))
        \right].
        \label{eq:dressed-transfer}
    \end{equation}

\subsection{Signed Fierz inversion}\label{subsec:signed-fierz}

    The inversion uses the following completeness identities for the Majorana basis.

\begin{lemma}[Majorana completeness]\label{lem:maj-completeness}
        For all matrix indices $p,q,r,s$,
        \begin{equation}\label{eq:majorana-completeness}
            \sum_{A\subseteq[2m]}
            (\Gamma_A)_{pq}(\Gamma_A)_{rs}
            =
            d\,\delta_{ps}\delta_{qr}.
        \end{equation}
        Consequently, for all matrices $X,Y,Z,W$ on the $d$-dimensional Fock space,
        \begin{equation}\label{eq:majorana-fierz}
            \sum_{C,D\subseteq[2m]}
            \tr[\Gamma_C X\Gamma_D Y]
            \tr[\Gamma_C Z\Gamma_D W]
            =
            d^2\tr[XW]\tr[YZ].
        \end{equation}
    \end{lemma}

    \begin{proof}
        The operators $\{\Gamma_A/\sqrt d\}_A$ form a Hermitian orthonormal Hilbert--Schmidt basis of the full Fock-space matrix algebra.  Both identities therefore follow from the completeness argument in \cite[proof of Lemma~3.2]{MobusEtAl2026}, which uses only these basis properties.
    \end{proof}

    For comparison, undressed Majorana data satisfy
    \begin{equation*}
        \chi_{A,B}
        =
        d^{-4}
        \sum_{C,D}
        \tr[\Gamma_C\cL(\Gamma_D)]
        \tr[\Gamma_C\Gamma_B\Gamma_D\Gamma_A].
    \end{equation*}
    These algebraically complete data are not physical under parity superselection when $C$ or $D$ is odd.

\begin{lemma}[Signed dressed transfer relation]\label{lem:signed-transfer}
        For a parity-covariant map and every $C,D\subseteq[2m]$,
        \begin{equation}
            \Lambda_{C,D}
            =
            d^{-1}
            \sum_{\substack{A,B\\ \pi(A)=\pi(B)}}
            (-1)^{\pi(D)\pi(B)}
            \chi_{A,B}
            \tr[\Gamma_C\Gamma_A\Gamma_D\Gamma_B].
            \label{eq:signed-transfer}
        \end{equation}
        Both sides vanish whenever $\pi(C)\neq\pi(D)$.
    \end{lemma}

    \begin{proof}
        If $\pi(C)\neq\pi(D)$, the ancillary factor contains one power of $\eta$, so $\Lambda_{C,D}=0$. The trace on the right also vanishes termwise: since $\pi(A)=\pi(B)$, the product $\Gamma_C\Gamma_A\Gamma_D\Gamma_B$ is odd. Now suppose $\pi(C)=\pi(D)=q$. In one summand,
        \begin{equation*}
            J_\eta(\Gamma_C)\Gamma_AJ_\eta(\Gamma_D)\Gamma_B
            =
            i^{2q}\Gamma_C\eta^q\Gamma_A\Gamma_D\eta^q\Gamma_B.
        \end{equation*}
        Moving the first $\eta^q$ through $\Gamma_A\Gamma_D$ contributes $(-1)^{q(\pi(A)+q)}$.  Since $i^{2q}=(-1)^q$ and $q^2=q$ modulo two, the total phase is $(-1)^{q\pi(A)}=(-1)^{\pi(D)\pi(B)}$.  The ancillary trace then contributes a factor of two, leaving the system trace in \eqref{eq:signed-transfer}.
    \end{proof}

    \begin{lemma}[Signed dressed inversion]\label{lem:signed-dressed-inversion}
        For every $A,B\subseteq[2m]$,
        \begin{equation}\label{eq:signed-dressed-inversion}
            \chi_{A,B}
            =
            d^{-3}
            \sum_{C,D\subseteq[2m]}
            (-1)^{\pi(D)\pi(B)}
            \Lambda_{C,D}\,
            \tr[\Gamma_C\Gamma_B\Gamma_D\Gamma_A].
        \end{equation}
    \end{lemma}

    \begin{proof}
        If $\pi(A)\neq\pi(B)$, then $\chi_{A,B}=0$ by parity covariance. Moreover, $\Lambda_{C,D}$ vanishes unless $\pi(C)=\pi(D)$, in which case $\Gamma_C\Gamma_B\Gamma_D\Gamma_A$ is odd and its trace vanishes. Thus both sides are zero. Assume henceforth that $\pi(A)=\pi(B)$.

        Insert \eqref{eq:signed-transfer} into the right-hand side of \eqref{eq:signed-dressed-inversion}, and denote the summation indices in \eqref{eq:signed-transfer} by $A',B'$.  The sign multiplying the two traces is
        \begin{equation*}
            (-1)^{\pi(D)\pi(B)}
            (-1)^{\pi(D)\pi(B')}
            =
            (-1)^{\pi(D)(\pi(B)+\pi(B'))}.
        \end{equation*}
        If $\pi(B')=\pi(B)$, this sign is one.  Applying \eqref{eq:majorana-fierz} to
        \begin{equation*}
            \tr[\Gamma_C\Gamma_{A'}\Gamma_D\Gamma_{B'}]
            \tr[\Gamma_C\Gamma_B\Gamma_D\Gamma_A]
        \end{equation*}
        gives $d^2\tr[\Gamma_{A'}\Gamma_A]\tr[\Gamma_{B'}\Gamma_B]$, hence only $A'=A$ and $B'=B$ survive after the prefactors.

        It remains to exclude the block $\pi(B')\neq\pi(B)$.  In that case the leftover factor is $(-1)^{\pi(D)}$, which can be written as $\Pi\Gamma_D\Pi$ inside the $D$-sum.  Applying \eqref{eq:majorana-fierz} then yields a product of traces
        \begin{equation*}
            d^2\tr[\Gamma_{A'}\Pi\Gamma_A]\,
            \tr[\Pi\Gamma_{B'}\Gamma_B],
        \end{equation*}
        but the second trace vanishes because $\Pi$ is even and $\pi(B')\neq\pi(B)$. Thus only the desired block survives.
    \end{proof}

\subsection{Restricted signed inversion}\label{subsec:restricted-inversion}

The algebraic inversion uses the dressed entries $\Lambda_{C,D}$, which are the initial derivatives of the even transfer curves \eqref{eq:operational-transfer}. See \eqref{eq:local-transfer-derivative}.
    To recover a target coefficient $\chi_{A,B}$, set $R=\fsupp(A,B)$ and use only derivatives with $C,D\subseteq\cM(R)$.  Their signed Fierz transform is the restricted array $\chi^{[R]}$ defined below.  By \eqref{eq:diagonal-extension}, $\chi^{[R]}_{A,B}$ equals $\chi_{A,B}$ plus coefficients $\chi_{A\sqcup E,B\sqcup E}$ with strictly larger mode support.  Hence $k$-locality reduces recovery to a triangular recursion from larger to smaller supports.  Section~\ref{sec:statistics} gives the physical even preparations and measurements, while Proposition~\ref{prop:fermionic-derivative-poly} gives robust estimation of \eqref{eq:local-transfer-derivative} from noisy short-time data.

    In general, a partial trace of the undressed system dynamics in a Jordan--Wigner representation does not yield the dressed restricted coefficients, because marker dressing changes the contributions of odd complement operators.  We therefore invert directly from dressed derivative entries with local labels, and then remove the diagonal extensions contributed by modes outside the chosen region.

Let $R\subseteq[m]$, $d_R=2^{|R|}$, and let $\tr_R$ denote the trace in the abstract $|R|$-mode CAR algebra, with its Majoranas ordered as in the corresponding subsequence of $\gamma_1<\cdots<\gamma_{2m}$.  For $A,B\subseteq\cM(R)$ with $\pi(A)=\pi(B)$, define
    \begin{equation}
        \chi^{[R]}_{A,B}
        :=
        d_R^{-3}
        \sum_{C,D\subseteq\cM(R)}
        (-1)^{\pi(D)\pi(B)}
        \Lambda_{C,D}\,
        \tr_R[\Gamma_C\Gamma_B\Gamma_D\Gamma_A].
        \label{eq:restricted-chi}
    \end{equation}
    We call $\chi^{[R]}$ the restricted dressed coefficient array.  It is a local marginal of the global $\chi$-array, not the coefficient matrix of an ordinary partial-trace marginal generator.

    For disjoint Majorana label sets $A,E$, define
    \begin{equation}
        \nu(A,E):=\#\{(a,e)\in A\times E:e<a\},
        \qquad
        \omega(A,E):=(-1)^{\nu(A,E)}i^{|A||E|}.
        \label{eq:omega}
    \end{equation}
    Then
    \begin{equation}
        \Gamma_{A\sqcup E}
        =
        \omega(A,E)\Gamma_A\Gamma_E.
        \label{eq:union-product}
    \end{equation}
    For $A,B\subseteq\cM(R)$ and $E\subseteq\cM(R^c)$, set
    \begin{equation}
        s_R(A,B;E)
        :=
        \omega(A,E)\omega(B,E)(-1)^{\pi(E)\pi(B)}.
        \label{eq:extension-phase}
    \end{equation}
    If $\pi(A)=\pi(B)$, then $s_R(A,B;E)\in\{\pm1\}$.

    \begin{theorem}[Restricted-data diagonal-extension identity]\label{thm:diagonal-extension}
        For every $A,B\subseteq\cM(R)$ with $\pi(A)=\pi(B)$,
        \begin{equation}
            \chi^{[R]}_{A,B}
            =
            \sum_{E\subseteq\cM(R^c)}
            s_R(A,B;E)\,
            \chi_{A\sqcup E,\,B\sqcup E}.
            \label{eq:diagonal-extension}
        \end{equation}
        Under $k$-locality, only terms with
        $|\fsupp(A,B)\cup\supp_{\rm f}(E)|\le k$ can be nonzero.
    \end{theorem}

    \begin{proof}
        Write arbitrary global labels as $A'=A_0\sqcup E$ and $B'=B_0\sqcup F$, with $A_0,B_0\subseteq\cM(R)$ and $E,F\subseteq\cM(R^c)$. Since $C,D$ are supported on $R$, the trace over $R^c$ contains $\tr_{R^c}[\Gamma_E\Gamma_F]$, which vanishes unless $E=F$ by Majorana-basis orthogonality. Using \eqref{eq:union-product}, commuting the complement monomial through the local factors, and tracing over $R^c$ gives
        \begin{align}
            d^{-1}\tr[\Gamma_C\Gamma_{A'}\Gamma_D\Gamma_{B'}]
            =
            \delta_{E,F}\,d_R^{-1}
            \omega(A_0,E)\omega(B_0,E)
            (-1)^{\pi(E)(\pi(D)+\pi(B_0))}
            \tr_R[\Gamma_C\Gamma_{A_0}\Gamma_D\Gamma_{B_0}].
            \label{eq:trace-factor}
        \end{align}
        The signed transfer relation contributes
        \begin{equation*}
            (-1)^{\pi(D)\pi(B')}
            =
            (-1)^{\pi(D)(\pi(B_0)+\pi(E))}.
        \end{equation*}
        Multiplying this with the sign in \eqref{eq:trace-factor} cancels the two copies of $\pi(D)\pi(E)$ and leaves
        \begin{equation*}
            \omega(A_0,E)\omega(B_0,E)
            (-1)^{\pi(E)\pi(B_0)}
            (-1)^{\pi(D)\pi(B_0)}.
        \end{equation*}
        The definition of $\chi^{[R]}_{A,B}$ supplies one more inversion sign,
        $(-1)^{\pi(D)\pi(B)}$.  Hence the $C,D$-dependent part is the local signed Fierz kernel with sign
        \begin{equation*}
            (-1)^{\pi(D)(\pi(B)+\pi(B_0))},
        \end{equation*}
        while the $C,D$-independent coefficient is
        \begin{equation*}
            \omega(A_0,E)\omega(B_0,E)(-1)^{\pi(E)\pi(B_0)}.
        \end{equation*}
        Applying Lemma~\ref{lem:signed-dressed-inversion} in the abstract $|R|$-mode Clifford algebra therefore leaves only $A_0=A$, $B_0=B$.  The remaining coefficient is exactly $s_R(A,B;E)$, proving \eqref{eq:diagonal-extension}.
    \end{proof}

\subsection{Descending recursion and stability}

    Choose $R=\fsupp(A,B)$.  Since $s_R(A,B;\varnothing)=1$, \eqref{eq:diagonal-extension} gives the descending recursion
    \begin{equation}
        \chi_{A,B}
        =
        \chi^{[R]}_{A,B}
        -
        \sum_{\substack{E\subseteq\cM(R^c),\ E\neq\varnothing\\
        |R\cup\supp_{\rm f}(E)|\le k}}
        s_R(A,B;E)\,
        \chi_{A\sqcup E,\,B\sqcup E}.
        \label{eq:descending-recursion}
    \end{equation}
    Process support sizes in descending order $k,k-1,\ldots,0$.  Every nontrivial extension on the right has strictly larger mode-support union and has already been recovered.

    \begin{lemma}[Triangular error propagation]\label{lem:triangular-error-propagation}
        Let $\mathcal V$ be a finite directed graph with no directed cycles and whose edges strictly increase an integer rank. Suppose that
        \begin{equation*}
            e_v\le a_v+\sum_{v\to w}e_w,
            \qquad v\in\mathcal V,
        \end{equation*}
        with $e_v,a_v\ge0$.  If $p(v,u)$ is the number of directed paths from $v$ to $u$, including the trivial path, then
        \begin{equation*}
            e_v\le\sum_{u\in\mathcal V}p(v,u)a_u.
        \end{equation*}
        In particular, if the graph has depth at most $L$, out-degree at most $d$, and $a_v\le a$, then
        \begin{equation*}
            e_v\le\left(\sum_{j=0}^{L}d^j\right)a.
        \end{equation*}
    \end{lemma}

    \begin{proof}
        Descending induction in the rank, substituting the bound at every child, indexes each occurrence of $a_u$ by one directed path from $v$ to $u$.  Counting paths by length proves the final assertion.
    \end{proof}

    \begin{lemma}[Worst-case stability]\label{lem:fermionic-stability}
        Suppose every restricted coefficient $\chi^{[R]}_{A,B}$ used by \eqref{eq:descending-recursion} is known to additive error at most $\eps_L$.  For fixed $k$, the recovered coefficients satisfy
        \begin{equation}
            \max_{\substack{\pi(A)=\pi(B)\\ |\fsupp(A,B)|\le k}}
            \abs{\widehat\chi_{A,B}-\chi_{A,B}}
            \le
            C_k m^k\eps_L
        \end{equation}
        for a constant $C_k$ depending only on $k$.  More precisely, on the punctured set of nonidentity coefficients,
        \begin{equation}
            \max_{\substack{\pi(A)=\pi(B)\\
            1\le |\fsupp(A,B)|\le k}}
            \abs{\widehat\chi_{A,B}-\chi_{A,B}}
            \le
            C_k m^{k-1}\eps_L .
            \label{eq:physical-stability}
        \end{equation}
    \end{lemma}

    \begin{proof}
        Let $e_s$ be the maximal coefficient error at mode-support size $s$.  Since every phase in \eqref{eq:descending-recursion} has modulus one and each added mode has three nonempty Majorana-label choices,
        \begin{equation*}
            e_s
            \le
            \eps_L+\sum_{t=1}^{k-s}\binom{m-s}{t}3^t e_{s+t},
            \qquad 0\le s\le k.
        \end{equation*}
        For $s\ge1$, the induction in \cite[proof of Lemma~3.3]{MobusEtAl2026} gives $e_s\le C_{k-s}m^{k-s}\eps_L$, proving \eqref{eq:physical-stability}.  Substituting these bounds into the recurrence at $s=0$ gives $e_0\le C_km^k\eps_L$ and hence the full bound.
    \end{proof}

\section{Even short-time experiments and internal-marker calibration}\label{sec:statistics}

    Building on classical-shadow and process-shadow methods \cite{HuangKuengPreskill2020,Kunjummen2023ShadowProcess,LevyLuoClark2024,StilckFranca2024}, we construct common-data estimators for dressed input--output overlaps using even pair-block experiments. A separate-query variant appears in Remark~\ref{rem:one-marker-recovery}. We first describe physical transfer data with a passive marker mode, then calibrate the additional contributions when markers belong to the evolving system.

\subsection{Experimental access and physical transfer curves}\label{subsec:experimental-access}

\begin{assumption}[Experimental access and sample accounting]\label{ass:experimental-access}
        Each sample is one channel use on a freshly prepared full system: the protocol prepares the prescribed even state, evolves it under $e^{t\cL}$ for a preassigned time $t$, and performs the even measurements.  Inactive modes are maximally mixed, and reference modes are idle.  The protocols are nonadaptive and apply no control during the evolution.  Disjoint local blocks may be prepared and measured in parallel within one sample.

        We count local readouts, total evolution time, and classical postprocessing separately.  The geometric protocol uses bounded-diameter blocks. Routing for the nonlocal operations in the all-to-all protocol is not included.  Since all upper bounds use $t_i\le T:=(4k\bar\alpha)^{-1}$, $N$ samples satisfy
        \begin{equation}
            T_{\rm evol}=\sum_{i=1}^N t_i\le TN=\frac{N}{4k\bar\alpha}.
            \label{eq:total-evolution-time}
        \end{equation}
    \end{assumption}

    \begin{remark}[Strictly number-conserving access]\label{rem:number-conserving-identifiability}
        Some parity-preserving but number-nonconserving access is necessary for full identification. Our protocol provides it through pairing rotations.  Indeed, let $\cL_0$ be $U(1)$ covariant, meaning that it is invariant under the global phase rotations generated by $\mathsf N$.  Equivalently, $\cL_0$ commutes with $-i[\mathsf N,\cdot]$, and for states, effects, and controls commuting with $\mathsf N$,
        \begin{equation*}
            \tr\!\left[
                Oe^{t(\cL_0-i\theta[\mathsf N,\cdot])}(\rho)
            \right]
            =
            \tr\!\left[Oe^{t\cL_0}(\rho)\right].
        \end{equation*}
        Thus strictly number-conserving experiments cannot distinguish generators that differ by the chemical-potential term $-i\theta[\mathsf N,\cdot]$.  Pairing rotations remove this obstruction by coupling particle-number sectors that differ by two while preserving parity.  This does not imply that fixing this ambiguity suffices to identify every $U(1)$-covariant generator.
    \end{remark}

    Let $d_{\rm ext}=2d$ be the dimension of the system plus one marker mode, and set $\Phi_t^{\rm ext}:=e^{t\cL_{\rm ext}}$.  Define the dressed transfer curve
    \begin{equation}
        M^\eta_{C,D}(t)
        :=
        d_{\rm ext}^{-1}
        \tr\!\left[
            J_\eta(\Gamma_C)\Phi_t^{\rm ext}(J_\eta(\Gamma_D))
        \right].
        \label{eq:operational-transfer}
    \end{equation}
    For $D\neq\varnothing$, the formal input $J_\eta(\Gamma_D)$ is represented by the difference of two physical even states
    \begin{equation}
        \rho_D^\pm
        :=
        \frac{\Id\pm J_\eta(\Gamma_D)}{d_{\rm ext}}.
        \label{eq:physical-inputs}
    \end{equation}
    In that case,
    \begin{equation}
        M^\eta_{C,D}(t)
        =
        \frac12\left(
            \tr[J_\eta(\Gamma_C)\Phi_t^{\rm ext}(\rho_D^+)]
            -
            \tr[J_\eta(\Gamma_C)\Phi_t^{\rm ext}(\rho_D^-)]
        \right).
        \label{eq:physical-transfer-states}
    \end{equation}
    For $D=\varnothing$, prepare the maximally mixed state $\rho_\varnothing=\Id/d_{\rm ext}$. The measured expectation is then directly $M^\eta_{C,\varnothing}(t)$, with no difference of states required.  Thus every queried preparation and effect is even, and
    \begin{equation}
        \left.\frac{d}{dt}M^\eta_{C,D}(t)\right|_{t=0}
        =
        \Lambda_{C,D}.
        \label{eq:local-transfer-derivative}
    \end{equation}
    The protocol requires preparation of the local even eigenspace states $\rho_D^\pm$ and measurement of the even observable $J_\eta(\Gamma_C)$.  Proposition~\ref{prop:even-clifford-eigenaccess} realizes these operations using local fermionic Gaussian control, randomized occupation-state preparations, and occupation measurements.  This is a targeted-query primitive, distinct from the common-data shadow primitive used below.

    \begin{remark}[Comparison with Pauli inputs]
        For a Pauli string $Q$ on $n$ qubits, the analogous states $(\Id\pm Q)/2^n$ can be prepared by randomly choosing product Pauli eigenstates whose eigenvalues have the prescribed product.  Similarly, Proposition~\ref{prop:even-clifford-eigenaccess} maps the even monomial $J_\eta(\Gamma_D)$ to a signed product of occupation parities, samples occupation strings with the required product parity, and applies the inverse Gaussian unitary.  Hence no arbitrary many-body state preparation or non-Gaussian control is needed, and for fixed locality the operation involves only a constant number of modes.  Our preparation scheme uses local control, including marker access, Gaussian mode mixing, and pairing rotations.
    \end{remark}

    \begin{assumption}[Statistical observation model]\label{ass:fermionic-statistical-model}
        Conditional on the recorded preparation and measurement settings, each outcome obeys the Born rule for the specified channel $e^{t\cL}$, and distinct samples are independent. Random settings and input signs are drawn independently as specified by the protocol (All probabilistic guarantees below are stated under this assumption).
    \end{assumption}

\subsection{Pairing-compatible process shadows}\label{sec:pairing-shadows}

    The fermionic analogue of a product Pauli shadow is a product over paired two-mode blocks.  Let $W$ contain an even number of modes and let $\mathfrak p$ pair them into blocks $b=\{u,v\}$.  For such a block define
    \begin{equation}
        \cE_b^{\rm ev}
        :=
        \{\Gamma_A:\ A\subseteq\cM(b),\ \pi(A)=0,\ A\neq\varnothing\}.
    \end{equation}
    This set has seven Hermitian unitaries: the six Majorana bilinears and the block parity.  Together with the identity they form an orthogonal basis of the even algebra on the two-mode block.

    A Majorana monomial $O$ is $\mathfrak p$-compatible if every paired block contains an even number of Majorana labels from $O$.  Equivalently,
    \begin{equation}
        O=\sigma_{\mathfrak p}(O)\prod_{b\in\mathfrak p} O_b,
        \qquad
        O_b\in \{\Id\}\cup\cE_b^{\rm ev},
        \qquad
        \sigma_{\mathfrak p}(O)\in\{\pm1\},
    \end{equation}
    where the sign $\sigma_{\mathfrak p}(O)$ is fixed by the order $\gamma_1<\cdots<\gamma_{2m}$.
    Let $w_{\mathfrak p}(O)$ count the nonidentity block factors.  The pairing $\mathfrak p$ is arbitrary, but it restricts which monomials are estimable in that shot.
    Pairing entries label modes, whereas the subscripts of $\gamma$ label Majoranas, with mode $j$ carrying $\gamma_{2j-1}$ and $\gamma_{2j}$. Figure~\ref{fig:pair-block-shadow} illustrates why multiple pairings are needed: $i\gamma_1\gamma_5$ is incompatible with $\{\{1,2\},\{3,4\}\}$, which places one Majorana in each block, but compatible with $\{\{1,3\},\{2,4\}\}$.
    \vspace{0ex}
    \begin{figure}[H]
        \centering
        \begin{tikzpicture}[
            font=\small,
            >={Latex[length=2mm]},
            flow/.style={-Latex,thick,draw=blue!55!black},
            stage/.style={draw=blue!45!black,fill=blue!4,rounded corners=2pt,
                minimum height=1.35cm,align=center,inner sep=5pt},
            mode/.style={circle,draw=black!65,minimum size=6mm,inner sep=0pt},
            active/.style={mode,fill=orange!55},
            block/.style={draw=blue!50!black,rounded corners=3pt,inner sep=3pt}
        ]
            \node[anchor=west,font=\bfseries] at (-0.9,1.3) {(a) One sample};
            \node[stage,text width=3.35cm] (input) at (1.9,0)
                {\textbf{Input blocks}\\
                 choose $(R_b,S_b)$\\
                 prepare $\rho_b$};
            \node[draw=blue!55!black,fill=blue!10,circle,minimum size=1.25cm,
                right=7mm of input] (channel) {evolve};
            \node[stage,text width=3.35cm,right=7mm of channel] (output)
                {\textbf{Output blocks}\\
                 choose $L_a$\\
                 observe $Y_a$};
            \node[right=2mm of output,text=blue!55!black] (capture)
                {\scalebox{2.2}[5.0]{$\}$}};
            \node[align=center,text width=2.85cm,right=2mm of capture] (record)
                {\textbf{Recorded data}\\ $R_b,S_b,L_a,Y_a$};
            \draw[flow] (input) -- (channel);
            \draw[flow] (channel) -- (output);

            \node[anchor=west,font=\bfseries] at (-0.9,-1.95)
                {(b) Pairing compatibility for $O=i\gamma_1\gamma_5$};

            \node[active] (i1) at (0.5,-3.1) {$1$};
            \node[mode]   (i2) at (1.3,-3.1) {$2$};
            \node[active] (i3) at (2.55,-3.1) {$3$};
            \node[mode]   (i4) at (3.35,-3.1) {$4$};
            \node[block,fit=(i1)(i2)] (ib1) {};
            \node[block,fit=(i3)(i4)] (ib2) {};
            \node[below=4mm of $(ib1.south)!0.5!(ib2.south)$] (badlabel)
                {$\mathfrak p=\{\{1,2\},\{3,4\}\}$};
            \node[at=(ib2.north east),anchor=center,rotate=-12,
                preaction={draw=white,fill=white,line width=5pt},
                draw=red!60!black,fill=white,rounded corners=1pt,
                text=red!60!black,font=\small\bfseries,inner sep=2.5pt]
                {$\times$ incorrect};

            \node[active] (c1) at (6.9,-3.1) {$1$};
            \node[active] (c3) at (7.7,-3.1) {$3$};
            \node[mode]   (c2) at (8.95,-3.1) {$2$};
            \node[mode]   (c4) at (9.75,-3.1) {$4$};
            \node[block,fit=(c1)(c3)] (cb1) {};
            \node[block,fit=(c2)(c4)] (cb2) {};
            \node[below=4mm of $(cb1.south)!0.5!(cb2.south)$] (goodlabel)
                {$\mathfrak p'=\{\{1,3\},\{2,4\}\}$};
            \node[at=(cb2.north east),anchor=center,rotate=-12,
                preaction={draw=white,fill=white,line width=5pt},
                draw=green!45!black,fill=white,rounded corners=1pt,
                text=green!40!black,font=\small\bfseries,inner sep=2.5pt]
                {$\checkmark$ correct};
        \end{tikzpicture}
        \vspace{-2ex}
        \caption{Pair-block process shadow. (a) Random even settings are chosen independently on the input and output blocks, and one sample contributes to all compatible queries. (b) Orange modes carry the factors of $O=i\gamma_1\gamma_5$.  Compatibility requires an even number of factors in each block.}
        \label{fig:pair-block-shadow}
    \end{figure}

    \begin{proposition}[Pair-block process-shadow estimator]\label{prop:pair-block-process-shadow}
        Fix an input pairing $\mathfrak p_{\rm in}$ and an output pairing $\mathfrak p_{\rm out}$ of $W$, and let $\Phi$ be a parity-covariant channel on the corresponding CAR algebra.  Let $\cR_{\rm in}$ be a finite family of $\mathfrak p_{\rm in}$-compatible Hermitian Majorana monomials and let $\cR_{\rm out}$ be a finite family of $\mathfrak p_{\rm out}$-compatible Hermitian Majorana monomials.  At one shot, independently for every input block $b\in\mathfrak p_{\rm in}$ choose an input axis $R_b\in\cE_b^{\rm ev}$ uniformly and a sign $S_b\in\{\pm1\}$ uniformly, and prepare
        \begin{equation}
            \rho_{b,R_b,S_b}
            =
            \frac{\Id+S_bR_b}{4}.
        \end{equation}
        Evolve under $\Phi$.  Independently for every output block $a\in\mathfrak p_{\rm out}$ choose an output axis $L_a\in\cE_a^{\rm ev}$ uniformly and measure $L_a$, obtaining $Y_a\in\{\pm1\}$.  Then for every $P\in\cR_{\rm out}$ and $Q\in\cR_{\rm in}$ there is a single-shot random variable $Z_{P,Q}$, defined in \eqref{eq:pair-block-estimator}, such that
        \begin{equation}
            \mathbb E Z_{P,Q}
            =
            2^{-|W|}
            \tr[P\Phi(Q)]
        \end{equation}
        and
        \begin{equation}
            \mathbb E Z_{P,Q}^2
            \le
            7^{w_{\mathfrak p_{\rm out}}(P)+w_{\mathfrak p_{\rm in}}(Q)}.
            \label{eq:pair-block-second-moment}
        \end{equation}
    \end{proposition}

    \begin{proof}
        Write $P=\sigma_{\mathfrak p_{\rm out}}(P)\prod_{a\in\mathfrak p_{\rm out}}P_a$ and $Q=\sigma_{\mathfrak p_{\rm in}}(Q)\prod_{b\in\mathfrak p_{\rm in}}Q_b$ with $P_a\in\{\Id\}\cup\cE_a^{\rm ev}$ and $Q_b\in\{\Id\}\cup\cE_b^{\rm ev}$.  Let $H_{P,Q}$ be the event that $L_a=P_a$ for every output block with $P_a\neq\Id$ and $R_b=Q_b$ for every input block with $Q_b\neq\Id$.  Define
        \begin{equation}
            Z_{P,Q}
            :=
            \sigma_{\mathfrak p_{\rm out}}(P)
            \sigma_{\mathfrak p_{\rm in}}(Q)
            7^{w_{\mathfrak p_{\rm out}}(P)+w_{\mathfrak p_{\rm in}}(Q)}
            \mathbf 1_{H_{P,Q}}
            \prod_{b:\,Q_b\neq\Id} S_b
            \prod_{a:\,P_a\neq\Id} Y_a .
            \label{eq:pair-block-estimator}
        \end{equation}
        On an active input block,
        $\mathbb E_S[S\rho_{b,Q_b,S}]=Q_b/4$, whereas an inactive block averages to $\Id/4$.  Cancelling the input- and output-axis hit probabilities therefore gives $\mathbb E Z_{P,Q}=2^{-|W|}\tr[P\Phi(Q)]$.

        The hit probability is exactly $7^{-w_{\mathfrak p_{\rm out}}(P)-w_{\mathfrak p_{\rm in}}(Q)}$, while $|Z_{P,Q}|=7^{w_{\mathfrak p_{\rm out}}(P)+w_{\mathfrak p_{\rm in}}(Q)}$ on the hit event and is zero otherwise.  This gives \eqref{eq:pair-block-second-moment}.
    \end{proof}
    Randomizing pairings allows us to estimate many input--output overlaps from the same data.

    \begin{lemma}[Randomized pairing estimators]\label{lem:randomized-pairing-estimators}
        In the setting of Proposition~\ref{prop:pair-block-process-shadow}, draw $s=(\mathfrak p_{\rm out},\mathfrak p_{\rm in})$ from a fixed, query-independent, given joint distribution. Let $\mathcal T$ be the finite nonempty set of required input--output queries. For each $\tau=(P,Q)$, the protocol specifies a set $\mathsf A_\tau$ of accepted pairings, all compatible with $P$ and $Q$, with known probability $p_\tau>0$. Suppose $P$ and $Q$ have at most $w_*$ nonidentity block factors in total on every accepted setting. With $Z_\tau^s$ from Proposition~\ref{prop:pair-block-process-shadow}, define
        \begin{equation}
            X_\tau:=p_\tau^{-1}\mathbf1_{\{s\in\mathsf A_\tau\}}Z_\tau^s,
            \label{eq:random-pairing-estimator}
        \end{equation}
        setting $X_\tau=0$ on rejection. Then
        \begin{equation*}
            \mathbb E X_\tau=2^{-|W|}\tr[P\Phi(Q)],
            \qquad
            \mathbb E X_\tau^2\le p_\tau^{-1}7^{w_*}.
        \end{equation*}
        If $p_\tau\ge p_*>0$ for all queries, one common dataset estimates all overlaps to accuracy $\eps_M>0$ with probability at least $1-\delta$, for $\delta\in(0,1)$, using
        \begin{equation}
            \cO\!\left(
                p_*^{-1}7^{w_*}\eps_M^{-2}
                \log\frac{2|\mathcal T|}{\delta}
            \right)
            \label{eq:random-pairing-cost}
        \end{equation}
        samples.
    \end{lemma}

    \begin{proof}
        Fix a query $\tau=(P,Q)$. For every accepted pairing $s$, Proposition~\ref{prop:pair-block-process-shadow} gives conditional mean $2^{-|W|}\tr[P\Phi(Q)]$ and conditional second moment at most $7^{w_*}$ for $Z_\tau^s$. Averaging over the pairing, with $X_\tau=0$ on rejection, therefore gives
        \begin{align*}
            \mathbb E X_\tau
            &=p_\tau^{-1}\Pr\{s\in\mathsf A_\tau\}\,2^{-|W|}\tr[P\Phi(Q)]
            =2^{-|W|}\tr[P\Phi(Q)],\\
            \mathbb E X_\tau^2
            &\le p_\tau^{-2}\Pr\{s\in\mathsf A_\tau\}\,7^{w_*}
            =p_\tau^{-1}7^{w_*}.
        \end{align*}
        Thus acceptance with probability $p_\tau$ cancels one factor of $p_\tau^{-1}$ in the second-moment bound.

        We next estimate all overlaps from repeated independent samples, redrawing the pairings each time. The variance of every $X_\tau$ is at most $p_*^{-1}7^{w_*}$. For $0<\eps_M<1$, divide the data into $\cO(\log(2|\mathcal T|/\delta))$ groups, each containing $\cO(p_*^{-1}7^{w_*}\eps_M^{-2})$ samples, with sufficiently large constants. For each query, average its $X_\tau$ values within each group and take the median of these averages. The median-of-means bound \cite[Theorem~4.1]{DevroyeEtAl2016} gives error at most $\eps_M$ with failure probability at most $\delta/|\mathcal T|$. A union bound gives simultaneous success with probability at least $1-\delta$ and the sample count \eqref{eq:random-pairing-cost}. The same groups serve all queries, since the union bound does not require independence between their estimates. For $\eps_M\ge1$, zero estimates suffice because every target overlap has absolute value at most one.
    \end{proof}

    Section~\ref{sec:ancilla-structure} gives unknown-structure learning with ancillary modes, and Appendix~\ref{sec:modewise-ancilla} supplies its measurement construction. Separate-query access is discussed in Remark~\ref{rem:one-marker-recovery} and model misspecification in Corollary~\ref{cor:fermionic-misspecification}.

\subsection{Internal-marker dynamics and pair calibration}\label{subsec:internal-calibration}

    The preceding construction used an idle auxiliary marker, with no Hamiltonian or jump operator involving its mode. The finite-range protocol instead uses markers within the system. Terms acting on these modes contribute to the measured derivatives and must be calibrated. For odd local jumps, the contribution is not the bare self-derivative of the marker Majorana because a remote odd factor changes the sign of the sandwich term. This sign modification is described algebraically by the remote-parity-twisted form of the same local generator.

    Fix a region $e$ and a local parity-covariant term $\cK:=\cL_e$ from the decomposition fixed in Section~\ref{sec:model}. In this subsection abbreviate
    \begin{equation*}
        H:=H_e,
        \qquad
        \cI_p:=\cI_{e,p},
        \qquad
        K^{(p)}:=K^{(e,p)}.
    \end{equation*}
    For $q\in\{0,1\}$, define
    \begin{align}
        \cK^{\langle q\rangle}(X)
        :=
        -i[H,X]
        +
        \sum_{p=0}^1\sum_{A,B\in\cI_p}
        K^{(p)}_{A,B}
        \left(
            (-1)^{pq}\Gamma_AX\Gamma_B
            -
            \frac12\{\Gamma_B\Gamma_A,X\}
        \right).
        \label{eq:twisted-generator}
    \end{align}
    Thus $\cK^{\langle0\rangle}=\cK$, while $\cK^{\langle1\rangle}$ flips only the sandwich term of the odd Kossakowski block.  For the local parity $\Pi_e:=\prod_{j\in e}(\Id-2n_j)$ and $R_{\Pi_e}(X):=X\Pi_e$, one has $\cK^{\langle1\rangle}=R_{\Pi_e}\cK R_{\Pi_e}$.  Since $R_{\Pi_e}$ is Hilbert--Schmidt unitary, $\|\cK^{\langle1\rangle}\|_{2\to2}=\|\cK\|_{2\to2}$.

    \begin{lemma}[Graded factorization]\label{lem:graded-factor}
        Let $X$ be homogeneous and supported on the region of $\cK$, and let $Y$ be homogeneous of parity $q$ on a disjoint region.  Under the fermionic extension,
        \begin{equation}
            \cK(XY)=\cK^{\langle q\rangle}(X)Y,
            \qquad
            \cK(YX)=Y\cK^{\langle q\rangle}(X).
            \label{eq:graded-factor}
        \end{equation}
    \end{lemma}

    \begin{proof}
        Hamiltonian and anticommutator multipliers are even and commute with $Y$.  In sandwich terms, the right multiplier of parity $p$ crosses $Y$, producing $(-1)^{pq}$.  The second is analogous.
    \end{proof}

    As in the triplet protocol for continuum free-fermion potential learning \cite{BluhmLemmMobusSiebert2026Coulomb}, fix distinct marker modes $x_1,x_2,x_3$, choose one Majorana $\eta_i$ on mode $x_i$, and set $X_i:=\{x_i\}$. Disjoint mode sets $U$ and $V$ are \emph{termwise separated} relative to the fixed decomposition if every nonzero local term acts on at most one of them. Equivalently, for each $e$ with $\cL_e\ne0$, either $e\cap U=\varnothing$ or $e\cap V=\varnothing$. Set
    \begin{equation*}
        \cK_i:=\sum_{\substack{e:\,\cL_e\ne0\\e\cap X_i\ne\varnothing}}\cL_e,
        \qquad
        Y_i:=X_i\cup
        \bigcup_{\substack{e:\,\cL_e\ne0\\e\cap X_i\ne\varnothing}}e .
    \end{equation*}
    The twist $\cK_i^{\langle q\rangle}$ is applied termwise as in \eqref{eq:twisted-generator}.  Define the remote-odd marker scalar
    \begin{equation}
        b_i
        :=
        2^{-|Y_i|}
        \tr_{Y_i}
        \left[
            \eta_i\cK_i^{\langle1\rangle}(\eta_i)
        \right].
        \label{eq:marker-b}
    \end{equation}
    For $i\neq j$, set
    \begin{equation}
        O_{ij}:=i\eta_i\eta_j,
        \qquad
        s_{ij}:=2^{-m}\tr\!\left[O_{ij}\cL(O_{ij})\right].
        \label{eq:marker-pair-derivative}
    \end{equation}

    \begin{proposition}[Pair-marker calibration identity]\label{prop:pair-marker}
        If marker cores $X_i$ and $X_j$ are termwise separated, then
        \begin{equation}
            s_{ij}=b_i+b_j.
        \end{equation}
        Hence, for three pairwise termwise-separated markers,
        \begin{align}
            b_1&=\frac12(s_{12}+s_{13}-s_{23}),\\
            b_2&=\frac12(s_{12}+s_{23}-s_{13}),\\
            b_3&=\frac12(s_{13}+s_{23}-s_{12}).
            \label{eq:three-marker-calibration}
        \end{align}
    \end{proposition}

    \begin{proof}
        Consider first a local term meeting $X_i$. Termwise separation makes its support disjoint from $X_j$, so Lemma~\ref{lem:graded-factor} gives
        \begin{equation*}
            \cL_e(O_{ij})
            =
            i\cL_e^{\langle1\rangle}(\eta_i)\eta_j.
        \end{equation*}
        The odd operator $\cL_e^{\langle1\rangle}(\eta_i)$ anticommutes with the disjoint Majorana $\eta_j$.  Hence
        \begin{equation*}
            O_{ij}\cL_e(O_{ij})
            =
            \eta_i\cL_e^{\langle1\rangle}(\eta_i),
        \end{equation*}
        and normalized tracing gives the contribution of $\cL_e$ to $b_i$. The same calculation for a term meeting $X_j$ gives its contribution to $b_j$. A term meeting neither marker contributes a factor $\tr[\cL_e(\Id)]=0$, while termwise separation excludes a term meeting both. Summing the local terms proves $s_{ij}=b_i+b_j$, and solving the three equations gives \eqref{eq:three-marker-calibration}.
    \end{proof}

\subsection{Corrected transfer derivatives}\label{subsec:corrected-transfer}

Let $R$ be a learning region termwise separated from marker core $X_i$. For this calculation, order the marker after the region. Any other fixed mode ordering introduces only known phases. Set $J_i:=J_{\eta_i}$ as in \eqref{eq:ancilla-dressing} and measure
    \begin{equation}
        \Lambda^{R|i}_{C,D}
        :=
        2^{-m}
        \tr\!\left[
            J_i(\Gamma_C)\cL(J_i(\Gamma_D))
        \right],
        \qquad C,D\subseteq\cM(R).
        \label{eq:internal-data}
    \end{equation}
    Let $\Lambda_{C,D}$ denote the ideal dressed transfer derivative obtained with an idle external marker for the same system generator.

    \begin{theorem}[Internal-marker correction]\label{thm:internal-correction}
        If $R$ and $X_i$ are termwise separated, then
        \begin{equation}
            \Lambda^{R|i}_{C,D}
            =
            \Lambda_{C,D}
            +
            \mathbf 1_{\{\pi(C)=\pi(D)=1\}}\delta_{C,D}\,b_i.
            \label{eq:internal-correction}
        \end{equation}
        Therefore the corrected data
        \begin{equation}
            \Lambda^{\rm corr}_{C,D}
            :=
            \Lambda^{R|i}_{C,D}
            -
            \mathbf 1_{\{\pi(C)=\pi(D)=1\}}\delta_{C,D}\,b_i
        \end{equation}
        are exactly the external-marker data needed in Theorem~\ref{thm:diagonal-extension}.
    \end{theorem}

    \begin{proof}
        Split the fixed decomposition term by term. A term meeting $R$ is disjoint from $X_i$ and therefore acts on the marker exactly as on an external remote odd factor. Summing these terms gives the ideal datum $\Lambda_{C,D}$. A term meeting $X_i$ is disjoint from $R$. Its normalized overlap vanishes unless $C$ and $D$ are both odd and equal. In the remaining case, Lemma~\ref{lem:graded-factor} and one graded reordering give
        \begin{equation*}
            2^{-|R|}\tr_R[\Gamma_C\Gamma_D]
            \;2^{-|Y_i|}
            \tr[\eta_i\cK_i^{\langle1\rangle}(\eta_i)]
            =
            \delta_{C,D}b_i.
        \end{equation*}
        A term meeting neither region contributes a factor $\tr[\cL_e(\Id)]=0$. Termwise separation excludes a term meeting both regions, proving \eqref{eq:internal-correction}.
    \end{proof}

    Under the stated separation, corrected derivatives enter the same signed inversion as external-marker data. Section~\ref{sec:pair-marker-blocks} gives the finite-range marker geometry and sampling bounds. Appendix~\ref{app:sdp-implementation} adapts the local PSD-block fit of \cite{MobusEtAl2026} to return a physical generator from noisy coefficients, as required for diamond-norm recovery and local-observable prediction.

\section{Geometric learning without external modes}\label{sec:pair-marker-blocks}

    Dissipative Fermi--Hubbard models combine hopping and onsite interactions with particle loss, injection, and dephasing. Under the geometric and experimental assumptions below, we learn the underlying canonical generator coefficients from even preparations and measurements after short-time evolution, using only modes already present in the system. We first treat finite-range generators and then extend the protocol to exponentially decaying interactions.

\subsection{Geometric model and marker packing}
    Let the modes be vertices of a graph $\mathsf G=([m],E)$ with graph distance $\dist_{\mathsf G}$. We take the interaction range $r\in\mathbb N_0$, without loss of generality. For $S\subseteq[m]$, write $\diam_{\mathsf G}(S)$ for its diameter and define
    \begin{equation*}
        B_R(S):=\{v\in[m]:\dist_{\mathsf G}(v,S)\le R\},
        \qquad
        \dist_{\mathsf G}(v,S):=\min_{u\in S}\dist_{\mathsf G}(v,u).
    \end{equation*}
    Set $B_R(\varnothing)=\varnothing$ and $\diam_{\mathsf G}(\varnothing)=0$. A coefficient pair $(A,B)$ is geometrically $(k,r)$-local if
    \begin{equation}
        |\fsupp(A,B)|\le k,
        \qquad \diam_{\mathsf G}(\fsupp(A,B))\le r.
    \end{equation}
    Let $\mathfrak L^{\rm geo}_{k,r}$ be the subcone of $\mathfrak L_k$ admitting a decomposition as in Definition~\ref{def:localcone} with $\diam_{\mathsf G}(e)\le r$ for every nonzero term $\cL_e$. Its coefficients vanish outside the geometrically allowed supports. These supports need not be connected, since a basis coefficient may omit an intermediate mode of a local term. Geometric locality is stronger than all-to-all $k$-locality, treated with auxiliary modes in Corollary~\ref{cor:alltoall-modewise}.

    \begin{assumption}[Geometric finite-range marker packing]\label{ass:marker-packing}
        The graph has maximum degree $\Delta=\cO(1)$. Every nonempty region $S$ with $|S|\le k$ and $\diam_{\mathsf G}(S)\le r$ has three single-mode marker cores $X_1(S),X_2(S),X_3(S)\subseteq B_{4r+3}(S)$ satisfying
        \begin{equation}
            \dist_{\mathsf G}(S,X_i)>r,
            \qquad \dist_{\mathsf G}(X_i,X_j)>r\quad(i\ne j).
        \end{equation}
    \end{assumption}

    Whenever a separated triplet exists at finite distance from $S$, one can choose it in $B_{4r+3}(S)$. Either this ball contains the connected component containing $S$, or a shortest path from $S$ supplies markers at distances $r+1$, $2(r+1)$, and $3(r+1)$.

    On a periodic cubic lattice with every side length greater than $2(4r+3)$, fix $x\in S$ and place markers along one coordinate direction at distances $r+j(r+1)$, $j=1,2,3$, from $x$. These markers satisfy the assumption. The scale-dependent construction and coloring bound used in Corollary~\ref{cor:exponential-lattice} are proved in Appendix~\ref{app:exponential-proofs}. The separation conditions ensure that a term of diameter at most $r$ meets at most one of $S,X_1,X_2,X_3$, exactly as required by the calibration identities of Subsections~\ref{subsec:internal-calibration} and~\ref{subsec:corrected-transfer}.

\subsection{Local even block frames}
    For a learning region $S$ and its marker cores, define the extended block
    \begin{equation}
        W(S):=S\cup X_1(S)\cup X_2(S)\cup X_3(S),
        \qquad |W(S)|=|S|+3\le k+3.
    \end{equation}
    Choose a Majorana on each marker mode and use the dressing maps $J_i$ and marker-pair observables $O_{ij}$ from Subsections~\ref{subsec:internal-calibration} and~\ref{subsec:corrected-transfer}. The required curves are
    \begin{equation}
        M^{S|i}_{C,D}(t):=2^{-m}\tr\!\left[
            J_i(\Gamma_C)e^{t\cL}(J_i(\Gamma_D))\right],
        \qquad C,D\subseteq\cM(S),
        \label{eq:block-internal-transfer}
    \end{equation}
    and
    \begin{equation}
        t\longmapsto 2^{-m}\tr[O_{ij}e^{t\cL}(O_{ij})],
        \qquad 1\le i<j\le3.
        \label{eq:block-pair-curves}
    \end{equation}
    Their derivatives at zero are $\Lambda^{S|i}_{C,D}$ and $s_{ij}$, respectively.

    \begin{definition}[Local internal block frame]\label{def:local-block-frame}
        A local internal block frame on $W(S)$ is an even randomized preparation-and-measurement scheme that, with maximally mixed input and identity measurement on $W(S)^c$, supplies from common data an unbiased estimator $Z_q$ for every curve $q$ in \eqref{eq:block-internal-transfer} and \eqref{eq:block-pair-curves}, with
        \begin{equation*}
            \mathbb E Z_q^2\le\kappa
        \end{equation*}
        uniformly over learning regions, queried curves, and sampled times, for a constant $\kappa$ independent of $m$. Its average input on $W(S)$ is maximally mixed.
    \end{definition}

    The frame need only cover the specified dressed-transfer and calibration observables. Local fermionic Gaussian control, including pairing rotations, together with occupation measurements provides a sufficient implementation \cite{JozsaMiyake2008,WanEtAl2023}. Related pair-creation and interferometric primitives appear in \cite{BluhmLemmMobusSiebert2026Coulomb}. Remark~\ref{rem:number-conserving-identifiability} explains the obstruction under strictly number-conserving access.
    \begin{proposition}[Pair-block realization of local block frames]\label{prop:pair-block-local-frame}
        Suppose the two-mode even preparations and measurements of Proposition~\ref{prop:pair-block-process-shadow} are implementable for every pairing of $W(S)$. If $|W(S)|$ is odd, allow one unpaired mode whose parity $Z_j=-i\gamma_{2j-1}\gamma_{2j}=\Id-2n_j$ can be prepared in either eigenstate $(\Id+\sigma Z_j)/2$, $\sigma\in\{\pm1\}$, and measured. These operations realize Definition~\ref{def:local-block-frame} with $\kappa=\kappa_k=\cO_k(1)$.
    \end{proposition}
    \begin{proof}
        Fix an even Majorana monomial on $W(S)$. The modes on which it contains exactly one Majorana occur in an even number, so pair them with one another and pair the remaining modes arbitrarily. The monomial then has even parity on every two-mode block. If one mode $j$ remains unpaired, its factor is $\Id$ or $-Z_j$, both accessible. Thus a suitable pairing exists for each required monomial. The experiment chooses input and output pairings independently and uniformly, without tailoring them to a query. Since $|W(S)|\le k+3$, the probability that both pairings suit the two monomials of any fixed query is bounded below by a positive constant depending only on $k$.

        Let $\mathbb E_{W(S)}:\mathfrak A_{[m]}\to\mathfrak A_{W(S)}$ be the normalized-trace-preserving conditional expectation. The reduced channel $\Phi_{t,W(S)}:=\mathbb E_{W(S)}\circ e^{t\cL}|_{\mathfrak A_{W(S)}}$ is completely positive, trace preserving, and parity covariant. It describes evolution with the modes outside $W(S)$ initialized in the maximally mixed state and satisfies
        \begin{equation*}
            2^{-|W(S)|}\tr_{W(S)}[P\Phi_{t,W(S)}(Q)]
            =2^{-m}\tr[P e^{t\cL}(Q)]
        \end{equation*}
        for even monomials $P,Q$ on $W(S)$. Thus the local estimator describes the actual full-system evolution, including interactions across the boundary of $W(S)$.

        Apply Lemma~\ref{lem:randomized-pairing-estimators}, with the singleton parity primitive when needed, to obtain $\kappa_k=\cO_k(1)$. Each input ensemble averages to the maximally mixed state, completing the frame construction.
    \end{proof}
\subsection{Parallel sampling}
    The triplet system \eqref{eq:three-marker-calibration} recovers the marker self-terms from the pair derivatives $s_{ij}$ in \eqref{eq:marker-pair-derivative}. Under Assumption~\ref{ass:marker-packing}, Theorem~\ref{thm:internal-correction} converts block-frame derivatives into external-marker derivatives with constant error amplification.

    To reuse each full-system sample across local regions, we follow the parallel process-shadow strategy of \cite{StilckFranca2024}, replacing product-Pauli inputs by independent even block frames. The following lemma justifies this parallel sampling for fermionic blocks.

    \begin{lemma}[Parallel even-block marginalization]\label{lem:parallel-block-marginalization}
        Let $U_1,\ldots,U_q$ be pairwise disjoint mode sets, and let $\rho_i(\omega_i)$ be independently sampled even input states with $\mathbb E_{\omega_i}\rho_i(\omega_i)=2^{-|U_i|}\Id_{U_i}$. Fix $j$ and compare two full-system preparations: (i) prepare $\rho_i(\omega_i)$ on every $U_i$, or (ii) prepare $\rho_j(\omega_j)$ on $U_j$ and every other $U_i$ maximally mixed. In both cases, prepare the complement maximally mixed and apply the same parity-covariant channel and local measurement on $U_j$. Any estimator using only the preparation and measurement records on $U_j$ has the same expectation in (i) and (ii).
    \end{lemma}

    \begin{proof}
        Even states on disjoint mode sets commute, so their product is a valid joint input. Let $\rho_{\rm rest}$ be maximally mixed on the complement, write $\omega_{-j}=(\omega_i)_{i\ne j}$, and let $\Phi$ be the full-system channel. Let $M_y$ be the POVM element for outcome $y$ of the measurement on $U_j$, extended by the identity elsewhere. Conditional on $\omega_j$, Born's rule and independence give
        \begin{equation*}
            \mathbb E_{\omega_{-j}}\tr\!\left[M_y\Phi\!\left(\rho_{\rm rest}\prod_{i=1}^q\rho_i(\omega_i)\right)\right]=\tr\!\left[M_y\Phi\!\left(\rho_{\rm rest}\rho_j(\omega_j)\prod_{i\ne j}2^{-|U_i|}\Id_{U_i}\right)\right].
        \end{equation*}
        The right side is the probability of outcome $y$ in (ii). Summing with the estimator's weights and averaging $\omega_j$ proves the claim, even when the dynamics couples the blocks.
    \end{proof}

    Color the intersection graph of $W(S)$ over all nonempty regions of size at most $k$ and diameter at most $r$. Their bounded radius and the graph degree give
    \begin{equation}
        K_{\rm col}=\cO_{k,r,\Delta}(1)
    \end{equation}
    colors, with mode-disjoint blocks within each color. For each sample, choose a color uniformly and run independent frames on its blocks. Multiply each selected block's estimator by $K_{\rm col}$ and set it to zero otherwise. Lemma~\ref{lem:parallel-block-marginalization} gives unbiasedness and second moment at most $K_{\rm col}\kappa_k$. Proposition~\ref{prop:fermionic-derivative-poly} then estimates the derivatives from these exact full-system curves.

\subsection{Learning guarantees and the Fermi--Hubbard example}

    Geometric inversion uses nonempty anchor supports. The dependent scalar $\chi_{\varnothing,\varnothing}$ receives contributions from the entire system and is excluded from the learned index set.

    \begin{theorem}[Geometric fermionic Lindbladian learning]\label{thm:geo-noaux}
        Fix $k,r$ and $\delta\in(0,1)$. Let $\mathsf G$ be a known graph of maximum degree $\Delta$, and suppose $\cL$ admits a decomposition as in Definition~\ref{def:localcone} with
        \begin{equation*}
            |e|\le k,
            \qquad \diam_{\mathsf G}(e)\le r,
            \qquad \alpha(\{\cL_e\})\le\bar\alpha.
        \end{equation*}
        The learner receives $k$, $r$, and $\bar\alpha>0$, but neither the decomposition nor the nonzero coefficient locations. Assume that a marker assignment satisfying Assumption~\ref{ass:marker-packing} is known or constructible from $\mathsf G$, and that the pair-block operations of Proposition~\ref{prop:pair-block-local-frame} are available for this assignment.

        A nonadaptive protocol using only local even preparations and measurements on modes of the system has the following guarantees. For $0<\eps_\chi\le\bar\alpha$, it estimates all induced Hamiltonian coefficients $h_A$ and Kossakowski entries $G_{A,B}$ to entrywise error $\eps_\chi$ using
        \begin{equation}
            \widetilde{\cO}_{k,r,\Delta}\!\left(
                \bar\alpha^2\eps_\chi^{-2}\log\frac m\delta
            \right)
            \label{eq:geo-noaux-entrywise}
        \end{equation}
        samples. For $0<\eps_\diamond\le m\bar\alpha$, a geometric consistency fit returns a valid $\widehat\cL\in\mathfrak L^{\rm geo}_{k,r}$ with $\|\widehat\cL-\cL\|_\diamond\le\eps_\diamond$ using
        \begin{equation}
            \widetilde{\cO}_{k,r,\Delta}\!\left(
                \bar\alpha^2m^2\eps_\diamond^{-2}\log\frac m\delta
            \right)
            \label{eq:geo-noaux-diamond}
        \end{equation}
        samples. Each guarantee holds with probability at least $1-\delta$. All evolution times lie in $[0,(4k\bar\alpha)^{-1}]$, and no control is applied during evolution. The notation $\widetilde{\cO}$ hides polylogarithmic interpolation factors and constants fixed by the displayed parameters and experimental frame.
    \end{theorem}

    \begin{proof}
        Choose an anchor mode in each nonempty support. All remaining modes lie in its bounded radius-$r$ ball, so the required query and coefficient lists have $\cO_{k,r,\Delta}(m)$ entries. Coefficients with support diameter greater than $r$ vanish and are set to zero. All error bounds below hold uniformly over the required regions and indices.

        Set $T=(4k\bar\alpha)^{-1}$ and choose a derivative accuracy $0<\eps_L\le\bar\alpha$. Proposition~\ref{prop:pair-block-local-frame} and Lemma~\ref{lem:parallel-block-marginalization}, with the coloring above, give unbiased estimators of the curves in \eqref{eq:block-internal-transfer} and \eqref{eq:block-pair-curves}, with second moments at most $K_{\rm col}\kappa_k=\cO_{k,r,\Delta}(1)$. Use the common time grid and degree $\ell$ of Proposition~\ref{prop:fermionic-derivative-poly}. Median-of-means estimation \cite[Theorem~4.1]{DevroyeEtAl2016} and a union bound over all curves and times give curve-value error $\eps_M=T\eps_L/(4C_0\ell^2)$ with probability at least $1-\delta$, using
        \begin{equation*}
            N=\widetilde{\cO}_{k,r,\Delta}\!\left(\bar\alpha^2\eps_L^{-2}\log\frac m\delta\right)
        \end{equation*}
        samples. By \eqref{eq:fermionic-curve-precision}, the proposition converts this curve-value error into derivative errors
        \begin{equation*}
            |\widehat\Lambda^{S|i}_{C,D}-\Lambda^{S|i}_{C,D}|\le\eps_L,
            \qquad |\widehat s_{ij}-s_{ij}|\le\eps_L.
        \end{equation*}

        The triplet formulas \eqref{eq:three-marker-calibration} give $|\widehat b_i-b_i|\le3\eps_L/2$. Subtracting these estimates as in Theorem~\ref{thm:internal-correction} therefore gives $|\widehat\Lambda^{\rm corr}_{C,D}-\Lambda_{C,D}|\le5\eps_L/2$. Insert the corrected estimates into \eqref{eq:restricted-chi}. Since $d_R\le2^k$, there are $d_R^4$ summands, and each trace has modulus at most $d_R$, the resulting restricted coefficients satisfy
        \begin{equation*}
            |\widetilde\chi^{[R]}_{A,B}-\chi^{[R]}_{A,B}|
            \le\tfrac52\,4^k\eps_L.
        \end{equation*}

        Recover $\widetilde\chi$ by the descending recursion \eqref{eq:descending-recursion}, using only extensions of diameter at most $r$. Each such extension enlarges $S=\fsupp(A,B)$ within $B_r(S)$, giving depth at most $k$ and bounded out-degree. Lemma~\ref{lem:triangular-error-propagation} therefore gives $|\widetilde\chi_{A,B}-\chi_{A,B}|\le C_{k,r,\Delta}\eps_L$. Extracting $\widetilde h$ and $\widetilde G$ by \eqref{eq:G-from-chi} and \eqref{eq:h-from-chi} does not increase the error. Choosing $\eps_L=\eps_\chi/C_{k,r,\Delta}$, with $C_{k,r,\Delta}\ge1$, proves \eqref{eq:geo-noaux-entrywise}.

        For physical reconstruction, first obtain raw coefficients $\widetilde h,\widetilde G$ to entrywise error $\eta$ by the preceding procedure. Fit a physical generator to these estimates using the semidefinite program \eqref{eq:sdp-projection} in Appendix~\ref{app:sdp-implementation}, using only Hamiltonian coordinates and local PSD blocks of diameter at most $r$. The true decomposition is feasible with objective at most $\eta$. Proposition~\ref{prop:fermionic-sdp-runtime} returns a feasible fit with objective at most $\eta+\eps_{\rm SDP}$, and \eqref{eq:sdp-explicit-error} gives
        \begin{equation*}
            \|\widehat h-h\|_\infty\le2\eta+\eps_{\rm SDP},
            \qquad
            \|\widehat G-G\|_\infty\le3\eta+2\eps_{\rm SDP}.
        \end{equation*}
        The PSD and support constraints ensure $\widehat\cL\in\mathfrak L^{\rm geo}_{k,r}$. Applying \eqref{eq:diamond-sum} from Proposition~\ref{prop:diamond} to the $\cO_{k,r,\Delta}(m)$ allowed coefficients yields
        \begin{equation*}
            \|\widehat\cL-\cL\|_\diamond
            \le C_{k,r,\Delta}m(\eta+\eps_{\rm SDP}).
        \end{equation*}
        Choose $\eta=c_{k,r,\Delta}\eps_\diamond/m$ and $\eps_{\rm SDP}\le\min\{1,c_{k,r,\Delta}\eps_\diamond/m\}$ with $c_{k,r,\Delta}>0$ sufficiently small. The diamond error is then at most $\eps_\diamond$, and substituting $\eta$ for the coefficient accuracy in \eqref{eq:geo-noaux-entrywise} proves \eqref{eq:geo-noaux-diamond}, on the same event of probability at least $1-\delta$.
    \end{proof}

    Thresholding the estimates identifies which Hamiltonian and dissipative coefficients are nonzero, assuming a sufficient gap between zero and nonzero magnitudes. For the coefficient vector $x=(h,G)$, write $\Omega_a:=\{v:|x_v|>a\}$. Uniform accuracy $\gamma$ and threshold $\lambda\ge\gamma$ give
    \begin{equation}
        \Omega_{\lambda+\gamma}
        \subseteq\{v:|\widehat x_v|>\lambda\}
        \subseteq\Omega_{\lambda-\gamma}.
        \label{eq:threshold-support-sandwich}
    \end{equation}
    This follows from the reverse triangle inequality and requires only $\cO_{k,r,\Delta}(m)$ additional operations. If every coefficient is zero or has modulus greater than $\lambda+\gamma$, the recovered support is exact. The identified structure is that of the canonical coefficients, not a particular local jump decomposition.

    \begin{example}[Dissipative Fermi--Hubbard models]\label{ex:dissipative-fermi-hubbard}
        Let a bounded-degree lattice carry two spin modes per site. Recall
        \begin{equation}
            H_{\rm FH}
            =-\sum_{\langle x,y\rangle,\sigma}
            \left(J_{xy,\sigma}a_{x,\sigma}^\dagger a_{y,\sigma}+\mathrm{h.c.}\right)
            +\sum_x U_xn_{x,\uparrow}n_{x,\downarrow}
            +\sum_{x,\sigma}\mu_{x,\sigma}n_{x,\sigma},
            \label{eq:local-fermi-hubbard-hamiltonian}
        \end{equation}
        where each edge is counted once and $J_{yx,\sigma}:=\overline{J_{xy,\sigma}}$. Allow the loss, injection, dephasing, pair-loss, and nearest-neighbor incoherent-hopping jumps in \eqref{eq:intro-fermi-hubbard-jumps}. Write $L_\ell=\sqrt{\gamma_\ell}F_\ell$, with rates $\gamma_\ell\ge0$ and mode supports $e_\ell$. Remove an irrelevant scalar from $H_{\rm FH}$ to make it traceless. Every term acts on at most two modes, including the quartic onsite interaction.

        Each lattice site $x$ contributes two vertices to the mode graph, $(x,\uparrow)$ and $(x,\downarrow)$. The mode graph consists of two copies of the site lattice, one per spin, with an additional edge connecting the two modes at each site. It has bounded degree, and the generator has $k=2$ and $r=1$. Writing $u=(x,\sigma)$, a sufficient strength bound is
        \begin{equation}
            \bar\alpha\ge
            2\max_{x,\sigma}\left(
            |\mu_{x,\sigma}|+|U_x|
            +\sum_{y:\,y\sim x}|J_{xy,\sigma}|
            +\sum_{\ell:\,u\in e_\ell}\gamma_\ell
            \right).
            \label{eq:fermi-hubbard-strength}
        \end{equation}
        Indeed, a Hamiltonian commutator has induced Hilbert--Schmidt norm at most $2\|H\|$, each hopping term has norm $|J_{xy,\sigma}|$, and each $F_\ell$ has norm at most one. Theorem~\ref{thm:geo-noaux} therefore applies whenever the mode graph satisfies its marker and pair-block access assumptions.
    \end{example}

    \begin{corollary}[Matching geometric diamond lower bound]\label{cor:geometric-diamond-lower}
        Let $\bar\alpha,t_{\max}>0$ and $0<\eps_\diamond\le\bar\alpha m/64$. There is a family of even, number-conserving, geometrically one-local Hamiltonian generators of weighted strength at most $\bar\alpha$ for which diamond recovery to error $\eps_\diamond$, with success probability at least $2/3$ for every member, requires
        \begin{equation*}
            N\ge\frac{m^2}{4096t_{\max}^2\eps_\diamond^2}
        \end{equation*}
        samples and worst-case total evolution time at least $m^2/(4096t_{\max}\eps_\diamond^2)$. This holds even for adaptive protocols using arbitrary fresh system--reference inputs and joint measurements, with evolution time at most $t_{\max}$ per sample.
    \end{corollary}

    \begin{proof}[Proof sketch]
        For the singleton Hamiltonian family constructed in Appendix~\ref{sec:lower-bounds}, Lemma~\ref{lem:adaptive-exposure} and the testing argument leading to \eqref{eq:geometric-quadratic-exposure-lower} give
        \begin{equation*}
            \sup_{\rm transcripts}\sum_j t_j^2
            \ge\frac{m^2}{4096\eps_\diamond^2}.
        \end{equation*}
        Since $\sum_jt_j^2\le Nt_{\max}^2$ and $\sum_jt_j^2\le t_{\max}\sum_jt_j$ on every transcript, the sample and worst-case total-time bounds follow.
    \end{proof}

    For $t_{\max}=(4k\bar\alpha)^{-1}$, the bounds are $\Omega_k(\bar\alpha^2m^2\eps_\diamond^{-2})$ samples and $\Omega_k(\bar\alpha m^2\eps_\diamond^{-2})$ total evolution time, matching Theorem~\ref{thm:geo-noaux} up to logarithmic and fixed geometric factors. Proposition~\ref{prop:fermionic-coefficient-lower} also shows that inverse-square precision dependence is necessary for estimating a single Hamiltonian coefficient at fixed maximum evolution time.

\subsection{Beyond finite range: exponentially decaying interactions}\label{sec:exponentially-decaying}
    Without a strict finite-range cutoff, spatial separation need not remove target--marker cross terms. The same internal-marker protocol still applies when their cumulative strength decays exponentially. Only the reconstruction is truncated. Every experiment uses the full unknown evolution.

    \begin{definition}[Exponentially decaying $k$-body decomposition]\label{def:exponential-witness}
        A decomposition as in Definition~\ref{def:localcone}, with $\ell_e:=\|\cL_e^\dagger\|_{2\to2}$, satisfies an exponential-tail bound with parameters $\bar\alpha,\mu>0$ if
        \begin{equation}
            a(R):=\max_{u\in[m]}
            \sum_{\substack{e\ni u\\\diam_{\mathsf G}(e)\ge R}}\ell_e
            \le\bar\alpha e^{-\mu R}
            \qquad(R\ge0).
            \label{eq:exponential-local-tail}
        \end{equation}
        The learner receives $\bar\alpha,\mu$, but neither the decomposition nor the nonzero coefficient locations. Since $a(0)=\alpha(\{\cL_e\})$, this bound also controls the interaction strength.
    \end{definition}

    Truncating the inversion at diameter $R$ and separating markers by $L$ introduces biases controlled by $a(R)$ and $a(L)$. For general mode graphs, Theorem~\ref{thm:exponentially-decaying-learning} in Appendix~\ref{app:exponential-proofs} bounds the reconstruction error and sample complexity in terms of the interaction tail and graph geometry. On periodic lattices, this yields the following guarantee with only polylogarithmic overhead.

    \begin{corollary}[Exponential tails on a periodic lattice]\label{cor:exponential-lattice}
        Let $\cL$ admit a decomposition as in Definition~\ref{def:exponential-witness} on the known nearest-neighbor mode graph $\mathsf G=(\mathbb Z/M\mathbb Z)^D$, with $m=M^D$. Fix $\delta\in(0,1)$, $0<\eps_\chi\le\bar\alpha$, and $0<\eps_\diamond\le m\bar\alpha$. There is a constant $C=C(k,D,\mu)>0$ (see Appendix~\ref{app:exponential-proofs}) such that the following holds. Set
        \begin{equation*}
            R_\chi:=\left\lceil C\log\!\left(e+\frac{\bar\alpha}{\eps_\chi}\right)\right\rceil,
            \qquad
            R_\diamond:=\left\lceil C\log\!\left(e+\frac{m\bar\alpha}{\eps_\diamond}\right)\right\rceil.
        \end{equation*}
        Assume $4\max\{R_\chi,R_\diamond\}+3<M/2$ and the pair-block access of Proposition~\ref{prop:pair-block-local-frame} on the marker blocks constructed in Appendix~\ref{app:exponential-proofs}. The protocol estimates all canonical coefficients to entrywise error $\eps_\chi$, or returns a valid generator with interaction range at most $R_\diamond$ and diamond error at most $\eps_\diamond$, using respectively
        \begin{align}
            N_{\rm coeff}
            &=\widetilde{\cO}_{k,D,\mu}\!\left(
                \bar\alpha^2\eps_\chi^{-2}\log\frac m\delta
            \right),
            \label{eq:exponential-lattice-entrywise}\\
            N_\diamond
            &=\widetilde{\cO}_{k,D,\mu}\!\left(
                \bar\alpha^2m^2\eps_\diamond^{-2}\log\frac m\delta
            \right)
            \label{eq:exponential-lattice-diamond}
        \end{align}
        samples. Each guarantee holds with probability at least $1-\delta$. The protocols use only system modes and even preparations and measurements, are nonadaptive, and apply no control during evolution. All evolution times are at most $(4k\bar\alpha)^{-1}$. Here $\widetilde{\cO}$ also hides polylogarithmic factors in $\bar\alpha/\eps_\chi$ and $m\bar\alpha/\eps_\diamond$, respectively.
    \end{corollary}

    \begin{proof}[Proof sketch]
        Apply Theorem~\ref{thm:exponentially-decaying-learning} with $L=R$ and $R\in\{R_\chi,R_\diamond\}$. Under the finite-size condition, Appendix~\ref{app:exponential-proofs} constructs the required markers. The lattice bound $|B_R(\{u\})|\le(2R+1)^D$ makes the coloring and inversion factors polynomial in $R$. The stated logarithmic cutoffs satisfy \eqref{eq:exponential-entrywise-condition} and \eqref{eq:exponential-diamond-condition} for sufficiently large $C$, since exponential decay dominates these polynomial factors. Substitution into the theorem's sample bounds gives the claimed polylogarithmic overhead.
    \end{proof}

    At cutoff $R$, the required pair operations span distances $\cO(R)$, and their routing cost is not included in the sample count. This polylogarithmic overhead uses polynomial growth of graph balls. Bounded degree alone does not ensure it.

\section{All-to-all structure learning with ancillary modes}\label{sec:ancilla-structure}

    We next learn unknown all-to-all interactions using one idle ancillary mode per system mode. Approximate coefficient degree replaces a supplied graph. Building on the local inversion framework of \cite{MobusEtAl2026}, we threshold signed fermionic marginals before inversion to obtain local coefficient error guarantees as in qubit structure learning \cite{LewisTangWright2026}. The ancillary measurement construction in Appendix~\ref{sec:modewise-ancilla} supplies common even data, with an additional sampling cost from pairing compatibility.

\subsection{Local coefficient error and approximate degree}

    Write $x=(h,G)$ for the physical coefficient array. Index its entries by $v$ and define their mode supports by
    \begin{equation*}
        S_v:=
        \begin{cases}
            \supp_{\rm f}(A), & \text{if }v\text{ indexes }h_A,\\
            \fsupp(A,B), & \text{if }v\text{ indexes }G_{A,B}.
        \end{cases}
    \end{equation*}
    Only parity-compatible coordinates with nonempty support of size at most $k$ are included. Define
    \begin{equation}
        \|x\|_{\mathrm{loc},1}
        :=\max_{j\in[m]}\sum_{v:j\in S_v}|x_v|,
        \qquad
        \deg(x)
        :=\max_{j\in[m]}\#\{v:j\in S_v,\ x_v\neq0\}.
        \label{eq:structure-local-norm}
    \end{equation}
    The same definitions apply to other arrays with their indicated nonempty supports. We impose the \emph{approximate coefficient-degree condition}, adapted from \cite[Definition~2.5]{LewisTangWright2026} to Majorana coefficients and fermionic mode supports: there exists an unknown coefficient array $x^{\rm big}$ on the same index set such that
    \begin{equation}
        \deg(x^{\rm big})\le\mathsf d,
        \qquad
        \|x-x^{\rm big}\|_{\mathrm{loc},1}\le b.
        \label{eq:structure-approx-degree}
    \end{equation}
    Its locations are not supplied, and it need not generate a positive semigroup. Exact degree at most $\mathsf d$ is the case $b=0$. This assumption concerns the identifiable coefficients, rather than a particular decomposition into local jumps.

\subsection{Thresholding before signed inversion}

    For each $(A,B)\in\cI_k^\chi$, set $R=\fsupp(A,B)$ and write
    \begin{equation*}
        y_{A,B}:=\chi^{[R]}_{A,B}.
    \end{equation*}
    The restriction region thus depends on $(A,B)$. Under $k$-locality, \eqref{eq:diagonal-extension} defines a signed linear map $\mathsf T$ on arrays indexed by $\cI_k^\chi$, with $y=\mathsf T\chi$.

    \begin{lemma}[Signed inversion with local coefficient control]\label{lem:structure-signed-inverse}
        The inverse of $\mathsf T$ on $\cI_k^\chi$ is
        \begin{equation}
            \chi_{A,B}
            =\sum_{\substack{E\subseteq\cM(R^c)\\ |R\cup\supp_{\rm f}(E)|\le k}}
            (-1)^{|\supp_{\rm f}(E)|}
            s_R(A,B;E)\,y_{A\sqcup E,B\sqcup E}.
            \label{eq:structure-signed-inverse}
        \end{equation}
        Both $\mathsf T$ and $\mathsf T^{-1}$ have induced local coefficient one-norm at most $2^k$. If $\|\widehat y-y\|_\infty\le\nu$ for $\nu>0$, set
        \begin{equation}
            \widetilde y_{A,B}
            :=\widehat y_{A,B}
            \mathbf1_{\{|\widehat y_{A,B}|>2\nu\}}.
            \label{eq:structure-marginal-threshold}
        \end{equation}
        Extract $\widehat x$ from $\mathsf T^{-1}\widetilde y$ using \eqref{eq:G-from-chi} and \eqref{eq:h-from-chi}. Under \eqref{eq:structure-approx-degree},
        \begin{equation}
            \|\widehat x-x\|_{\mathrm{loc},1}
            \le9\cdot4^k\mathsf d\nu+2\cdot4^k b.
            \label{eq:structure-threshold-error}
        \end{equation}
        The output has degree $\cO_k(\mathsf d+b/\nu)$. Taking the real part of $\widehat h$ and the Hermitian part of $\widehat G$ preserves the error bound and changes this degree bound by at most a constant factor.
    \end{lemma}

    \begin{proof}
        We first prove the inverse formula \eqref{eq:structure-signed-inverse} by applying inclusion--exclusion to the diagonal-extension identity \eqref{eq:diagonal-extension} of Theorem~\ref{thm:diagonal-extension}. This requires a composition rule for the extension signs.

        For extensions $E,F$ on disjoint additional mode sets, associativity and the product rule \eqref{eq:union-product} give $\omega(A,E)\omega(A\sqcup E,F)=\omega(A,E\sqcup F)\omega(E,F)$. Applying this to $A$ and $B$ produces the factor $\omega(E,F)^2=(-1)^{\pi(E)\pi(F)}$ by \eqref{eq:omega}. In the extension-sign definition \eqref{eq:extension-phase}, this cancels the extra factor $(-1)^{\pi(F)\pi(E)}$ arising from $\pi(B\sqcup E)=\pi(B)+\pi(E)$ modulo two. Hence, with $R'=R\cup\supp_{\rm f}(E)$,
        \begin{equation*}
            s_R(A,B;E)s_{R'}(A\sqcup E,B\sqcup E;F)
            =s_R(A,B;E\sqcup F).
        \end{equation*}
        We now apply \eqref{eq:diagonal-extension} on $R'$ to obtain
        \begin{equation*}
            y_{A\sqcup E,B\sqcup E}
            =\sum_{\substack{F\subseteq\cM((R')^c)\\ |R'\cup\supp_{\rm f}(F)|\le k}}
            s_{R'}(A\sqcup E,B\sqcup E;F)\,
            \chi_{A\sqcup E\sqcup F,B\sqcup E\sqcup F}.
        \end{equation*}
        Substituting into the right-hand side of \eqref{eq:structure-signed-inverse}, using the sign identity above, and grouping by the final extension $H=E\sqcup F$ yields
        \begin{align*}
            &\sum_{\substack{E\subseteq\cM(R^c),\ F\subseteq\cM((R')^c)\\ |R'\cup\supp_{\rm f}(F)|\le k}}
            (-1)^{|\supp_{\rm f}(E)|}
            s_R(A,B;E)s_{R'}(A\sqcup E,B\sqcup E;F)\,
            \chi_{A\sqcup E\sqcup F,B\sqcup E\sqcup F}
            \\
            &\qquad\qquad\qquad=\sum_{\substack{H\subseteq\cM(R^c)\\ |R\cup\supp_{\rm f}(H)|\le k}}
            s_R(A,B;H)\,\chi_{A\sqcup H,B\sqcup H}
            \sum_{V\subseteq\supp_{\rm f}(H)}(-1)^{|V|}
            =\chi_{A,B}.
        \end{align*}
        Indeed, $E$ and $F$ occupy disjoint mode sets. For fixed $H$, choosing the intermediate mode set $V\subseteq\supp_{\rm f}(H)$ uniquely fixes $E=H\cap\cM(V)$ and $F=H\setminus E$. The final equality follows from inclusion--exclusion: the inner sum is zero unless $H=\varnothing$, when it equals one. This proves \eqref{eq:structure-signed-inverse}, whose exponent counts modes rather than Majorana labels.

        Next, we control error propagation through $\mathsf T$ and $\mathsf T^{-1}$. Since supports only shrink, maximum absolute column sums bound their induced local coefficient one-norms. Each input coefficient contributes to at most $2^k$ outputs by removing subsets of its occupied modes with matching Majorana labels in both indices. The corresponding entries of $\mathsf T$ and, by \eqref{eq:structure-signed-inverse}, $\mathsf T^{-1}$ are signs $\pm1$. Thus both induced norms and degree amplification are at most $2^k$, and the condition number, the product of these norms, is at most $4^k$.

        To relate these bounds to the physical coefficients, let $\mathsf F$ map $x=(h,G)$ to the Majorana-superoperator coefficients indexed by $\cI_k^\chi$, and let $\Psi$ extract $(h,G)$ by \eqref{eq:G-from-chi} and \eqref{eq:h-from-chi}. In \eqref{eq:local-gksl}, each Hamiltonian coordinate contributes two entries of modulus $|h_A|$, and each dissipative coordinate contributes at most three entries of total modulus at most $2|G_{A,B}|$. These contributions do not enlarge supports, and extraction does not increase degree. Hence $\Psi\mathsf F=\mathrm{id}$ and
        \begin{equation}
            \|\mathsf F x\|_{\mathrm{loc},1}\le2\|x\|_{\mathrm{loc},1},
            \qquad
            \deg(\mathsf F x)\le3\deg(x),
            \qquad
            \|\Psi z\|_{\mathrm{loc},1}\le\|z\|_{\mathrm{loc},1},
            \label{eq:structure-coefficient-maps}
        \end{equation}
        We now transfer the sparse approximation in \eqref{eq:structure-approx-degree} to the restricted coefficients. Set $y^{\rm big}:=\mathsf T\mathsf F x^{\rm big}$. Since $y=\mathsf T\mathsf F x$, the preceding bounds give
        \begin{equation*}
            \deg(y^{\rm big})\le3\cdot2^k\mathsf d,
            \qquad
            \|y-y^{\rm big}\|_{\mathrm{loc},1}\le2^{k+1}b.
        \end{equation*}
        We next bound the error of the thresholded estimates $\widetilde y$ defined in \eqref{eq:structure-marginal-threshold}. On the support of $y^{\rm big}$, an estimate satisfying $|\widehat y_v|>2\nu$ has error at most $\nu$, while a discarded one has error $|y_v|\le|\widehat y_v|+\nu\le3\nu$. Outside that support, $y_v^{\rm big}=0$. A discarded entry has error $|y_v|$, while an entry with $|\widehat y_v|>2\nu$ has error at most $\nu<|y_v|$. Summing per mode yields
        \begin{equation*}
            \|\widetilde y-y\|_{\mathrm{loc},1}
            \le3\nu\deg(y^{\rm big})+\|y-y^{\rm big}\|_{\mathrm{loc},1}
            \le9\cdot2^k\mathsf d\nu+2^{k+1}b.
        \end{equation*}
        Applying the reconstruction map $\Psi\mathsf T^{-1}$ now gives
        \begin{equation*}
            \|\widehat x-x\|_{\mathrm{loc},1}
            \le2^k\|\widetilde y-y\|_{\mathrm{loc},1}
            \le9\cdot4^k\mathsf d\nu+2\cdot4^k b,
        \end{equation*}
        proving \eqref{eq:structure-threshold-error}.

        Finally, we bound the output degree. Every index with $|\widehat y_v|>2\nu$ outside the support of $y^{\rm big}$ has $|y_v|>\nu$, so the tail bound allows at most $2^{k+1}b/\nu$ such indices per mode. Together with $\deg(y^{\rm big})\le3\cdot2^k\mathsf d$, this gives $\deg(\widetilde y)=\cO_k(\mathsf d+b/\nu)$. Inversion increases degree by at most $2^k$, and extraction does not increase it. Since the true $h$ is real and $G$ is Hermitian, taking real and Hermitian parts preserves the error bound and increases degree by at most a factor two.
    \end{proof}

\subsection{Learning with idle ancillary modes}

    Pair every system mode with one idle ancillary mode at the input. Ancillary Majoranas make all required inputs even on these pairs. The ancillary factor in each dressed transfer curve is fixed by its input label, so different output queries can require different pairings. Proposition~\ref{prop:modewise-process-shadow} gives a target-independent randomized output pairing and a common-data estimator with second moment $\cO_k(m^{\lfloor k/2\rfloor})$. The preparations and measurements act on pairs of modes, which need not be nearby. Appendix~\ref{sec:modewise-ancilla} gives the explicit construction.

    \begin{theorem}[All-to-all structure learning with ancillary modes]\label{thm:ancilla-structure}
        Let $\cL\in\mathfrak L_k$ have a local decomposition of strength at most the supplied bound $\bar\alpha>0$. Fix $0<\eps\le\bar\alpha$, $\delta\in(0,1)$, and a supplied bound $\mathsf d\ge1$. Suppose the approximate coefficient-degree condition \eqref{eq:structure-approx-degree} holds with
        \begin{equation*}
            b\le\frac{\eps}{4\cdot4^k}.
        \end{equation*}
        Under the idle-ancilla and pair-operation access of Proposition~\ref{prop:modewise-process-shadow}, there is a nonadaptive protocol returning real $\widehat h$ and Hermitian $\widehat G$ such that, with probability at least $1-\delta$,
        \begin{equation}
            \|\widehat x-x\|_{\mathrm{loc},1}\le\eps,
            \qquad
            \deg(\widehat x)=\cO_k(\mathsf d).
            \label{eq:ancilla-structure-guarantee}
        \end{equation}
        Neither an interaction graph nor nonzero coefficient locations are supplied. The protocol uses $m$ idle ancillary modes and
        \begin{equation}
            N=\widetilde{\cO}_k\!\left(
                \bar\alpha^2\mathsf d^2m^{\lfloor k/2\rfloor}
                \eps^{-2}\log\frac m\delta
            \right)
            \label{eq:ancilla-structure-samples}
        \end{equation}
        samples at polylogarithmically many times in $[0,(4k\bar\alpha)^{-1}]$, with no control during evolution. Its total evolution time obeys $T_{\rm evol}\le N/(4k\bar\alpha)$. Once the restricted coefficients are estimated, thresholding and inversion cost $\cO_k(m^k)$ arithmetic operations.
    \end{theorem}

    \begin{proof}
        We first estimate the restricted coefficients. Set $\nu=\eps/(18\cdot4^k\mathsf d)$ and $\eps_L=\nu/4^k$. The transform \eqref{eq:restricted-chi} amplifies entrywise derivative errors by at most $4^{|R|}\le4^k$, so derivative accuracy $\eps_L$ ensures $\|\widehat y-y\|_\infty\le\nu$. Apply Proposition~\ref{prop:modewise-process-shadow} at the common times and curve accuracy required by Proposition~\ref{prop:fermionic-derivative-poly} and \eqref{eq:fermionic-curve-precision}, with $\beta=0$. Allocating total failure probability $\delta$ across all $\cO_k(m^k)$ queries and interpolation times gives this bound with probability at least $1-\delta$. The variance $\cO_k(m^{\lfloor k/2\rfloor})$ and inverse squared curve accuracy $\widetilde{\cO}_k(T^{-2}\eps_L^{-2})$ give \eqref{eq:ancilla-structure-samples}, since $T=(4k\bar\alpha)^{-1}$.

        Next, threshold and invert as in Lemma~\ref{lem:structure-signed-inverse}. On the same event, its error bound gives
        \begin{equation*}
            \|\widehat x-x\|_{\mathrm{loc},1}
            \le9\cdot4^k\mathsf d\nu+2\cdot4^k b
            \le\frac\eps2+\frac\eps2=\eps.
        \end{equation*}
        Since $b/\nu\le(9/2)\mathsf d$, the lemma also gives degree $\cO_k(\mathsf d)$. Taking real and Hermitian parts gives the stated output without increasing the error.

        Finally, scanning the $\cO_k(m^k)$ restricted coefficients and distributing each nonzero entry of $\widetilde y$ to at most $2^k$ inverse entries costs $\cO_k(m^k)$ arithmetic operations. The total evolution-time bound follows from $t_i\le T$ for every sample.
    \end{proof}

    The local-norm bound implies entrywise error at most $\eps$, so the significant-support inclusions in \eqref{eq:threshold-support-sandwich} apply with accuracy $\eps$. Exact support recovery additionally requires every nonzero coefficient to have magnitude greater than $2\eps$, in which case thresholding at $\eps$ suffices. The theorem identifies Hamiltonian and Kossakowski coefficients, not a unique choice of local jump operators. Compared with dense all-to-all recovery in Corollary~\ref{cor:alltoall-modewise}, its sample bound replaces the factor $m^{2k-2}$ by $\mathsf d^2$. The remaining factor $m^{\lfloor k/2\rfloor}$ comes from the available pairing scheme and is not a lower bound for general even measurements.

    The coefficient estimate need not define a valid GKSL generator. A classical fit gives the following physical reconstruction guarantee.
    \begin{corollary}[Physical reconstruction in local coefficient norm]\label{cor:ancilla-physical-reconstruction}
        Under the assumptions of Theorem~\ref{thm:ancilla-structure}, there is a protocol returning a valid generator $\widehat\cL=\cL_{\widehat H,\widehat G}\in\mathfrak L_k$ whose coefficients $\widehat x=(\widehat h,\widehat G)$ satisfy, with probability at least $1-\delta$,
        \begin{equation*}
            \|\widehat x-x\|_{\mathrm{loc},1}\le3\eps.
        \end{equation*}
        The protocol uses
        \begin{equation*}
            N=\widetilde{\cO}_k\!\left(
                \bar\alpha^2\mathsf d^2m^{\lfloor k/2\rfloor}
                \eps^{-2}\log\frac m\delta
            \right)
        \end{equation*}
        samples, with the same ancillary access and evolution times as Theorem~\ref{thm:ancilla-structure}. The fit requires no additional samples and may increase the coefficient degree.
    \end{corollary}

    \begin{proof}
        Denote the raw estimate from Theorem~\ref{thm:ancilla-structure} by $\widetilde x$ and fit the local PSD blocks of Definition~\ref{def:localcone} with objective
        \begin{equation}
            \inf_{x'=(h',G'):\,\cL_{H',G'}\in\mathfrak L_k}
            \|x'-\widetilde x\|_{\mathrm{loc},1}.
            \label{eq:local-norm-fit}
        \end{equation}
        Choose $\eps_{\rm SDP}=\min\{1,\eps\}$. On the theorem's success event, $\|\widetilde x-x\|_{\mathrm{loc},1}\le\eps$. The bounded rational implementation in Appendix~\ref{app:sdp-implementation} returns an exactly feasible fit with
        \begin{equation*}
            \|\widehat x-\widetilde x\|_{\mathrm{loc},1}
            \le\|x-\widetilde x\|_{\mathrm{loc},1}+\eps_{\rm SDP}
            \le2\eps.
        \end{equation*}
        The triangle inequality gives $\|\widehat x-x\|_{\mathrm{loc},1}\le3\eps$.
    \end{proof}
    Using the local norm in the fit gives the $3\eps$ bound on the total coefficient error at each mode. In the entrywise fit \eqref{eq:sdp-projection}, many small errors at one mode can add up. Next, we take the maximal coefficient degree, which describes the dense case, with no additional sparsity assumption.
    \begin{corollary}[Dense all-to-all recovery]\label{cor:alltoall-modewise}
        Let $\cL\in\mathfrak L_k$ have a local decomposition of strength at most the supplied bound $\bar\alpha>0$. Under the ancillary access of Theorem~\ref{thm:ancilla-structure}, no degree or tail assumption is needed to estimate its coefficients with $\|\widetilde x-x\|_{\mathrm{loc},1}\le\eps_\chi$ using
        \begin{equation}
            \widetilde{\cO}_k\!\left(
                \bar\alpha^2m^{2k-2+\lfloor k/2\rfloor}
                \eps_\chi^{-2}\log\frac m\delta
            \right)
            \label{eq:alltoall-modewise-entrywise}
        \end{equation}
        samples, for $0<\eps_\chi\le\bar\alpha$ and $\delta\in(0,1)$. The local-norm fit \eqref{eq:local-norm-fit} instead returns a valid $\widehat\cL\in\mathfrak L_k$ with $\|\widehat\cL-\cL\|_\diamond\le\eps_\diamond$ using
        \begin{equation}
            \widetilde{\cO}_k\!\left(
                \bar\alpha^2m^{2k+\lfloor k/2\rfloor}
                \eps_\diamond^{-2}\log\frac m\delta
            \right)
            \label{eq:alltoall-modewise-diamond}
        \end{equation}
        samples, for $0<\eps_\diamond\le m\bar\alpha$. Each guarantee holds with probability at least $1-\delta$. Both protocols are nonadaptive, use $m$ idle ancillary modes and evolution times at most $(4k\bar\alpha)^{-1}$, and apply no control during evolution.
    \end{corollary}

    \begin{proof}
        Every $k$-local coefficient array has degree at most $C_km^{k-1}$. Apply Theorem~\ref{thm:ancilla-structure} with this bound and $b=0$. Proposition~\ref{prop:diamond} and nonempty coefficient supports give
        \begin{equation*}
            \|\widehat\cL-\cL\|_\diamond
            \le 2\sum_v|\widehat x_v-x_v|
            \le 2m\|\widehat x-x\|_{\mathrm{loc},1}.
        \end{equation*}
        Corollary~\ref{cor:ancilla-physical-reconstruction}, with coefficient accuracy $\eps_\diamond/(8m)$, therefore gives the second claim.
    \end{proof}

    \begin{remark}[Separate-query recovery with one external marker]\label{rem:one-marker-recovery}
        Under the generator assumptions and accuracy ranges of Corollary~\ref{cor:alltoall-modewise}, suppose one idle external marker is available and, for every $(C,D)$ with $|\fsupp(C,D)|\le k$, the experiment can prepare $\rho_D^\pm$ in \eqref{eq:physical-inputs} and measure $J_\eta(\Gamma_C)$, with all other modes initialized in the maximally mixed state. Proposition~\ref{prop:even-clifford-eigenaccess} supplies this access using local Gaussian operations. Separate datasets for the queries give local coefficient error $\eps_\chi$ or valid-generator diamond error $\eps_\diamond$ using, respectively,
        \begin{equation*}
            \widetilde{\cO}_k\!\left(
                \bar\alpha^2m^{3k-2}\eps_\chi^{-2}\log\frac m\delta
            \right)
            \quad\text{and}\quad
            \widetilde{\cO}_k\!\left(
                \bar\alpha^2m^{3k}\eps_\diamond^{-2}\log\frac m\delta
            \right)
        \end{equation*}
        samples, with the same success probability and evolution times. For a uniformly random input sign $S$ and measurement outcome $Y$, the bounded estimator $SY$ is unbiased by \eqref{eq:physical-transfer-states}. For $D=\varnothing$, use the maximally mixed input and $Y$ itself. Hoeffding's inequality and a union bound over the $\cO_k(m^k)$ queries replace the common-data measurement factor $m^{\lfloor k/2\rfloor}$ by $m^k$. Proposition~\ref{prop:fermionic-derivative-poly} and the same thresholded reconstruction and local-norm fit as in Corollary~\ref{cor:alltoall-modewise} give the stated bounds.
    \end{remark}

    The next result allows deviations from the assumed model. Here $\cL$ governs the actual experiments, while $\cK$ is a comparison generator satisfying the model assumptions. The estimate $\widetilde x$ is computed from these experiments and compared with $x(\cK)$, with additional error controlled by $\|\cL-\cK\|_\diamond$. The learner need not know $\cK$. If $\cL$ satisfies the assumptions, take $\cK=\cL$.

    \begin{corollary}[Stability to model misspecification]\label{cor:fermionic-misspecification}
        Let $\cL$ be the unknown time-independent, parity-covariant Lindbladian. Choose the experimental access of Theorem~\ref{thm:geo-noaux}, Corollary~\ref{cor:alltoall-modewise}, or Remark~\ref{rem:one-marker-recovery}, and let $\cK$ satisfy the corresponding model assumptions with weighted strength at most $\bar\alpha$, including marker separation in the geometric case. Fix $0<\eps_L\le\bar\alpha$ and $\delta\in(0,1)$. Use the polynomial degree $\ell$, sampling times, and curve-estimation accuracy specified in Proposition~\ref{prop:fermionic-derivative-poly} with $\beta=0$, allocating total failure probability $\delta$ to curve estimation. Unthresholded inversion returns canonical coefficient estimates $\widetilde x$ satisfying, with probability at least $1-\delta$,
        \begin{equation*}
            \|\widetilde x-x(\cK)\|_\infty
            \le C_{\rm inv}\left(
                \eps_L+\ell^2\|\cL-\cK\|_\diamond
            \right),
        \end{equation*}
        where $x(\cK)$ denotes the canonical coefficients of $\cK$, and $C_{\rm inv}=C_{k,r,\Delta}$ for geometric learning or $C_{\rm inv}=C_km^{k-1}$ for all-to-all learning. The curve-estimation sample counts at nominal derivative accuracy $\eps_L$ are unchanged.
    \end{corollary}

    \begin{proof}
        Following \cite[proof of Proposition~6.1]{MobusEtAl2026}, Duhamel's formula gives
        \begin{equation}
            \|e^{t\cL}-e^{t\cK}\|_\diamond
            \le t\|\cL-\cK\|_\diamond.
            \label{eq:misspecified-duhamel}
        \end{equation}
        Lemma~\ref{lem:graded-diamond-amplification} gives the same bound with idle external modes. Thus every normalized transfer or calibration curve differs from its comparator curve by at most $T\|\cL-\cK\|_\diamond$ on $[0,T]$. Apply the polynomial approximation to the comparator curves. The regression and Markov bounds in the proof of Proposition~\ref{prop:fermionic-derivative-poly} then give derivative error at most $\eps_L+C\ell^2\|\cL-\cK\|_\diamond$ relative to $\cK$.

        For geometric learning, Theorem~\ref{thm:internal-correction} and the inversion bound in the proof of Theorem~\ref{thm:geo-noaux} propagate this error through calibration and reconstruction. For all-to-all learning, use the restricted transform \eqref{eq:restricted-chi} and Lemma~\ref{lem:fermionic-stability}. These give the stated constants for the physical coefficients.
    \end{proof}

    If $\cK$ is feasible for the entrywise semidefinite fit, Proposition~\ref{prop:fermionic-sdp-runtime} preserves this bound up to a factor of three and additive error $2\eps_{\rm SDP}$, for $0<\eps_{\rm SDP}\le1$. This corollary concerns entrywise stability, without asserting sparse local-norm guarantees or support recovery below the misspecification error.

\section{Local-observable prediction}\label{sec:local-observable-prediction}
    The preceding results control generator coefficients. We now show that their local one-norm error also controls predictions of even observables, without requiring global diamond accuracy. This applies to both geometric learning and ancillary structure learning. With a known finite-range geometry, the reconstruction and subsequent simulation can additionally be restricted to a neighborhood of the observable. The proofs in Appendix~\ref{app:local-prediction-proofs} adapt dissipative Lieb--Robinson methods \cite{Poulin2010LR,BarthelKliesch2012} to even CAR observables, explicitly checking fermionic locality. For related Hamiltonian bounds, see \cite{NachtergaeleSimsYoung2017FermionLR}. Related learning and prediction guarantees appear in \cite{OnoratiEtAl2023,KraftEtAl2026}.

    For $U\subseteq[m]$, let $\mathfrak A_U^+$ denote the even subalgebra of $\mathfrak A_U$, and write $\|\cdot\|$ for the operator norm and $\|\cdot\|_1$ for the trace norm. For a support bound $q\in\{1,\ldots,m\}$ and prediction time $t\ge0$, define the local-observable prediction error
    \begin{equation}
        d_{q,t}(\cL,\widehat\cL)
        :=\max_{\substack{U\subseteq[m]\\|U|\le q}}
        \sup_{\substack{O\in\mathfrak A_U^+,\ O=O^\dagger\\\|O\|\le1}}
        \left\|e^{t\cL^\dagger}(O)-e^{t\widehat\cL^\dagger}(O)\right\|.
        \label{eq:local-observable-distance}
    \end{equation}
    This measures the largest error over even observables of norm at most one on at most $q$ modes.

    \begin{proposition}[Local stability from coefficient error]\label{prop:local-stability-coefficients}
        Let $\cL,\widehat\cL\in\mathfrak L_k$ have coefficient arrays $x=(h,G)$ and $\widehat x=(\widehat h,\widehat G)$. Suppose $\cL$ admits a local decomposition of weighted strength at most $\bar\alpha$. Then
        \begin{equation}
            d_{q,t}(\cL,\widehat\cL)
            \le 2qt\,e^{k2^k\bar\alpha t}
            \|\widehat x-x\|_{\mathrm{loc},1}.
            \label{eq:local-stability-coefficients}
        \end{equation}
        No geometric assumption or strength bound on the output decomposition is required.
    \end{proposition}

    \begin{proof}
        See Appendix~\ref{app:local-stability-proof}.
    \end{proof}

    For geometric models, use the entrywise fit in Appendix~\ref{app:sdp-implementation}, restricted to the same geometric supports. Only $\cO_{k,r,\Delta}(1)$ physical coefficients involve a given mode, so entrywise error $\eps_{\rm loc}$ implies local one-norm error $\cO_{k,r,\Delta}(\eps_{\rm loc})$. For ancillary structure learning, Corollary~\ref{cor:ancilla-physical-reconstruction} supplies a valid generator with local coefficient error at most $3\eps$. In both cases, the fitted Heisenberg evolution does not increase operator norms. The sparse coefficient estimate alone need not define a valid Lindbladian.

    \begin{corollary}[Simultaneous local-observable prediction]\label{cor:simultaneous-local-observables}
        Under the assumptions of Theorem~\ref{thm:geo-noaux}, fix $q\in\{1,\ldots,m\}$, $t>0$, and $\eps,\delta\in(0,1)$. The geometric protocol and semidefinite fit return a valid $\widehat\cL$ with $d_{q,t}(\cL,\widehat\cL)\le\eps$, with probability at least $1-\delta$, using
        \begin{equation}
            \widetilde{\cO}_{k,r,\Delta}\!\left(
                \left[1+\bar\alpha^2q^2t^2e^{2k2^k\bar\alpha t}\eps^{-2}\right]
                \log\frac m\delta
            \right)
            \label{eq:local-observable-samples}
        \end{equation}
        samples. The guarantee holds simultaneously for all even observables on at most $q$ modes. For every parity-invariant initial state, all corresponding reduced states are also predicted with trace-norm error at most $\eps$.
    \end{corollary}

    \begin{proof}
        See Appendix~\ref{app:simultaneous-prediction-proof}.
    \end{proof}

    \begin{figure}[t]
        \centering
        \includegraphics[width=\linewidth]{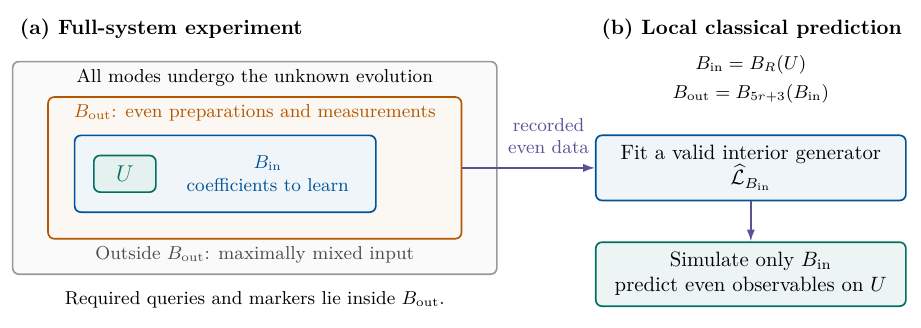}
        \caption{Local reconstruction and prediction. (a) Even preparations and measurements within $B_{\rm out}$ estimate the canonical coefficients in $B_{\rm in}$, including internal-marker calibration. Each sample evolves the full system, with maximally mixed input outside the active blocks. (b) Classical simulation uses the fitted generator $\widehat\cL_{B_{\rm in}}$ to predict even observables supported on $U$. Arrows indicate classical data processing.}
        \label{fig:locality-schematic}
    \end{figure}

    At fixed support size, prediction time, precision, and model parameters, the sample count is logarithmic in $m$. For Fermi--Hubbard models, the observables include occupations, double occupations, currents, spin observables, and even pairing correlations.

    To predict an observable on a specified region $U$, one can instead learn only a surrounding finite volume. For $B\subseteq[m]$, let $P_B$ be the diagonal projection onto the coefficient indices $A\ne\varnothing$ with $\supp_{\rm f}(A)\subseteq B$, and define
    \begin{equation}
        H_B^{\rm can}:=\sum_{\substack{A\ne\varnothing,\ \pi(A)=0\\\supp_{\rm f}(A)\subseteq B}}h_A\Gamma_A,
        \qquad G_B^{\rm can}:=P_BGP_B,
        \qquad \cL_B^{\rm can}:=\cL_{H_B^{\rm can},G_B^{\rm can}}.
        \label{eq:canonical-local-restriction}
    \end{equation}
    This canonical restriction is a valid local generator, as shown in Appendix~\ref{app:finite-volume-prediction-proof}. It is determined by the identifiable coefficients and does not require the unknown local decomposition. When comparing local and full-system evolutions, we embed local operators in the full CAR algebra and extend local generators by idle modes as in \eqref{eq:graded-extension}.

    \begin{corollary}[Local prediction by finite-volume reconstruction]\label{cor:finite-volume-local-prediction}
        Under the assumptions of Theorem~\ref{thm:geo-noaux}, fix a nonempty region $U$, $t>0$, $\delta\in(0,1)$, an integer $R\ge r$, and $0<\eps_{\rm loc}\le\bar\alpha$. Set
        \begin{equation*}
            B_{\rm in}:=B_R(U),
            \qquad B_{\rm out}:=B_{5r+3}(B_{\rm in}).
        \end{equation*}
        Restrict the geometric protocol to the transfer and calibration curves needed for the coefficients in $B_{\rm in}$, and fit a valid generator $\widehat\cL_{B_{\rm in}}$ using supports $e\subseteq B_{\rm in}$ with $|e|\le k$ and $\diam_{\mathsf G}(e)\le r$. With probability at least $1-\delta$, its coefficients approximate those of $\cL_{B_{\rm in}}^{\rm can}$ to entrywise error $\eps_{\rm loc}$ and, for a constant $C$ depending only on $k,r,\Delta$,
        \begin{equation}
            \sup_{\substack{O\in\mathfrak A_U^+,\ O=O^\dagger\\\|O\|\le1}}
            \left\|e^{t\cL^\dagger}(O)-e^{t\widehat\cL_{B_{\rm in}}^\dagger}(O)\right\|
            \le C|U|t e^{C\bar\alpha t}
            \left(\eps_{\rm loc}+\bar\alpha e^{-(R-r)}\right).
            \label{eq:finite-volume-local-prediction}
        \end{equation}
        This requires
        \begin{equation}
            N=\widetilde{\cO}_{k,r,\Delta}\!\left(
                \left[1+\bar\alpha^2\eps_{\rm loc}^{-2}\right]
                \log\frac{|B_{\rm out}|}{\delta}
            \right)
            \label{eq:finite-volume-local-samples}
        \end{equation}
        samples and $\cO_{k,r,\Delta}(|B_{\rm out}|)$ local readouts per sample. The preparations and measurements are supported in $B_{\rm out}$, with all other modes initialized in the maximally mixed state. Each sample evolves the full system for its prescribed short time $t_i\le(4k\bar\alpha)^{-1}$, independently of the prediction time $t$.
    \end{corollary}

    \begin{proof}
        See Appendix~\ref{app:finite-volume-prediction-proof}.
    \end{proof}

    At fixed $|U|$, $t$, target error, and model parameters, one can choose $R$ and $\eps_{\rm loc}$ independently of $m$. Bounded degree then makes the sample count, local readouts, and classical reconstruction cost independent of $m$. Subsequent simulation uses only $B_{\rm in}$ and can still cost exponentially in its volume. Evaluating an expectation also requires the initial state's marginal on $B_{\rm in}$, supplied or estimated separately. Its acquisition cost is not included above. Figure~\ref{fig:locality-schematic} distinguishes the full-system experiment from the subsequent local simulation.

\section{Discussion}\label{sec:discussion}
    Calibrated internal markers enable local Lindbladian learning under parity superselection without external modes in the geometric setting. For prediction in a prescribed region, only a surrounding neighborhood needs to be reconstructed. At fixed region size, prediction time, accuracy, and model parameters, the sample count and local reconstruction cost are independent of the total system size. Experiments still evolve the full system, and dense classical simulation can remain exponential in the reconstructed neighborhood's volume.

    For unknown all-to-all structure, thresholding before signed inversion controls local coefficient error through the approximate coefficient degree and the weak-coefficient tail. The additional sample factor $m^{\lfloor k/2\rfloor}$ comes from the present ancillary pairing construction. Further directions include reducing this factor through collective even measurements, obtaining degree-dependent sample bounds without a supplied degree bound, and controlling the routing costs of nonlocal pair operations. Further work could determine which Hamiltonian and dissipative coefficients can still be learned using only number-conserving preparations, controls, and measurements.

\section*{Acknowledgments}

    D.S.F. acknowledges support from the Novo Nordisk Foundation through Quantum for Life grant NNF20OC0059939 and from the ERC through GIFNEQ grant 101163938. T.M. acknowledges support from the Deutsche Forschungsgemeinschaft (DFG, German Research Foundation), Project-ID 470903074, TRR 352, and Project-ID 575156903. This project was funded within the QuantERA II program, which has received funding from the EU's H2020 research and innovation program under GA No. 101017733. DSF and CR are supported by France 2030 under the French National Research Agency award number ``ANR-22-PNCQ-0002''.

\section*{AI assistance statement}

    The authors guided the use of OpenAI Codex (GPT-5/6, Aug.--Sept. 2026) and supplied the main proof ideas, algorithmic structure, and key concepts. Under this guidance, Codex drafted supporting results and their proofs, suggested possible extensions, and assisted with proofreading. The authors checked and revised the outputs and take full responsibility for the manuscript.

\appendix
\addtocontents{toc}{\protect\setcounter{tocdepth}{1}}
\addcontentsline{toc}{section}{Appendices}
\addtocontents{toc}{\protect\setcounter{tocdepth}{0}}

\section{Physical GKSL reconstruction and numerical implementation}\label{app:sdp-implementation}\label{sec:sdp}

\subsection{Physical reconstruction and coefficient bounds}\label{app:sdp-reconstruction}

    After signed inversion, noisy coefficient estimates need not generate a completely positive trace-preserving semigroup.  Projecting local estimates separately onto the PSD cone need not preserve the estimated coefficients, because one entry $G_{A,B}$ may receive contributions from several local blocks.  Following the qubit construction in \cite{MobusEtAl2026}, we fit embedded local PSD decompositions, with parity-resolved blocks and fermionic support constraints.

    Given estimates $\widetilde h_A$ and $\widetilde G_{A,B}$, introduce real Hamiltonian variables $\widehat h_A$, with $\widehat h_A=0$ for odd $A$, and local PSD matrix variables $\widehat K^{(e,p)}$ with rows and columns indexed by $\cI_{e,p}$, as in Definition~\ref{def:localcone},
    \begin{equation}
        \widehat K^{(e,p)}\succeq0,
        \qquad |e|\le k,\quad p\in\{0,1\},
    \end{equation}
    and define the global Kossakowski matrix by
    \begin{equation}
        \widehat G_{A,B}
        :=
        \sum_{\substack{e\supseteq \fsupp(A,B)\\ |e|\le k}}
        \widehat K^{(e,\pi(A))}_{A,B}.
        \label{eq:Ghat-aggregate}
    \end{equation}
    A convenient entrywise projection is
    \begin{align}
        \text{minimize}\quad & t \nonumber\\
        \text{subject to}\quad
        &\widehat K^{(e,p)}\succeq0, \nonumber\\
        &\abs{\widehat h_A-\RePart\widetilde h_A}\le t, \nonumber\\
        &\abs{\RePart(\widehat G_{A,B}-\widetilde G_{A,B})}\le t, \nonumber\\
        &\abs{\ImPart(\widehat G_{A,B}-\widetilde G_{A,B})}\le t,
        \label{eq:sdp-projection}
    \end{align}
    for all nonempty Hamiltonian and Kossakowski indices satisfying the parity and locality conditions of Definition~\ref{def:localcone}. The absolute-value constraints are linear after splitting signs.  Standard semidefinite programming references include \cite{VandenbergheBoyd1996SDP,BoydVandenberghe2004}. For fixed $k$, there are $\sum_{j=1}^k\binom mj=\cO_k(m^k)$ candidate supports, each with a bounded number of Majorana indices and PSD blocks of dimension at most $4^k$. The program therefore has $\cO_k(m^k)$ scalar variables and constraints.

    \begin{proposition}[SDP accuracy and computational complexity]\label{prop:fermionic-sdp-runtime}
        Suppose $0<\eps_{\rm SDP}\le1$, $\|\widetilde G-G\|_\infty\le\eps_\chi$, and $\|\RePart\widetilde h-h\|_\infty\le\eps_\chi$.  Every feasible point of \eqref{eq:sdp-projection} with objective $t\le\eps_\chi+\eps_{\rm SDP}$ satisfies
        \begin{equation}
            \|\widehat G-G\|_\infty
            \le
            3\eps_\chi+2\eps_{\rm SDP},
            \qquad
            \|\widehat h-h\|_\infty
            \le
            2\eps_\chi+\eps_{\rm SDP},
            \label{eq:sdp-explicit-error}
        \end{equation}
        The program has $\cO_k(m^k)$ scalar affine constraints and PSD blocks, each PSD block having dimension at most $4^k$. The method in Subsection~\ref{app:sdp-numerics} runs in time polynomial in the input size and $\log(1/\eps_{\rm SDP})$ and returns an exactly feasible point with positive-semidefinite Kossakowski blocks.
    \end{proposition}

    \begin{proof}
        The true local PSD decomposition is feasible with objective $t\le\eps_\chi$.  By the assumed objective bound, each real and imaginary part of $\widehat G_{A,B}-G_{A,B}$ has absolute value at most $2\eps_\chi+\eps_{\rm SDP}$, so its absolute value is at most $\sqrt2(2\eps_\chi+\eps_{\rm SDP})\le3\eps_\chi+2\eps_{\rm SDP}$.  The real Hamiltonian variables obey the sharper bound stated in \eqref{eq:sdp-explicit-error}.
        The size, runtime, and exact-feasibility claims are proven in Subsection~\ref{app:sdp-numerics}.
    \end{proof}

    Define $\widehat H:=\sum_A\widehat h_A\Gamma_A$.  The output generator is
    \begin{equation}
        \widehat\cL:=\cL_{\widehat H,\widehat G}.
        \label{eq:output-generator}
    \end{equation}
    The local PSD constraints ensure $\widehat\cL\in\mathfrak L_k$, so it is parity covariant and generates a completely positive trace-preserving semigroup.

    \begin{proposition}[Coefficient-to-diamond conversion]\label{prop:diamond}
        For two generators in the model class,
        \begin{equation}
            \norm{\cL-\cL'}_\diamond
            \le
            2\sum_A\abs{h_A-h_A'}
            +
            2\sum_{A,B}\abs{G_{A,B}-G'_{A,B}}.
            \label{eq:diamond-sum}
        \end{equation}
        Consequently, if every induced Majorana-basis coefficient differs by at most $\eps_\chi$, then
        \begin{equation}
            \norm{\cL-\cL'}_\diamond
            \le
            C_km^k\eps_\chi.
            \label{eq:diamond-entry}
        \end{equation}
    \end{proposition}

    \begin{proof}
        Left and right multiplication by a Majorana unitary have diamond norm one.  A commutator basis map therefore has diamond norm at most two, and each normalized dissipator basis map has diamond norm at most two.  The number of physical coefficient pairs with mode-support union at most $k$ is $\cO_k(m^k)$.
    \end{proof}

\subsection{Certified numerical implementation}\label{app:sdp-numerics}

    We justify the runtime bound and explain how to obtain exactly positive-semidefinite output from approximate numerical blocks. By Lemma~\ref{lem:local-coefficient-norm} and \eqref{eq:weighted-strength}, the true coefficients admit a decomposition satisfying
    \begin{equation*}
        \sum_A|h_A|+
        \sum_{e,p}\tr[K^{(e,p)}]
        \le
        c_km\bar\alpha.
    \end{equation*}
    This bound, together with the accuracy of the coefficient estimates, provides finite trace bounds for the PSD blocks and finite bounds for the scalar optimization variables. We choose these bounds so that the true decomposition remains feasible and the point constructed below lies strictly inside them. This may exclude other minimizers of \eqref{eq:sdp-projection}.

    \begin{proof}[Computational claims in Proposition~\ref{prop:fermionic-sdp-runtime}]
        First, use rational input data and parameterize the SDP by real Hamiltonian coefficients and Hermitian Kossakowski blocks on the allowed indices, setting all coordinates forbidden by parity, locality, or the identity-free gauge to zero. We next specify a convenient interior point for the method. Choose a small positive rational multiple of the identity in each nonempty PSD block, set all Hamiltonian coefficients to zero, and take $t$ strictly larger than every resulting residual. This gives a known strictly feasible point. Choose the added bounds to contain this point strictly, and choose the positive scalar multiplying each identity matrix inverse-polynomially in the number of blocks. The strict inequalities above and the finite variable bounds provide inner and outer Euclidean radii $R_{\rm in}^{\rm SDP},R_{\rm out}^{\rm SDP}$ with polynomial binary encoding length. Thus $\log(R_{\rm out}^{\rm SDP}/R_{\rm in}^{\rm SDP})$ is polynomial in the input size. The bounded rational SDP theorem \cite[Theorem~1.1]{deKlerkVallentin2016} gives runtime polynomial in the input size and $\log(1/\eps_{\rm SDP})$, as in the qubit reconstruction of \cite{MobusEtAl2026}.

        Next, we bound the cost of restoring exact feasibility. Let $K_{\rm num}^{(e,p)}$ be the numerical Hermitian blocks and $G_{\rm num}$ their sum as in \eqref{eq:Ghat-aggregate}. Require objective value at most $\eps_{\rm SDP}/2$ above the bounded optimum, scalar-inequality violations at most $\xi$, and the rigorously verified bounds $K_{\rm num}^{(e,p)}\succeq-\xi\Id$. The shifts
        \begin{equation*}
            \widehat K^{(e,p)}:=K_{\rm num}^{(e,p)}+\xi\Id
        \end{equation*}
        make every block PSD. For fixed $k$, these lower bounds can be checked exactly on the rational shifted blocks, whose dimensions are bounded. Rebuild $\widehat G$ using \eqref{eq:Ghat-aggregate}. Only its diagonal entries change, and
        \begin{equation*}
            \sum_{A,B}|\widehat G_{A,B}-(G_{\rm num})_{A,B}|
            =\xi\sum_{e,p}|\cI_{e,p}|
            \le C_km^k\xi.
        \end{equation*}
        Leave the Hamiltonian coordinates unchanged and increase $t$ by $\xi$ plus this total coefficient change. Choosing $0<\xi\le\eps_{\rm SDP}/(4C_km^k)$, with $C_k\ge1$, makes this increase at most $\eps_{\rm SDP}/2$. The output is exactly feasible for \eqref{eq:sdp-projection}. Since the true decomposition has objective at most $\eps_\chi$, the final value satisfies $t\le\eps_\chi+\eps_{\rm SDP}$. The additional precision is $\cO_k(\log m+\log(1/\eps_{\rm SDP}))$ bits. For related rigorous postprocessing of approximate SDP solutions, see \cite{JanssonChaykinKeil2007}.
    \end{proof}

    The same argument applies to the local-norm fit \eqref{eq:local-norm-fit}. For each coefficient index $v$ with support $S_v$, introduce a residual variable $u_v\ge|\widehat x_v-\widetilde x_v|$ and impose $\sum_{v:j\in S_v}u_v\le t$ for every mode $j$, minimizing $t$. Here $\widetilde x$ contains the input estimates and $\widehat x=(\widehat h,\widehat G)$ the fitted coefficients. For a complex residual $z$, the constraint $u\ge|z|$ is equivalent to the $2\times2$ Hermitian PSD constraint $\bigl(\begin{smallmatrix}u&z\\\overline z&u\end{smallmatrix}\bigr)\succeq0$. Real residuals need only linear inequalities.

    Choose finite bounds that keep the true decomposition feasible and admit a strictly feasible point as above, with every $u_v$ larger than its residual and $t$ larger than each mode sum. Apply the eigenvalue tolerance also to the residual PSD blocks. For the numerical residual $z_v$ at coordinate $v$, this gives $u_v+\xi\ge|z_v|$. Increase each $u_v$ by $\xi$ plus the change in its fitted coefficient, then increase $t$ to cover these changes and the mode-sum feasibility error. Since there are $\cO_k(m^k)$ variables and blocks of bounded dimension, these adjustments increase $t$ by at most $C_km^k\xi$. The same choice $\xi\le\eps_{\rm SDP}/(4C_km^k)$ keeps this below $\eps_{\rm SDP}/2$. Together with objective error at most $\eps_{\rm SDP}/2$, this gives an exactly feasible fit satisfying
    \begin{equation*}
        \|\widehat x-\widetilde x\|_{\mathrm{loc},1}
        \le \|x-\widetilde x\|_{\mathrm{loc},1}+\eps_{\rm SDP}.
    \end{equation*}
    The triangle inequality gives the coefficient error used in Corollary~\ref{cor:ancilla-physical-reconstruction}. The runtime has the same polynomial dependence as above.

\section{Proofs for local-observable prediction}\label{app:local-prediction-proofs}

\subsection{Local stability}\label{app:local-stability-proof}

    Write $\|\cdot\|_{\rm cb}$ for the completely bounded operator norm, the supremum of the induced operator norm after adjoining an idle finite-dimensional system. For a map $\mathcal T$ on $j$ modes, adjoining $j$ ancillary modes and using $\|X\|_2\le\|X\|_1\le2^j\|X\|_2$ gives
    \begin{equation*}
        \|\mathcal T^\dagger\|_{\rm cb}=\|\mathcal T\|_\diamond
        \le2^j\|\mathcal T^\dagger\|_{2\to2}.
    \end{equation*}
    For parity-covariant maps, Lemma~\ref{lem:graded-diamond-amplification} preserves this bound under fermionic extension. For the fixed decomposition of $\cL$, define the local term strengths and their mode-pair sums by
    \begin{equation*}
        w_e:=\|\cL_e^\dagger\|_{\rm cb},
        \qquad
        \mathsf J_{u,v}:=\sum_{e\ni u,v}w_e,
        \qquad
        a:=k2^k\bar\alpha.
    \end{equation*}
    The matrix $\mathsf J$ is symmetric and nonnegative. The strength bound \eqref{eq:weighted-strength} and the power series for the matrix exponential give
    \begin{equation}
        \sum_u\mathsf J_{u,v}=\sum_{e\ni v}|e|w_e\le a,
        \qquad
        \sum_u(e^{s\mathsf J})_{u,v}\le e^{as}
        \quad(s\ge0).
        \label{eq:car-influence-column-sums}
    \end{equation}

    \begin{proof}[Proof of Proposition~\ref{prop:local-stability-coefficients}]
        We adapt the support-chain argument of \cite[arXiv version, Section~V]{BarthelKliesch2012} to parity-covariant fermionic maps and then apply Duhamel's formula to the coefficient error.

        We first bound the influence of a map supported on $S$ on an observable initially supported on $U$. Define
        \begin{equation*}
            \mathcal C_U(S,t):=\sup_{O,\mathcal K_S}
            \frac{\|\mathcal K_S^\dagger e^{t\cL^\dagger}(O)\|}
            {\|\mathcal K_S^\dagger\|_{\rm cb}\|O\|},
        \end{equation*}
        where the supremum is over nonzero $O\in\mathfrak A_U^+$ and nonzero parity-covariant maps $\mathcal K_S$ with Majorana multipliers supported on $S$ and $\mathcal K_S^\dagger(\Id)=0$. Disjointly supported parity-covariant maps commute: exchanging their basis terms $\Gamma_A(\cdot)\Gamma_B$ and $\Gamma_C(\cdot)\Gamma_D$ produces $(-1)^{\pi(A)\pi(C)+\pi(B)\pi(D)}=1$. Also, $\mathcal K_S^\dagger(O)=O\mathcal K_S^\dagger(\Id)=0$ when $S\cap U=\varnothing$, because $O$ is even.

        Split off the terms disjoint from $S$. Their Heisenberg semigroup is a unital completely positive contraction and commutes with $\mathcal K_S^\dagger$. Duhamel's formula therefore gives
        \begin{equation*}
            \mathcal C_U(S,t)
            \le\boldsymbol 1_{\{S\cap U\ne\varnothing\}}
            +\sum_{e:\,e\cap S\ne\varnothing}
            w_e\int_0^t\mathcal C_U(e,s)\,ds.
        \end{equation*}
        Iteration produces chains of intersecting supports. Summing over a mode at each intersection bounds these chains by paths of $\mathsf J$, giving
        \begin{equation}
            \mathcal C_U(S,t)
            \le\sum_{u\in S}\sum_{v\in U}
            (e^{t\mathsf J})_{u,v}.
            \label{eq:car-response-matrix-path-app}
        \end{equation}

        We now apply this bound to the coefficient error. Put $\eta:=\|\widehat x-x\|_{\mathrm{loc},1}$ and group $\widehat\cL-\cL=\sum_S\mathcal E_S$ by exact coefficient support. Each adjoint Hamiltonian or dissipative basis map is parity covariant, annihilates the identity, and has completely bounded norm at most two. Hence
        \begin{equation*}
            \max_u\sum_{S\ni u}\|\mathcal E_S^\dagger\|_{\rm cb}\le2\eta.
        \end{equation*}
        Apply Duhamel's formula with $e^{(t-s)\widehat\cL^\dagger}$ on the outside. This is a contraction because $\widehat\cL$ is valid. Equation~\eqref{eq:car-response-matrix-path-app}, followed by the column-sum bound \eqref{eq:car-influence-column-sums}, gives
        \begin{equation*}
            \|(e^{t\widehat\cL^\dagger}-e^{t\cL^\dagger})(O)\|
            \le2\eta\|O\|\int_0^t
            \sum_{v\in U}\sum_u(e^{s\mathsf J})_{u,v}\,ds
            \le2|U|t e^{at}\eta\|O\|.
        \end{equation*}
    \end{proof}

\subsection{Simultaneous prediction}\label{app:simultaneous-prediction-proof}

    \begin{proof}[Proof of Corollary~\ref{cor:simultaneous-local-observables}]
        Bounded geometry allows at most $C_{k,r,\Delta}$ coefficients involving any mode, so entrywise error $\eps_{\rm loc}$ gives local coefficient error at most $C_{k,r,\Delta}\eps_{\rm loc}$. With $a$ as above, choose
        \begin{equation*}
            \eps_{\rm loc}
            =\min\left\{\bar\alpha,
            \frac{\eps}{C_{k,r,\Delta}qt e^{at}}\right\}.
        \end{equation*}
        Apply Theorem~\ref{thm:geo-noaux} and the geometric semidefinite fit of Appendix~\ref{app:sdp-implementation} with raw coefficient error $c\eps_{\rm loc}$ and tolerance $\eps_{\rm SDP}=\min\{1,c\eps_{\rm loc}\}$, for sufficiently small $c>0$ depending only on $k,r,\Delta$. Proposition~\ref{prop:local-stability-coefficients} gives the prediction error, and substitution into \eqref{eq:geo-noaux-entrywise} gives \eqref{eq:local-observable-samples}.

        The same coefficient event controls every observable in \eqref{eq:local-observable-distance}, so no additional union bound is needed. For a parity-invariant initial state, both evolved marginals are even. Their trace-norm difference is therefore attained by an even observable of norm at most one, which gives the marginal error.
    \end{proof}

\subsection{Finite-volume prediction}\label{app:finite-volume-prediction-proof}

    \begin{proof}[Proof of Corollary~\ref{cor:finite-volume-local-prediction}]
        First, the restriction \eqref{eq:canonical-local-restriction} is a valid local generator. For each local term, keep only its Hamiltonian monomials supported entirely in $B$ and restrict each Kossakowski block to indices corresponding to monomials supported in $e\cap B$. These blocks remain PSD and sum to the canonical restriction, preserving the support-size and diameter bounds. Lemma~\ref{lem:local-coefficient-norm} bounds the restricted weighted strength by $C_k\bar\alpha$ and gives, for $\mathcal D_e:=\cL_e-(\cL_e)_B^{\rm can}$,
        \begin{equation*}
            \mathcal D_e=0\quad(e\subseteq B),
            \qquad
            \|\mathcal D_e^\dagger\|_{\rm cb}
            \le C_k\|\cL_e^\dagger\|_{2\to2}.
        \end{equation*}
        For the boundary error, take $B=B_R(U)$ and $\mu_0:=1+\log(\max\{1,\Delta-1\})$. Finite range implies
        \begin{equation*}
            \max_u\sum_v e^{\mu_0\dist_{\mathsf G}(u,v)}\mathsf J_{u,v}
            \le k2^ke^{\mu_0r}\bar\alpha.
        \end{equation*}
        The matrix entries in \eqref{eq:car-response-matrix-path-app} consequently satisfy $(e^{t\mathsf J})_{u,v}\le\exp(k2^ke^{\mu_0r}\bar\alpha t-\mu_0\dist_{\mathsf G}(u,v))$. Every contributing $e\not\subseteq B$ has all its modes at distance at least $R-r$ from $U$. Summing the graph-distance shells, which contain at most $\Delta\max\{1,\Delta-1\}^{\ell-1}$ modes at distance $\ell\ge1$ from a fixed mode, gives
        \begin{equation*}
            \sum_{u:\,\dist_{\mathsf G}(u,U)\ge R-r}
            \sum_{v\in U}e^{-\mu_0\dist_{\mathsf G}(u,v)}
            \le C_\Delta|U|e^{-(R-r)}.
        \end{equation*}
        Duhamel's formula, with the canonical restricted semigroup on the outside, now bounds the boundary error by $C_{k,r,\Delta}|U|\bar\alpha t\exp(C_{k,r,\Delta}\bar\alpha t-(R-r))\|O\|$.

        Finally, all data needed for the local fit can be acquired within $B_{\rm out}$. Every required diagonal extension contains the coefficient's nonempty anchor support and lies in $B_r(B_{\rm in})$. Its marker triplets lie in $B_{4r+3}(B_r(B_{\rm in}))=B_{\rm out}$. These queries use the full-system evolution at the prescribed sampling times, with modes outside the active blocks initialized in the maximally mixed state. The proof of Theorem~\ref{thm:geo-noaux}, restricted to this query list, gives \eqref{eq:finite-volume-local-samples} and the readout bound.

        The canonical restriction is feasible for the local geometric fit in Appendix~\ref{app:sdp-implementation}. Entrywise coefficient accuracy $c\eps_{\rm loc}$ before the semidefinite fit and tolerance $\eps_{\rm SDP}=\min\{1,c\eps_{\rm loc}\}$ give fitted entrywise error at most $\eps_{\rm loc}$ for sufficiently small $c=c_{k,r,\Delta}>0$, hence local coefficient error at most $C_{k,r,\Delta}\eps_{\rm loc}$. Apply Proposition~\ref{prop:local-stability-coefficients} with the restricted strength bound $C_k\bar\alpha$, then add the boundary error to obtain \eqref{eq:finite-volume-local-prediction}.
    \end{proof}

\section{Bounds and proofs for exponentially decaying interactions}\label{app:exponential-proofs}

    We prove the extension in Subsection~\ref{sec:exponentially-decaying}, keeping the cutoff and marker scales explicit. Throughout, use the decomposition and tail bound of Definition~\ref{def:exponential-witness}. For the supplied mode graph, define its maximal ball volume by
    \begin{equation*}
        \mathsf V(R):=\max_{u\in[m]}|B_R(\{u\})|.
    \end{equation*}

\subsection{Marker calibration and truncated inversion}

\begin{assumption}[Scale-dependent marker access]\label{ass:scale-dependent-markers}
        Fix integer scales $R,L\ge1$.  Every nonempty mode set $S\subseteq[m]$ with $|S|\le k$ and $\diam_{\mathsf G}(S)\le R$ has a known triplet of singleton marker modes $X_i(S)$ satisfying
        \begin{equation*}
            \dist_{\mathsf G}(S,X_i)\ge L,
            \qquad
            \dist_{\mathsf G}(X_i,X_j)\ge L
            \quad(i\neq j).
        \end{equation*}
        The extended blocks $W(S):=S\cup X_1(S)\cup X_2(S)\cup X_3(S)$ admit a known coloring with $K_{\rm col}(R,L)$ colors such that blocks assigned the same color share no modes. The experiment can implement the pair-block frame of Proposition~\ref{prop:pair-block-local-frame} on every $W(S)$.
    \end{assumption}

    Each extended block contains at most $k+3$ modes. Proposition~\ref{prop:pair-block-local-frame} therefore provides estimators $Z_q$ with $\mathbb E[Z_q^2]\le\kappa_k$, where $\kappa_k$ depends only on $k$. Geometry affects the number of colors and the distances between paired modes. For nonempty mode sets $U,V$, define the strength of terms meeting both sets by
    \begin{equation*}
        \tau(U,V)
        :=
        \sum_{\substack{e:\ e\cap U\neq\varnothing\\e\cap V\neq\varnothing}}
        \ell_e .
    \end{equation*}
    Every term in this sum has diameter at least $\dist_{\mathsf G}(U,V)$, so
    \begin{equation}
        \tau(U,V)
        \le
        |U|\,a(\dist_{\mathsf G}(U,V))
        \le
        |U|\bar\alpha
        e^{-\mu\dist_{\mathsf G}(U,V)}.
        \label{eq:cross-interaction-tail}
    \end{equation}

    Use the marker scalars $b_i$ and pair derivatives $s_{ij}$ from \eqref{eq:marker-b} and \eqref{eq:marker-pair-derivative}. The internal derivatives $\Lambda^{S|i}_{C,D}$ are defined in \eqref{eq:internal-data}, while $\Lambda_{C,D}$ denotes the corresponding derivative with an idle external marker.

\begin{lemma}[Approximate internal-marker correction]\label{lem:approximate-internal-marker}
        Let $S$ be a nonempty region with $|S|\le k$, let $X_1,X_2,X_3\subseteq S^c$ be distinct singleton marker modes, and let $C,D\subseteq\cM(S)$. For $i\ne j$,
        \begin{equation*}
            |s_{ij}-b_i-b_j|
            \le
            3\tau(X_i,X_j).
        \end{equation*}
        For each $i$,
        \begin{equation*}
            \left|
                \Lambda^{S|i}_{C,D}
                -
                \Lambda_{C,D}
                -
                \boldsymbol 1_{\{\pi(C)=\pi(D)=1\}}
                \delta_{C,D}b_i
            \right|
            \le
            3\tau(S,X_i).
        \end{equation*}
        Under Assumption~\ref{ass:scale-dependent-markers}, if every internal and calibration derivative is estimated to error at most $\eps_L$, the calibrated derivative obeys
        \begin{equation}
            |\widehat\Lambda^{\rm corr}_{C,D}-\Lambda_{C,D}|
            \le
            C_k\left(
                \eps_L+\bar\alpha e^{-\mu L}
            \right).
            \label{eq:approximate-marker-error}
        \end{equation}
    \end{lemma}

    \begin{proof}
        Split the fixed decomposition as in Proposition~\ref{prop:pair-marker} and Theorem~\ref{thm:internal-correction}. Terms meeting at most one of the two relevant regions satisfy the exact identities. For a term meeting both, each normalized overlap has absolute value at most $\ell_e$, including the twisted contributions because the twist in \eqref{eq:twisted-generator} preserves the Hilbert--Schmidt induced norm. Each defect contains at most three such contributions, giving both inequalities. The triplet formulas \eqref{eq:three-marker-calibration} amplify statistical error by $3/2$, so the corrected derivative has statistical error at most $5\eps_L/2$. Equation~\eqref{eq:cross-interaction-tail} and the marker distances bound the remaining defects by $C_k\bar\alpha e^{-\mu L}$.
    \end{proof}

    Let $n_{k,R}$ denote the number of allowed coordinates of $(h,G)$ with nonempty support of size at most $k$ and diameter at most $R$. Choosing an anchor mode gives
    \begin{equation}
        n_{k,R}
        \le
        C_km\mathsf V(R)^{k-1}.
        \label{eq:spatial-coordinate-count}
    \end{equation}
    The recursion's error amplification is bounded by
    \begin{equation}
        B_{k,R}
        :=
        C_k\sum_{t=0}^{k-1}
        \binom{\mathsf V(R)}{t}3^t
        \le
        C_k'\mathsf V(R)^{k-1}.
        \label{eq:spatial-inversion-amplification}
    \end{equation}
    Here $C_k\ge1$ is chosen sufficiently large to cover the fixed-dimensional transforms and the number of ordered extension chains, as shown below.

\begin{lemma}[Spatially truncated inversion]\label{lem:spatially-truncated-inversion}
        Suppose every required derivative is estimated to error at most $\eps_L$. Apply the recursion \eqref{eq:descending-recursion} only to coefficients and diagonal extensions whose supports have diameter at most $R$, and set all other physical coordinates to zero. Under Definition~\ref{def:exponential-witness} and Assumption~\ref{ass:scale-dependent-markers}, the resulting arrays satisfy
        \begin{equation}
            \|\widehat h-h\|_\infty
            ,\qquad
            \|\widehat G-G\|_\infty
            \le
            B_{k,R}\left(
                \eps_L
                +\bar\alpha e^{-\mu L}
                +\bar\alpha e^{-\mu R}
            \right).
            \label{eq:spatially-truncated-inversion}
        \end{equation}
        The corrected-derivative query list and the classical recursion have size $\cO_k(n_{k,R})$.
    \end{lemma}

    \begin{proof}
        Fix a nonempty anchor support and a mode in it. Every extension by $t$ modes whose support has diameter at most $R$ lies in this mode's radius-$R$ ball, giving at most
        \begin{equation*}
            \binom{\mathsf V(R)}{t}3^t
        \end{equation*}
        possible Majorana-label choices. Every omitted extension has support diameter greater than $R$. Group its coefficient sum by local terms. Each contributing term contains the anchor mode and has diameter greater than $R$. Lemma~\ref{lem:local-coefficient-norm} and \eqref{eq:exponential-local-tail} therefore bound the sum of absolute values of the omitted coefficients by $C_k\bar\alpha e^{-\mu R}$. The same argument bounds every physical coefficient whose support diameter exceeds $R$. Lemma~\ref{lem:approximate-internal-marker}, the restricted transform \eqref{eq:restricted-chi}, and this tail bound give the recursion of Lemma~\ref{lem:triangular-error-propagation}, with source error
        \begin{equation*}
            C_k\left(
                \eps_L+\bar\alpha e^{-\mu L}
                +\bar\alpha e^{-\mu R}
            \right).
        \end{equation*}
        Each endpoint is reached by only $\cO_k(1)$ ordered extension chains. Thus \eqref{eq:spatial-inversion-amplification} and Lemma~\ref{lem:triangular-error-propagation} prove \eqref{eq:spatially-truncated-inversion}. The recursion contains $\cO_k(n_{k,R})$ nodes. Each restricted Fierz transform has fixed dimension, and each coefficient contributes to at most $2^k$ lower-support indices, giving $\cO_k(n_{k,R})$ queries and recursion operations.
    \end{proof}

\subsection{Learning on general mode graphs}

\begin{theorem}[Learning exponentially decaying fermionic Lindbladians]\label{thm:exponentially-decaying-learning}
        There exists a constant $c_{k,0}>0$, depending only on $k$, such that the following holds. Let $\cL$ have the exponentially decaying $k$-body decomposition of Definition~\ref{def:exponential-witness}, suppose Assumption~\ref{ass:scale-dependent-markers} holds at scales $R,L$, and fix $\delta\in(0,1)$. For $0<\eps_\chi\le\bar\alpha$, if
        \begin{equation}
            B_{k,R}\bar\alpha
            \left(e^{-\mu R}+e^{-\mu L}\right)
            \le
            c_{k,0}\eps_\chi,
            \label{eq:exponential-entrywise-condition}
        \end{equation}
        then all coefficients can be recovered to entrywise error $\eps_\chi$ with probability at least $1-\delta$, using
        \begin{equation}
            N_{\rm coeff}
            =
            \widetilde{\cO}_k\!\left(
                K_{\rm col}(R,L)
                \bar\alpha^2 B_{k,R}^2\eps_\chi^{-2}
                \log\frac{n_{k,R}}{\delta}
            \right)
            \label{eq:exponential-entrywise-samples}
        \end{equation}
        samples. For diamond recovery, let $0<\eps_\diamond\le m\bar\alpha$ and suppose instead that
        \begin{equation}
            \bar\alpha
            \left(e^{-\mu R}+e^{-\mu L}\right)
            \le
            c_{k,0}\frac{\eps_\diamond}
            {m\mathsf V(R)^{2(k-1)}}
            \label{eq:exponential-diamond-condition}
        \end{equation}
        holds. Solve \eqref{eq:sdp-projection} with every candidate block $e$ satisfying $1\le|e|\le k$ and $\diam_{\mathsf G}(e)\le R$, to tolerance $0<\eps_{\rm SDP}\le\min\{1,c_{k,0}\eps_\diamond/n_{k,R}\}$. With probability at least $1-\delta$, it outputs a valid generator $\widehat\cL$ with interaction range at most $R$ satisfying
        \begin{equation*}
            \|\widehat\cL-\cL\|_\diamond
            \le
            \eps_\diamond
        \end{equation*}
        using
        \begin{equation}
            N_\diamond
            =
            \widetilde{\cO}_k\!\left(
                K_{\rm col}(R,L)
                \bar\alpha^2m^2\mathsf V(R)^{4(k-1)}\eps_\diamond^{-2}
                \log\frac{n_{k,R}}{\delta}
            \right)
            \label{eq:exponential-diamond-samples}
        \end{equation}
        samples. Both protocols are nonadaptive and use only even preparations and measurements on system modes, with evolution times at most $(4k\bar\alpha)^{-1}$ and no control during evolution. For either displayed sample count $N$, evolution time is bounded by \eqref{eq:total-evolution-time}.  Streaming costs $\widetilde{\cO}_k(n_{k,R}N)$ operations, while calibration and inversion cost $\cO_k(n_{k,R})$. The SDP restricted to supports of diameter at most $R$ has $\cO_k(n_{k,R})$ constant-dimensional PSD blocks and, for rational data, can be solved to tolerance $\eps_{\rm SDP}$ in time polynomial in its input size and $\log(\mathsf V(R)/\eps_{\rm SDP})$.
    \end{theorem}

    \begin{proof}
        Choose $c_{k,0}$ sufficiently small to satisfy the error bounds below and one color uniformly and independently for each sample. Run the frames on its disjoint blocks, multiply each selected estimator by $K_{\rm col}(R,L)$, and set it to zero otherwise. Proposition~\ref{prop:pair-block-local-frame} and Lemma~\ref{lem:parallel-block-marginalization} give unbiased estimators of the curve values, each with second moment at most $K_{\rm col}(R,L)\kappa_k$. Median-of-means estimation \cite[Theorem~4.1]{DevroyeEtAl2016} and Proposition~\ref{prop:fermionic-derivative-poly}, with a union bound over the $\cO_k(n_{k,R})$ curves and sampling times, estimate all derivatives to accuracy
        \begin{equation*}
            \eps_L
            =
            \Theta_k\!\left(
                \frac{\eps_\chi}{B_{k,R}}
            \right).
        \end{equation*}
        Choose the implied constant sufficiently small. Lemma~\ref{lem:spatially-truncated-inversion} and \eqref{eq:exponential-entrywise-condition} then give the entrywise guarantee and \eqref{eq:exponential-entrywise-samples}. For the physical output, let
        \begin{equation*}
            \cL_{\le R}
            :=
            \sum_{\diam_{\mathsf G}(e)\le R}\cL_e.
        \end{equation*}
        This is a valid generator with interaction range at most $R$ and a feasible point of the candidate-block SDP.  Moreover,
        \begin{equation*}
            \sum_{\diam(e)>R}\ell_e
            \le
            \sum_{u=1}^m
            \sum_{\substack{e\ni u\\\diam(e)>R}}
            \ell_e
            \le
            ma(R).
        \end{equation*}
        Hence Lemma~\ref{lem:local-coefficient-norm} and Proposition~\ref{prop:diamond} give
        \begin{equation*}
            \|\cL-\cL_{\le R}\|_\diamond
            \le
            C_k\sum_{\diam(e)>R}\ell_e
            \le
            C_km\bar\alpha e^{-\mu R}.
        \end{equation*}

        Let $x_v$ be any physical coordinate whose support has diameter at most $R$ and choose a mode in its nonempty union support. Every discarded local term contributing to $x_v$ contains this mode and has diameter greater than $R$. Thus
        \begin{equation}
            |x_v(\cL)-x_v(\cL_{\le R})|
            \le
            C_ka(R)
            \le
            C_k\bar\alpha e^{-\mu R}.
            \label{eq:truncated-comparator-coordinate-bias-app}
        \end{equation}
        Take
        \begin{equation*}
            \eps_L
            =
            \Theta_k\!\left(
                \frac{\eps_\diamond}
                {m\mathsf V(R)^{2(k-1)}}
            \right).
        \end{equation*}
        Equations~\eqref{eq:spatially-truncated-inversion}, \eqref{eq:exponential-diamond-condition}, and \eqref{eq:truncated-comparator-coordinate-bias-app} bound the error relative to each allowed coefficient of $\cL_{\le R}$ by $\cO_k(\eps_\diamond/n_{k,R})$. This gives a feasible comparison point with the same objective bound. Proposition~\ref{prop:fermionic-sdp-runtime} controls the fitted coefficient errors, and Proposition~\ref{prop:diamond} converts their sum into diamond error. Choosing the constants sufficiently small makes the fitting and truncation errors sum to at most $\eps_\diamond$. Substituting the derivative accuracy gives \eqref{eq:exponential-diamond-samples}.

        To restore exact numerical feasibility, use the procedure in Subsection~\ref{app:sdp-numerics}, with objective error at most $\eps_{\rm SDP}/2$, scalar-constraint violations at most $\xi$, and rigorously verified eigenvalue bounds $K_{\rm num}^{(e,p)}\succeq-\xi\Id$. Each coordinate in the fit belongs to at most $C_k\mathsf V(R)^{k-1}$ candidate blocks. The positivity shifts and scalar adjustments therefore increase the objective by at most $(1+C_k\mathsf V(R)^{k-1})\xi$. Choose $C_k\ge1$ and
        \begin{equation*}
            \xi
            \le
            \frac{\eps_{\rm SDP}}
            {4C_k\mathsf V(R)^{k-1}}.
        \end{equation*}
        These adjustments cost at most $\eps_{\rm SDP}/2$ in objective value and $\cO_k(\log\mathsf V(R))$ additional precision bits. The fitted generator need not satisfy the original exponential-tail bound.
    \end{proof}

\subsection{Periodic lattices and access requirements}

\begin{proof}[Proof of Corollary~\ref{cor:exponential-lattice}]
        Fix a cutoff $R\in\{R_\chi,R_\diamond\}$ and set $L=R$. For each nonempty support $S$ with $|S|\le k$ and $\diam_{\mathsf G}(S)\le R$, choose a deterministic anchor $x(S)\in S$ and place the markers at
        \begin{equation*}
            X_i(S):=\{x(S)+(R+i(R+1))e_1\},
            \qquad i=1,2,3,
        \end{equation*}
        where $e_1$ is the first coordinate vector. The finite-size condition makes each displacement from the anchor geodesic. Since $\diam(S)\le R$, one has $\dist(S,X_i)\ge i(R+1)>R$, and the markers are mutually separated by at least $R+1$. Thus Assumption~\ref{ass:scale-dependent-markers} holds. The extended blocks have diameter at most $5R+3$, so their pair operations span distances $\cO(R)$.

        A fixed vertex belongs to at most $C_k\mathsf V(R)^{k-1}$ target supports. If it occurs as a marker, its position and marker label determine the anchor, giving the same bound. Each extended block has at most $k+3$ modes. Greedy coloring therefore gives
        \begin{equation*}
            K_{\rm col}(R,R)\le C_k\mathsf V(R)^{k-1},
            \qquad \mathsf V(R)\le(2R+1)^D.
        \end{equation*}
        Put $d_*:=D(k-1)$. The cutoff conditions of Theorem~\ref{thm:exponentially-decaying-learning} follow from
        \begin{equation*}
            C_{k,D}\bar\alpha(1+R_\chi)^{d_*}e^{-\mu R_\chi}\le\eps_\chi,
            \qquad
            C_{k,D}\bar\alpha(1+R_\diamond)^{2d_*}e^{-\mu R_\diamond}\le\frac{\eps_\diamond}{m}.
        \end{equation*}
        The logarithmic radii in the corollary satisfy these inequalities when $C(k,D,\mu)$ is sufficiently large, since the exponential decay dominates either fixed polynomial. Substituting the volume and coloring bounds into the theorem gives the explicit spatial overheads
        \begin{align*}
            N_{\rm coeff}
            &=\widetilde{\cO}_{k,D,\mu}\!\left(
                \bar\alpha^2\eps_\chi^{-2}(1+R_\chi)^{3d_*}
                \log\frac m\delta
            \right),\\
            N_\diamond
            &=\widetilde{\cO}_{k,D,\mu}\!\left(
                \bar\alpha^2m^2\eps_\diamond^{-2}(1+R_\diamond)^{5d_*}
                \log\frac m\delta
            \right).
        \end{align*}
        Absorbing these polynomial factors in the logarithmic radii proves \eqref{eq:exponential-lattice-entrywise} and \eqref{eq:exponential-lattice-diamond}.
    \end{proof}

\begin{remark}[Routing and scope]\label{rem:exponential-routing}
        Strictly nearest-neighbor control additionally requires a routing construction. If each block $W(S)$ can be implemented inside $B_{cR}(x(S))$ for a fixed $c\ge1$, with the anchors chosen above, coloring these routing regions gives
        \begin{equation*}
            K_{\rm col}^{\rm route}(R)
            \le
            C_k\mathsf V(2cR)\mathsf V(R)^{k-1}.
        \end{equation*}
        Indeed, intersecting regions have anchors within distance $2cR$, and each anchor labels at most $C_k\mathsf V(R)^{k-1}$ blocks. This bound assumes such a routing construction. The corollary's finite-size condition also requires the marker separation to remain below the system scale.
    \end{remark}

\section{Ancillary measurement construction}\label{sec:modewise-ancilla}

    Here we present the even preparation-and-measurement scheme used in Theorem~\ref{thm:ancilla-structure} and Corollary~\ref{cor:alltoall-modewise}. It uses one idle ancillary mode per system mode and pair operations, without Bell-pair preparation. Ancillary parity operators make the inputs compatible with a fixed pairing, while randomized output pairings provide common data for all required transfer curves.

\subsection{Ancillary parity operators}
    Add an ancillary layer of $m$ idle fermionic modes, ordered after the system layer. Write $j_{\rm sys}$ and $j_{\rm aux}$ for mode $j$ in the system and ancillary layers, respectively, and choose one Majorana $\xi_j$ on $j_{\rm aux}$. The combined Fock space has dimension $d^2$. For $A\subseteq[2m]$, recall from \eqref{eq:local-majorana-parity} that $q_j(A)=|A\cap\{2j-1,2j\}|\bmod2$ is one exactly when $\Gamma_A$ contains a single Majorana on system mode $j_{\rm sys}$. Define
    \begin{align}
        r_A
        &:=
        \sum_{j=1}^m q_j(A),&
        \Xi_A
        &:=
        i^{r_A(r_A-1)/2}
        \prod_{\substack{j\in[m]\\q_j(A)=1}}\xi_j .
        \label{eq:modewise-tag}
    \end{align}
    The product follows the ancillary-mode order.  Thus $\Xi_A$ is a Hermitian unitary of parity $\pi(\Xi_A)=r_A\bmod2=\pi(A)$. For $\pi(C)=\pi(D)=p$, set
    \begin{equation}
        Q_D:=i^p\Gamma_D\Xi_D,
        \qquad
        P_{C|D}:=i^p\Gamma_C\Xi_D .
        \label{eq:modewise-input-output}
    \end{equation}
    Both are even Hermitian unitaries. Let $\Phi_t^{[m]}$ be the fermionic extension of $\Phi_t=e^{t\cL}$ to these idle modes, using the same system multipliers as in \eqref{eq:graded-extension}. No Hamiltonian or jump operator involves an ancillary mode.

    \begin{lemma}[Ancillary parity operators realize the dressed transfer curves]\label{lem:modewise-dressed-curves}
        For every $C,D$ with $\pi(C)=\pi(D)$,
        \begin{equation}
            d^{-2}
            \tr\!\left[
                P_{C|D}\,
                \Phi_t^{[m]}(Q_D)
            \right]
            =
            M^\eta_{C,D}(t).
            \label{eq:modewise-curve-identity}
        \end{equation}
        If $\pi(C)\neq\pi(D)$, then $M^\eta_{C,D}(t)=0$ by parity covariance.
    \end{lemma}

    \begin{proof}
        Apply the calculation of Lemma~\ref{lem:signed-transfer} to $\Phi_t$, replacing $\eta^p$ by $\Xi_D$. The operator $\Xi_D$ has parity $p$, so commuting it through the system multipliers gives the same sign $(-1)^{p\pi(B)}$. Since $\Xi_D^2=\Id$, the ancillary factor contributes $\tr[\Id]=d$, leaving the normalized system trace in \eqref{eq:signed-transfer}.
    \end{proof}

\subsection{Explicit input frame and structured output pairing}

    Pair system mode $j_{\rm sys}$ with ancillary mode $j_{\rm aux}$ at the input.  Every $Q_D$ is compatible with this vertical pairing because the system factor $\Gamma_D$ and the ancillary parity operator $\Xi_D$ have the same local parity $q_j(D)$ on each vertical block.  In each sample, use on these pairs the random even-axis preparations $(\Id+S_bR_b)/4$ from Proposition~\ref{prop:pair-block-process-shadow} and evolve only the system, keeping the ancillary modes idle.  Both $Q_D$ and $P_{C|D}$ use the ancillary factor $\Xi_D$, with the signs from the fermionic extension accounted for by Lemma~\ref{lem:modewise-dressed-curves}.

    A fixed vertical output pairing detects $P_{C|D}$ only when $q_j(C)=q_j(D)$ for every $j$. Otherwise it is Hilbert--Schmidt orthogonal to the block-even measurement algebra.  Replacing the ancillary parity operator $\Xi_D$ by $\Xi_C$ does not help, because $\tr[\Xi_C\Xi_D]=0$ unless the local-parity signatures agree.  A query-adapted pairing suffices for one pair, whereas common data require target-independent randomized output pairings.  The resulting $m^{\lfloor k/2\rfloor}$ factor is a pairing-coverage cost within this architecture, not a lower bound for collective measurements.

    Use the following target-independent output distribution. For $m\ge4k$, set $m_0=2\lfloor m/4\rfloor$ and choose a uniformly random set $T\subseteq[m]$ of size $m_0$. Keep the vertical output pair $\{j_{\rm sys},j_{\rm aux}\}$ for $j\notin T$, and choose a uniform perfect matching of $T$. For every matched index pair $\{j,\ell\}$, independently choose with equal probability either
    \begin{equation*}
        \{j_{\rm sys},\ell_{\rm sys}\},\ \{j_{\rm aux},\ell_{\rm aux}\}
        \qquad\text{or}\qquad
        \{j_{\rm sys},\ell_{\rm aux}\},\ \{j_{\rm aux},\ell_{\rm sys}\}.
    \end{equation*}
    For $m<4k$, instead choose a uniform perfect matching of all $2m$ modes. In either case, measure a uniformly random even axis on every output block.

    Let $\cQ_k:=\{(C,D):|\fsupp(C,D)|\le k\}$ be the query family, with $|\cQ_k|=\cO_k(m^k)$.

    \begin{proposition}[Modewise-ancilla process shadow]\label{prop:modewise-process-shadow}
        For every $\eps_M>0$ and $\delta\in(0,1)$, suppose the experiment can append the $m$ idle modes and implement the two-mode even preparations and measurements just described.  At every selected time $t$, one common dataset estimates all $M^\eta_{C,D}(t)$, $(C,D)\in\cQ_k$, to uniform accuracy $\eps_M$ and failure probability at most $\delta$ using
        \begin{equation}
            \widetilde{\cO}_k\!\left(
                m^{\lfloor k/2\rfloor}
                \eps_M^{-2}
                \log\frac{2m}{\delta}
            \right)
            \label{eq:modewise-shadow-cost}
        \end{equation}
        samples. Its unbiased single-sample estimators have second moments at most $\cO_k(m^{\lfloor k/2\rfloor})$.
    \end{proposition}

    \begin{proof}
        Queries with $\pi(C)\neq\pi(D)$ are zero deterministically. Suppose first that $m\ge4k$. For a remaining query define
        \begin{equation*}
            \mathcal B_{C,D}
            :=
            \{j:q_j(C)=q_j(D)=1\},
            \qquad
            \mathcal S_{C,D}
            :=
            \{j:q_j(C)\neq q_j(D)\}.
        \end{equation*}
        Write $b=|\mathcal B_{C,D}|$ and $s=|\mathcal S_{C,D}|$.  Equality of the total parities makes $s$ even, and
        \begin{equation*}
            b+s\le|\fsupp(C,D)|\le k.
        \end{equation*}
        Consider the event that $\mathcal S_{C,D}\subseteq T$, that $T\cap\mathcal B_{C,D}=\varnothing$, that the random matching of $T$ pairs the $s$ mismatch indices among themselves, and that each resulting pair receives the unique orientation joining its two odd layer modes.  On this event $P_{C|D}$ is compatible.  For even $n\ge2$, write $(n-1)!!=(n-1)(n-3)\cdots1$, with $(-1)!!=1$. The event has probability
        \begin{equation}
            p^\circ_{C,D}
            =
            \frac{\binom{m-b-s}{m_0-s}}{\binom m{m_0}}
            \frac{(s-1)!!(m_0-s-1)!!}{(m_0-1)!!}
            2^{-s/2}
            \ge
            C_k^{-1}m^{-s/2}
            \ge
            C_k^{-1}m^{-\lfloor k/2\rfloor}.
            \label{eq:structured-pairing-hit}
        \end{equation}
        Here $C_k\ge1$ depends only on $k$. The first factor is bounded below by a positive constant depending only on $k$, since $m_0=\Theta(m)$, and the second is $\Theta_k(m^{-s/2})$.

        Apply Lemma~\ref{lem:randomized-pairing-estimators} with this event as the acceptance set. Only the classical acceptance rule depends on $(C,D)$, while the experimental settings serve all queries. The estimator is unbiased for the left-hand side of \eqref{eq:modewise-curve-identity}. The input $Q_D$ is active on at most $k$ vertical blocks. The output $P_{C|D}$ has at most two active layer modes per mode in $\fsupp(C,D)$, hence at most $2k$ active output blocks. The estimator's second moment is therefore at most
        \begin{equation*}
            (p^\circ_{C,D})^{-1}7^{3k}
            \le
            C_k m^{\lfloor k/2\rfloor}.
        \end{equation*}
        For $m<4k$, every even output monomial is compatible with some perfect matching of the $2m$ modes. There are only $\cO_k(1)$ such matchings, so accepting all compatible matchings gives the same second-moment bound. Lemma~\ref{lem:modewise-dressed-curves} and the simultaneous estimation guarantee of Lemma~\ref{lem:randomized-pairing-estimators} now prove \eqref{eq:modewise-shadow-cost}.
    \end{proof}

\section{Occupation-parity sample lower bounds}\label{sec:lower-bounds}

    For a mode $u$ and a nonempty mode set $S$, define
    \begin{equation}
        Z_u:=-i\gamma_{2u-1}\gamma_{2u}=\Id-2n_u,
        \qquad
        A_S:=\bigcup_{u\in S}\{2u-1,2u\},
        \qquad
        Z_S:=\prod_{u\in S}Z_u=(-1)^{|S|}\Gamma_{A_S}.
        \label{eq:occupation-parity-packing}
    \end{equation}
    These are commuting even, number-conserving observables. Under the Jordan--Wigner mapping, $Z_u$ becomes the single-qubit Pauli operator $\sigma_u^z$, so $Z_S$ becomes $\prod_{u\in S}\sigma_u^z$ with the same support \cite{SzalayEtAl2021}. We adapt the qubit testing arguments of \cite[Section~7]{MobusEtAl2026} to arbitrary fresh system--reference inputs and joint measurements, and use single-mode terms to prove the matching geometric diamond bound. These lower bounds also apply under parity superselection, which only restricts the allowed experiments.

    In sample $j$, the learner may prepare an arbitrary fresh system--reference state, evolve the system for a chosen time $0\le t_j\le t_{\max}$, and make a joint measurement. Later choices may depend on the preceding measurement outcomes. For a protocol using at most $N$ samples, define its worst-case quadratic exposure and total evolution time by
    \begin{equation*}
        Q_{\rm evol}:=\sup_{\rm transcripts}\sum_jt_j^2,
        \qquad
        T_{\rm evol}^{\max}:=\sup_{\rm transcripts}\sum_jt_j .
    \end{equation*}
    Then $Q_{\rm evol}\le Nt_{\max}^2$ and $Q_{\rm evol}\le t_{\max}T_{\rm evol}^{\max}$. For every Hermitian $A$,
    \begin{equation}
        \|-i[A,\cdot]\|_\diamond
        =
        \lambda_{\max}(A)-\lambda_{\min}(A),
        \label{eq:commutator-spectral-diameter}
    \end{equation}
    because centering $A$ gives the upper bound and the coherence between its extreme eigenvectors attains it. Write $\operatorname{TV}$ for total variation distance and $\operatorname{BC}(P,Q):=\sum_y\sqrt{P(y)Q(y)}$ for the Bhattacharyya coefficient, replacing the sum by an integral for continuous outcomes.

    \begin{lemma}[Adaptive exposure bound]\label{lem:adaptive-exposure}
        Let $\theta\in\mathbb R\setminus\{0\}$ and let $H_0,W$ be Hermitian, with $W^2=\Id$ and $[H_0,W]=0$. For $H_\pm=H_0\pm\theta W$, let $P_\pm$ be the transcript distributions of any protocol described above. If $Q_{\rm evol}<1/(64\theta^2)$, then
        \begin{equation*}
            \operatorname{TV}(P_+,P_-)
            \le2|\theta|\sqrt{Q_{\rm evol}}.
        \end{equation*}
    \end{lemma}

    \begin{proof}
        Condition on the same preceding outcomes under both hypotheses. The relative unitary is $e^{2i\theta tW}$, whose eigenvalues lie at $e^{\pm2i\theta t}$. Every purified input therefore has output overlap of absolute value at least $\cos(2|\theta|t)$, and measurement cannot decrease this root fidelity. The exposure bound ensures $2|\theta|t_j<1/4$.

        Let $\boldsymbol y$ denote a complete transcript, adding zero-time samples after an early stop. Iterating the conditional fidelity bound backwards, then using $\cos x\ge1-x^2/2$ and $\prod_j(1-a_j)\ge1-\sum_j a_j$ for $a_j\in[0,1]$, gives
        \begin{equation*}
            \operatorname{BC}(P_+,P_-)
            \ge
            \inf_{\boldsymbol y}\prod_j
            \cos\!\bigl(2|\theta|t_j(\boldsymbol y_{<j})\bigr)
            \ge1-2\theta^2Q_{\rm evol}.
        \end{equation*}
        The claim follows from the classical fidelity bound
        \begin{equation*}
            \operatorname{TV}(P_+,P_-)
            \le
            \sqrt{2\bigl(1-\operatorname{BC}(P_+,P_-)\bigr)}.
        \end{equation*}
    \end{proof}

    \begin{proposition}[Coefficient lower bound]\label{prop:fermionic-coefficient-lower}
        Fix a nonempty $S\subseteq[m]$ with $|S|\le k$, and let $0<\eps_1\le\bar\alpha/4$. Uniformly estimating $h_{A_S}=\theta$ to error $\eps_1$, with success probability at least $2/3$, for the family
        \begin{equation*}
            \cL_\theta=-i[\theta\Gamma_{A_S},\cdot],
            \qquad |\theta|\le\bar\alpha/2,
        \end{equation*}
        requires $Q_{\rm evol}=\Omega(\eps_1^{-2})$. Consequently,
        \begin{equation}
            N=\Omega\!\left(\frac1{t_{\max}^2\eps_1^2}\right),
            \qquad
            T_{\rm evol}^{\max}=\Omega\!\left(\frac1{t_{\max}\eps_1^2}\right).
            \label{eq:single-fermionic-lower}
        \end{equation}
        The family has weighted interaction strength at most $\bar\alpha$. Choosing $S$ to be a singleton places it in the geometric class.
    \end{proposition}

    \begin{proof}
        Each generator is a single local term of norm $2|\theta|\le\bar\alpha$. The two parameters $\theta=\pm2\eps_1$ belong to the stated interval, and accurate estimation distinguishes them with success probability at least $2/3$. Hence their transcript laws have total variation at least $1/3$. If $Q_{\rm evol}<1/(256\eps_1^2)$, Lemma~\ref{lem:adaptive-exposure} instead gives
        \begin{equation*}
            \operatorname{TV}(P_+,P_-)
            \le4\eps_1\sqrt{Q_{\rm evol}}<\frac14,
        \end{equation*}
        a contradiction. The sample and total-time bounds follow from the two bounds on $Q_{\rm evol}$ above.
    \end{proof}

    \begin{proof}[Proof of Corollary~\ref{cor:geometric-diamond-lower}]
        We use a family of independent single-mode Hamiltonian coefficients. Set
        \begin{equation*}
            \zeta:=\frac{8\eps_\diamond}{m},
            \qquad
            H_v:=\zeta\sum_{u=1}^m v_uZ_u,
            \qquad v\in\{\pm1\}^m,
        \end{equation*}
        and let $\cL_v=-i[H_v,\cdot]$. Its singleton terms $\cL_{u,v}=-i[\zeta v_uZ_u,\cdot]$ have diameter zero, and the assumed accuracy range gives
        \begin{equation*}
            \alpha(\{\cL_{u,v}\})=2\zeta\le\frac{\bar\alpha}{4}.
        \end{equation*}
        Since the $Z_u$ have independently selectable eigenvalues, \eqref{eq:commutator-spectral-diameter} gives, with $d_H$ denoting Hamming distance,
        \begin{equation}
            \|\cL_v-\cL_w\|_\diamond
            =4\zeta d_H(v,w).
            \label{eq:singleton-exact-diamond-distance}
        \end{equation}

        Accurate generator recovery would identify most signs. Draw $V$ uniformly from $\{\pm1\}^m$ and choose the family member $\cL_{\widehat V}$ closest to the estimated generator in diamond norm. On the success event, the triangle inequality and \eqref{eq:singleton-exact-diamond-distance} give $d_H(V,\widehat V)\le m/16$. Therefore
        \begin{equation*}
            \mathbb E d_H(V,\widehat V)
            \le\frac23\frac m{16}+\frac13m=\frac{3m}{8}.
        \end{equation*}
        For each coordinate $u$, let $P_{u,\pm}$ be the transcript law averaged over the other coordinates, conditional on $V_u=\pm1$. Binary testing bounds the error probability for that coordinate below by $(1-\operatorname{TV}(P_{u,+},P_{u,-}))/2$. Summing over $u$ yields
        \begin{equation*}
            \mathbb E d_H(V,\widehat V)
            \ge\frac m2\left(1-\frac1m\sum_{u=1}^m
            \operatorname{TV}(P_{u,+},P_{u,-})\right).
        \end{equation*}
        Thus the average variation is at least $1/4$. Neighboring Hamiltonians have the form $H_0\pm\zeta Z_u$ with $[H_0,Z_u]=0$. If $Q_{\rm evol}<1/(64\zeta^2)$, Lemma~\ref{lem:adaptive-exposure} bounds every neighboring transcript variation strictly below $1/4$. Convexity gives the same bound for the coordinate mixtures, a contradiction. Hence
        \begin{equation}
            Q_{\rm evol}\ge\frac1{64\zeta^2}
            =\frac{m^2}{4096\eps_\diamond^2}.
            \label{eq:geometric-quadratic-exposure-lower}
        \end{equation}
        Finally, $Q_{\rm evol}\le Nt_{\max}^2$ and $Q_{\rm evol}\le t_{\max}T_{\rm evol}^{\max}$ give the claimed sample and total-time bounds.
    \end{proof}

\section{Technical access lemmas}\label{app:technical}

\subsection{Even preparations and measurements from Gaussian control}

    Let $S\subseteq[m]$ be a set of system modes and let $x_{\rm mark}\notin S$ be a marker mode with Majoranas $\eta,\eta'$. A fermionic Gaussian unitary on $S\cup\{x_{\rm mark}\}$ acts on Majoranas by an orientation-preserving real orthogonal transformation and preserves total parity \cite{JozsaMiyake2008,WanEtAl2023}. Local fermionic Gaussian access means that the experiment can implement these unitaries, prepare occupation states with chosen classical probabilities, and measure the occupation parities $-i\gamma_{2j-1}\gamma_{2j}$ and $-i\eta\eta'$.

    \begin{proposition}[Dressed preparations and measurements from Gaussian control]\label{prop:even-clifford-eigenaccess}
        For every $(C,D)$ with $|\fsupp(C,D)|\le k$, local fermionic Gaussian access on $\fsupp(C,D)\cup\{x_{\rm mark}\}$ allows preparation of $\rho_D^\pm$ in \eqref{eq:physical-inputs} for $D\neq\varnothing$ and measurement of $J_\eta(\Gamma_C)$, with all other modes initialized in the maximally mixed state. For $D=\varnothing$, use the maximally mixed input.
    \end{proposition}

    \begin{proof}
        Let $P=J_\eta(\Gamma_D)$ with $D\subseteq\cM(S)$.  If $D=\varnothing$, the required input is the maximally mixed state.  Otherwise $P$ is an even Hermitian Majorana monomial, so $P=\pm i^r\widetilde\gamma_1\cdots\widetilde\gamma_{2r}$ for distinct Majoranas in $S\cup\{x_{\rm mark}\}$.  An orientation-preserving signed permutation of the Majoranas, implemented by a Gaussian unitary, maps $P$ to a known sign times $\prod_{q=1}^r(-i\gamma_{2q-1}\gamma_{2q})$, after relabeling local modes and possibly using the marker pair.  Multiplying these parity outcomes and the known sign measures $P$.

        Uniformly sampling occupation strings whose parity product, multiplied by the known sign, equals the desired eigenvalue prepares the normalized eigenspace projector. Applying the inverse Gaussian unitary and initializing the remaining modes in the maximally mixed state gives $\rho_D^\pm$. The same rotation-and-measurement construction applies to $J_\eta(\Gamma_C)$.
    \end{proof}

\subsection{Derivative recovery from even CAR support growth}

    The dressed and calibration curves above are instances of the following derivative-recovery result.

    \begin{lemma}[Even CAR support growth]\label{lem:even-car-support-growth}
        Let $\cL_e$ be a local term from Definition~\ref{def:localcone}, and let $\widetilde\cL_e$ denote its fermionic extension, possibly with no added modes. Let $X$ be an even observable supported on a mode set $U$, which may include idle modes. Then $\widetilde\cL_e^\dagger(X)=0$ if $e\cap U=\varnothing$. Otherwise, $\widetilde\cL_e^\dagger(X)$ is even, supported on $U\cup e$, and
        \begin{equation*}
            \norm{\widetilde\cL_e^\dagger(X)}_2
            \le
            \norm{\cL_e^\dagger}_{2\to2}\norm{X}_2 .
        \end{equation*}
    \end{lemma}

    \begin{proof}
        Choose an even Hamiltonian and parity-homogeneous jumps supported on $e$. If $e\cap U=\varnothing$, the even operator $X$ commutes with all these operators, so the Hamiltonian commutator and each dissipative term $L^\dagger XL-\frac12\{L^\dagger L,X\}$ vanish. Otherwise, multiplication enlarges the support only to $U\cup e$. The output remains even because $\widetilde\cL_e^\dagger$ is parity covariant. The norm bound follows from Lemma~\ref{lem:graded-diamond-amplification} and invariance of the induced Hilbert--Schmidt norm under adjoints.
    \end{proof}

    \begin{proposition}[Robust derivative recovery from even CAR curves]\label{prop:fermionic-derivative-poly}
        Fix $k$, $\bar\alpha>0$, and a corruption fraction $\beta\in[0,1/2)$, with $\beta=0$ for uncorrupted data. Suppose $\alpha(\{\cL_e\})\le\bar\alpha$. Consider a finite family of curves
        \begin{equation*}
            f_{X,Y}(t)
            :=
            D^{-1}\tr\!\left[Y e^{t\cL}(X)\right]
        \end{equation*}
        where $D$ is the Hilbert-space dimension, including any idle modes, $\cL$ is a parity-covariant $k$-local generator or its fermionic extension, and $X,Y$ are even Hermitian Majorana monomials supported on $\cO_k(1)$ modes. This includes all dressed-transfer and marker-pair curves used above. Set $T=(4k\bar\alpha)^{-1}$. For every $0<\eps_L\le\bar\alpha$, one can choose a degree
        \begin{equation*}
            \ell=\cO_{k,\beta}\!\left(1+\log\!\left(1+\frac{1}{T\eps_L}\right)\right)
        \end{equation*}
        and a Chebyshev partition with endpoints $T(1+\cos(\pi j/J))/2$, $0\le j\le J$, where $J=\lceil c_\beta\ell\rceil$ for a sufficiently large constant $c_\beta$. Choose $\lceil\log(\ell+1)\rceil$ fixed interior times in each interval, common to all curves. This gives $\cO_{k,\beta}(\ell\log(\ell+1))$ times. For each curve, assume that at most a $\beta$ fraction of its estimates in each interval are arbitrarily corrupted, and all other estimates have error at most
        \begin{equation}
            \eps_M
            \le
            \frac{T\eps_L}{4C_\beta\ell^2},
            \label{eq:fermionic-curve-precision}
        \end{equation}
        where $C_\beta\ge1$ is a regression constant and $C_0$ denotes its value at $\beta=0$. Robust polynomial regression returns $\widehat p_{X,Y}$ such that $\widehat f_{X,Y}'(0):=\widehat p_{X,Y}'(0)$ satisfies
        \begin{equation*}
            \left|
            \widehat f_{X,Y}'(0)-f_{X,Y}'(0)
            \right|
            \le
            \eps_L
        \end{equation*}
        simultaneously for all curves whenever these data conditions hold.
    \end{proposition}

    \begin{proof}
        We first approximate each curve by a polynomial. Write $f_{X,Y}(t)=D^{-1}\tr[e^{t\cL^\dagger}(Y)X]$. By Lemma~\ref{lem:even-car-support-growth}, every nonzero sequence of local terms applied to $Y$ remains even, adds at most $k-1$ modes per step, and obeys the Hilbert--Schmidt norm bound used in \cite[Lemmas~4.1--4.2]{MobusEtAl2026}. Their support-growth and Taylor-tail arguments therefore apply, including with idle modes. Set $a:=\max\{1,\max_{(X,Y)}\lceil|\supp_{\rm f}(Y)|/k\rceil\}=\cO_k(1)$. Since $\|X\|_2\|Y\|_2/D=1$, Cauchy--Schwarz gives, for the degree-$\ell$ Taylor polynomial $p_{X,Y}^{(\ell)}$ with $\ell\ge a$,
        \begin{equation*}
            (p_{X,Y}^{(\ell)})'(0)=f_{X,Y}'(0),
            \qquad
            \|f_{X,Y}-p_{X,Y}^{(\ell)}\|_{\infty,[0,T]}
            \le 2^{a-\ell-1}=: \tau_\ell.
        \end{equation*}
        If $Y=\Id$, the curve is constant by trace preservation, so these conclusions hold as well.

        Next choose $\ell$ so that $\tau_\ell\le T\eps_L/(4C_\beta\ell^2)$. Exponential decay gives the stated degree. On the common grid, uncorrupted data differ from $p_{X,Y}^{(\ell)}$ by at most $\eps_M+\tau_\ell$. Robust Chebyshev regression \cite[Theorem~1.4]{KaneKarmalkarPrice2017}, as formulated in \cite[Lemma~4.4]{MobusEtAl2026}, therefore gives $\|\widehat p_{X,Y}-p_{X,Y}^{(\ell)}\|_{\infty,[0,T]}\le C_\beta(\eps_M+\tau_\ell)$ under the stated corruption bound.

        Finally, Markov brothers' inequality converts this uniform polynomial error into derivative error:
        \begin{equation*}
            |\widehat p_{X,Y}'(0)-f_{X,Y}'(0)|
            \le\frac{2C_\beta\ell^2}{T}(\eps_M+\tau_\ell)
            \le\eps_L.
        \end{equation*}
        This deterministic argument holds simultaneously for every curve on the stated curve-estimation event.
    \end{proof}

\bibliographystyle{abbrvurl-doi}
\bibliography{refs_v5}

\end{document}